\documentclass[11pt]{article}
\usepackage{amsmath,amssymb,graphicx,float,epsf,tikz}
\def\RR{\hbox{I\kern-.2em\hbox{R}}}
\newcommand{\qed}{\hbox to 0pt{}\hfill$\rlap{$\sqcap$}\sqcup$ \vspace{3mm}}
\numberwithin{equation}{section}
\restylefloat{figure}
\usepackage{subfigure}
\usepackage{graphicx} %%For loading graphic files
\usepackage{amsmath}
\usepackage{amsthm}
\usepackage{amsfonts}
\usepackage{subfigure}
\usepackage{wrapfig}
\usepackage{enumerate}
\usepackage{enumitem}
\usepackage{array}
\usepackage{caption}
\usepackage[margin=1cm]{caption}
\usepackage{authblk}

\usepackage[utf8]{inputenc}

\newtheorem{theorem}{Theorem}
\newtheorem{lemma}[theorem]{Lemma}

\usepackage{graphicx}
\usepackage{tikz}
\usetikzlibrary{backgrounds, shadows}
\usetikzlibrary{arrows, positioning, shapes.geometric}
\usepackage{graphicx}
\allowdisplaybreaks
\usetikzlibrary{mindmap,backgrounds,arrows.meta}
\usetikzlibrary{shapes.geometric, arrows}
\tikzstyle{rect} = [draw, rectangle, fill=blue!20, text width=6em, text centered, minimum height=2em]
\tikzstyle{elli} = [draw, ellipse, fill=red!20, minimum height=2em]
\tikzstyle{circ} = [draw, circle, fill=white!20, minimum width=8pt, inner sep=5pt]
\tikzstyle{diam} = [draw, diamond, fill=white!20, text width=6em, text badly centered, inner sep=0pt]
\tikzstyle{line} = [draw, -latex']
\usepackage{a4wide} %%Smaller margins = more text per page.
\usepackage{authblk} % for multiple authors affiliation

\usepackage{varioref}
\usepackage{xcolor}
\usepackage{hyperref}
\hypersetup{
	colorlinks=true,
	linkcolor=violet,
	citecolor=blue,
	urlcolor=teal
}
\makeatletter
\renewcommand{\@cite}[2]{\textcolor{red}{[}\textcolor{blue}{#1}\textcolor{red}{]}}
\makeatother
\usepackage{booktabs}
\date{}
\graphicspath{ {Images/} }

\begin{document}
	
%	\title{HIV/AIDS Models and Implementing Control Strategies to Obtain 90-90-90 and 95-95-95 Targets}
\title{HIV/AIDS Suppression in North America: Intervention Plans  and Cost-Effectiveness of UNAIDS 90-90-90 and 95-95-95 Targets} 
%Control Strategies/Intervention Plans/Mitigation Strategies/Optimization Methods

	\author[1]{\small Nuzhat Nuari Khan Rivu\thanks{Email: nuzhatnuarikhan-rivu@uiowa.edu }}
	\author[2]{\small Md Kamrujjaman\thanks{Corresponding author Email: kamrujjaman@du.ac.bd}}
	\author[3]{\small  Asif Iqbal\thanks{Email: aiqbal3@huskers.unl.edu }}

	\affil[1]{\footnotesize Department of Mathematics, University of Iowa, Iowa City, IA 52242, USA}
	\affil[2]{\footnotesize Department of Mathematics, University of Dhaka, Dhaka 1000, Bangladesh}
	\affil[3]{\footnotesize Department of Economics, University of Nebraska-Lincoln, Lincoln, NE 68588, USA}%T6G  

	\maketitle
	\vspace{-1.0cm}
	\noindent\rule{6.3in}{0.02in}\\
	\noindent {\bf Abstract}\\
\noindent This study utilizes mathematical models to assess progress toward achieving the UNAIDS 90-90-90 and 95-95-95 targets aimed at managing and eradicating HIV/AIDS. It contrasts stochastic and deterministic models, focusing on their utility in optimizing public health strategies. Stochastic models account for real-world unpredictability, offering more realistic insights compared to deterministic approaches. 
The 95-95-95 targets aim for 95\% of people living with HIV to know their status, 95\% of those diagnosed to receive antiretroviral therapy (ART), and 95\% of those on ART to achieve viral suppression. These benchmarks are critical for reducing transmission and improving health outcomes. This analysis establishes the basic reproduction number ($R_0$) to guide interventions and examines the stability of disease-free and endemic equilibria, providing a foundation for applying optimal control strategies to minimize HIV prevalence effectively and cost-efficiently. Moreover, the data for this study was sourced from the official UNAIDS website, focusing on North America.
An innovative feature of this study is the application of the Stochastic method, which enhances model accuracy and operational efficiency in simulating HIV transmission under various interventions. This research offers actionable insights for policymakers and contributes to global efforts to achieve the 95-95-95 targets by 2030, advancing the fight against HIV/AIDS.

\noindent{\it \footnotesize Keywords}: {\small Stochastic Modeling; HIV; UNAIDS 90-90-90 and 95-95-95; Cost Effectiveness.}\\
	\noindent
\noindent{\it \footnotesize AMS Subject Classification 2020}: 92-10, 92C42, 92C60, 92D30, 92D45. \\
\noindent\rule{6.3in}{0.02in}

\clearpage
\section*{Highlights}
\begin{enumerate}
	\item The research directly addresses the critical UNAIDS 90-90-90 and 95-95-95 goals, offering actionable strategies to reduce HIV transmission, enhance awareness, and optimize treatment coverage.
	\item It conducts detailed numerical simulations and sensitivity analyses, shedding light on the impact of key factors such as awareness rates and ART coverage on epidemic outcomes, making it highly relevant for policymakers.
	\item It evaluates cost-effective intervention strategies using mathematical optimization techniques, such as Pontryagin’s Maximum Principle, balancing resource allocation with maximizing public health benefits.
	\item By incorporating stochastic elements to model uncertainties in transmission, treatment, and awareness, the study offers a realistic representation of HIV/AIDS dynamics, making it a valuable reference for future epidemiological research and applications.
	\item It demonstrates that by implementing all three control measures simultaneously at the calculated rates starting in 2001, UNAIDS could have achieved the 90-90-90 and 95-95-95 targets significantly earlier than the projected years of 2020 and 2030, respectively.
\end{enumerate}
%\chapter{HIV/AIDS Models and Implementing Control Strategies to Obtain 90-90-90 and 95-95-95 Targets}\label{chap:03}

\clearpage
\section{Introduction}\label{sec:introduction}
HIV/AIDS, first identified in the early 1980s, has been one of the most devastating global health challenges in modern history. Initially misunderstood and stigmatized, the virus rapidly spread across populations, claiming millions of lives and instilling fear worldwide. The introduction of antiretroviral therapy (ART) in the mid-1990s marked a transformative moment in the fight against HIV/AIDS. ART not only extended the life expectancy of individuals living with HIV but also reduced transmission rates, demonstrating that the epidemic could be controlled with effective public health measures. However, despite significant progress, HIV/AIDS continues to burden low and middle-income countries, with millions still unable to access testing, treatment, or education due to systemic healthcare inequalities.

\noindent  To combat this crisis, the Joint United Nations Programme on HIV/AIDS (UNAIDS) introduced the 90-90-90 targets in 2014, which were later expanded to 95-95-95. These goals aim to ensure that 95\% of people living with HIV know their status, 95\% of those diagnosed receive sustained antiretroviral therapy, and 95\% of those on ART achieve viral suppression. Achieving these targets is critical to halting the HIV/AIDS epidemic and realizing the Sustainable Development Goals (SDGs) related to health and well-being by 2030 \cite{payne2023trends, awaidy2023progress, lundgren2024sweden}. However, reaching these ambitious benchmarks requires an integrated and scientifically informed approach that goes beyond traditional methods.

\noindent  Mathematical modeling has become an essential tool in understanding the dynamics of HIV transmission and the effectiveness of intervention strategies. Historically, epidemiological models have played a pivotal role in informing public health policies, from smallpox eradication to understanding the spread of influenza and COVID-19 \cite{khan2023vaccine,mohammad2025bifurcation,mohammad2025stochastic,mohammad2024wiener}. For HIV/AIDS, these models help unravel the complexities of disease transmission, assess the impact of behavioral and biological factors, and optimize resource allocation in diverse settings. This study leverages both deterministic and stochastic modeling frameworks to analyze HIV/AIDS dynamics, with particular emphasis on the UNAIDS targets.

\noindent  Predictive models offer a broad understanding of disease progression by assuming fixed parameters, while stochastic models incorporate random variations, making them better suited to capture real-world uncertainties. In this research, we examine the interplay between key population segments: those who are susceptible, those living with HIV but unaware of their status, those aware of their infection and seeking treatment, and those progressing to AIDS. The model also evaluates critical factors like treatment adherence, awareness campaigns, and the transmission rates between different groups.

\noindent  We collected the data from the official website of UNAIDS \cite{unaids2023data} and worked with the data of North America from the year 2001 to 2023. We apply Pontryagin’s Maximum Principle to derive optimal control strategies for achieving the goals to end AIDS \cite{ahmed2021optimal}.The results show that the optimal control strategy is the combination of all three control measures. The control effect is closely linked to the weight. By simulating the impact of various intervention strategies, this research provides actionable insights for policymakers and healthcare practitioners to fine-tune their efforts toward achieving the 95-95-95 targets. Furthermore, the study highlights the importance of increasing awareness, expanding ART coverage, and tailoring interventions to specific population needs.

\noindent Ultimately, this research not only contributes to the theoretical advancements in epidemiological modeling but also bridges the gap between mathematical insights and real-world public health applications. By aligning scientific innovation with global health initiatives, it aspires to accelerate progress toward ending the HIV/AIDS epidemic and improving the lives of millions worldwide \cite{ahmed2020dynamics}.

\subsection*{Research Gap}
While this study employs stochastic modeling to analyze HIV/AIDS transmission and control strategies, it considers only three key factors for stochasticity. Although these factors—variability in awareness, treatment adherence, and transmission rates—provide significant insights, they do not capture the full range of complexities present in real-world scenarios. In practice, numerous additional factors influence HIV dynamics, including socioeconomic disparities, healthcare infrastructure variability, population mobility, and behavioral heterogeneity. Current models may oversimplify these intricacies, limiting their applicability to diverse settings. Incorporating a broader set of parameters and stochastic events, such as resource availability, demographic shifts, and policy changes, would offer a more comprehensive representation of practical scenarios. Furthermore, while the study uses a stochastic method to improve model accuracy, its potential to integrate these additional complexities remains unexplored. Bridging this gap is essential for developing models that better inform targeted interventions and global HIV/AIDS eradication efforts.

\noindent The study could be enhanced by integrating AI-driven approaches, such as machine learning algorithms, to improve the accuracy of HIV transmission predictions and identify patterns in large datasets. AI-based optimization models can dynamically allocate resources more efficiently, ensuring cost-effective achievement of the 95-95-95 targets. Additionally, natural language processing (NLP) \cite{sabir2024artificial, datta2023transferable} tools could optimize public health messaging, improving awareness and treatment adherence. Reinforcement learning techniques could further refine intervention strategies by adapting to real-time epidemiological and social changes. These AI innovations would address existing gaps in modeling complexity, offering more comprehensive, responsive, and targeted solutions for HIV/AIDS eradication efforts.

\subsection*{Objective of the Study}
The objective of this study is to investigate how the simultaneous application of three key strategies—screening, education, and treatment—could significantly accelerate the achievement of the UNAIDS 90-90-90 and 95-95-95 targets for eradicating HIV/AIDS. The research highlights that implementing these strategies in a coordinated manner has the potential to drastically reduce HIV transmission rates, suggesting that earlier adoption of such an integrated approach might have achieved the targets years earlier \cite{martin2016progress}. However, the study also questions the cost-efficiency of this comprehensive strategy, emphasizing the need to balance intervention costs with their effectiveness in diverse resource-constrained settings.

\noindent Section \ref{sec:introduction}, highlighting the global HIV/AIDS epidemic and the importance of mathematical modeling in achieving the UNAIDS 90-90-90 and 95-95-95 targets. In section \ref{sec:modelformulation}, a compartmental model with six categories is developed to capture the dynamics of HIV transmission, focusing on transitions between different stages of infection and treatment. In Appendix \ref{sec:mathematicalanalysis}, we examine the stability of disease-free and endemic equilibria, providing insights into the conditions for HIV persistence or eradication. Section \ref{sec:simulations} are incorporated to account for real-world uncertainties in transmission, awareness, and treatment rates. The most significant section \ref{sec:optimal}, conducts extensive sensitivity analyses and simulations to assess the impact of key parameters, such as awareness and ART, on epidemic control. It demonstrates how reducing transmission rates and increasing awareness can substantially lower infection levels and improve treatment outcomes. Finally, section \ref{sec:discussion} underscores the importance of stochastic modeling in optimizing public health strategies and calls for further refinement to enhance practical applications.

\section{Model Formulation}\label{sec:modelformulation}
Within the framework of this study, we construct a mathematical model for HIV/AIDS transmission in order to investigate the spread of HIV among individuals. Each of the individuals is separated into one of six sections, which are, susceptible ($S$), latent ($E$), aware ($I_1$), unaware ($I_2$), treated ($T$), and AIDS ($A$) compartments.

\begin{figure} [H]
	\centering
	\begin{tikzpicture}[small mindmap, outer sep=0pt, text=black]
		
		\begin{scope}[concept color=yellow!100!white!80]
			\node (s) at (-8,0) [concept, scale=0.8] {\bf \small Susceptible Class}
			[counterclockwise from=70];
		\end{scope}
		
		\begin{scope}[concept color=orange!50]
			\node (l) at (-5,0) [concept, scale=0.8] {\bf \small Latent Class}
			[clockwise from=70];
		\end{scope}
  
		\begin{scope}[concept color=red!50]
			\node (ia) at (-2.5,-2) [concept, scale=0.8] {\bf \small Infected Class (Aware)}
			[clockwise from=70];
		\end{scope}
  
           \begin{scope}[concept color=red!60]
			\node (iu) at (-2.5,2) [concept, scale=0.8] {\bf \small Infected Class (Unaware)}
			[clockwise from=70];
		\end{scope}
		
		\begin{scope}[concept color=green!70]
			\node (t) at (1,2) [concept, scale=0.8] {\bf \small Treated Class}
			[clockwise from=70];
		\end{scope}
		
		\begin{scope}[concept color=red!80]
			\node (a) at (1,-2) [concept, scale=0.8] {\bf \small AIDS Class}
			[clockwise from=70];
		\end{scope}
		
		\path (s) to[circle connection bar switch color=from (yellow!100!white!80) to (orange!50)] (l) ;
		\path (l) to[circle connection bar switch color=from (orange!50) to (red!50)] (ia) ;
            \path (l) to[circle connection bar switch color=from (orange!50) to (red!60)] (iu) ;
            \path (iu) to[circle connection bar switch color=from (red!60) to (red!50)] (ia) ;
            \path (ia) to[circle connection bar switch color=from (red!50) to (green!70)] (t) ;
		\path (ia) to[circle connection bar switch color=from (red!50) to (red!80)] (a) ;
            \path (iu) to[circle connection bar switch color=from (red!60) to (red!80)] (a) ;
            \path (t) to[circle connection bar switch color=from (green!70!) to (red!80)] (a) ;
		
		\begin{pgfonlayer}{background}    
			\draw 
			[concept connection,->,orange!50,shorten >= -0.15pt,-{Stealth[angle=70:1pt 6]}] 
			(s) to (l); 
			\draw 
			[concept connection,->,red!50,shorten >= -0.15pt,-{Stealth[angle=70:1pt 6]}] 
			(l) to (ia); 
                \draw 
			[concept connection,->,red!50,shorten >= -0.15pt,-{Stealth[angle=70:1pt 6]}] 
			(l) to (iu);
                \draw
                [concept connection,->,red!50,shorten >= -0.15pt,-{Stealth[angle=70:1pt 6]}] 
			(iu) to (ia);
                \draw
                [concept connection,->,green!70,shorten >= -0.15pt,-{Stealth[angle=70:1pt 6]}] 
			(ia) to (t);
                \draw 
			[concept connection,->,red!75,shorten >= -0.15pt,-{Stealth[angle=70:1pt 6]}] 
			(ia) to (a);
			\draw 
			[concept connection,->,red!75,shorten >= -0.15pt,-{Stealth[angle=70:1pt 6]}] 
			(iu) to (a); 
               \draw 
			[concept connection,->,red!75,shorten >= -0.15pt,-{Stealth[angle=70:1pt 6]}] 
			(t) to (a); 

		\end{pgfonlayer}
		
	\end{tikzpicture}
	\caption{Compartmental model scheme visualization.}
	\label{fig:2.1}
\end{figure}
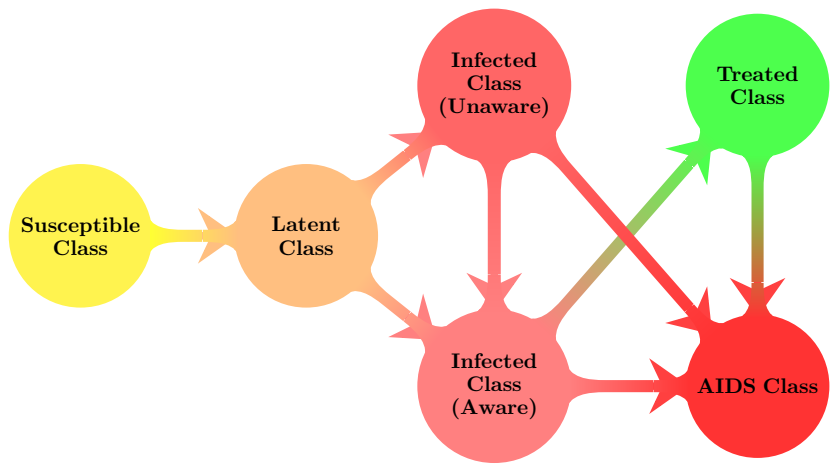
%\noindent At time $t$, the total number of individuals is denoted by $N(t)$, where $N(t)$ is the sum of $S(t)$, $E(t)$, $I_1(t)$, $I_2(t)$, $T(t)$, and $A(t)$. The initial conditions are given as $S(0), E(0), I_1(0), I_2(0), \\T(0),$ and $A(0)$.\\
Usually, any affected person can spread the disease. Individuals with infection in various compartments have varying infectivities. We make the assumption in our investigation that those who are aware do not spread the virus. Define $\lambda(t)$ as the force of infection for individuals infected with HIV, expressed as follows:
\begin{equation*}
\lambda(t) = \alpha \epsilon E(t) + \alpha I_{2}(t),   
\end{equation*}
In this equation, the term $\alpha$ refers to the effective contact rate for HIV transmission, whereas $\epsilon$ is a coefficient that represents the lower infectivity of latent persons in comparison to infected individuals. The parameter \(\epsilon\) at this point meets the condition \(0 < \epsilon \leq 1\). \(\kappa\) is the recruitment rate of persons who are vulnerable to the disease. When susceptible individuals come into touch with infectious persons at the rate \(\lambda(t)\), they have the potential to get infected. A portion of the latent individuals, denoted by $p$, enter the aware class, while the remaining latent individuals, denoted by $(1-p)$, enter the unaware infected class. The rate of symptomatic infection in latent persons is \(\beta\). Unaware infected people become aware at \(\gamma\). ART treatment rates for aware infected individuals are \(\psi\). At a rate of \(\xi\), ART-treated patients transition to latent class due to treatment failure. At the rate \(\delta_1\), aware infected people transition to AIDS class, whereas unaware ones transition at the rate \(\delta_2\). The natural mortality rate is \(\mu\) and the infectious mortality rate is \(\mu_0\). The flow diagram of the model is shown in Figure \ref{fig:2.1}.

Taking all these factors into account, we have formulated the following mathematical model:
\begin{equation}\label{Model_1}
	%\left\{
	\begin{cases}
	\begin{aligned}
		\frac{dS}{dt} &= \kappa - \alpha(\epsilon E + I_2)S - \mu S \\
		\frac{dE}{dt} &= \alpha(\epsilon E + I_2)S - (\beta + \mu)E \\
		\frac{dI_1}{dt} &= p\beta E + \gamma I_2 - (\mu + \psi + \delta_1)I_1 \\
		\frac{dI_2}{dt} &= (1 - p)\beta E - (\gamma + \mu + \delta_2)I_2 \\
		\frac{dT}{dt} &= \psi I_1 - (\mu + \xi)T \\
		\frac{dA}{dt} &= \delta_1 I_1 + \delta_2 I_2 + \xi T - \mu_0 A - \mu A
	\end{aligned}
	\end{cases}
%	\right.
%	\tag{1}
\end{equation}
where, the initial conditions, 
\begin{equation}
	S(0) \geq 0, \, E(0) \geq 0, \, I_1(0) \geq 0, \, I_2(0) \geq 0, \, T(0) \geq 0, \, A(0) \geq 0. \label{initial}    %\tag{7} 
\end{equation}
%$(S_0,E_0,I_{10},I_{20},T_0,A_0),$
and the total population is, $$N(t)=S(t)+E(t)+I_1(t)+I_2(t)+T(t)+A(t)$$ with $t\in [0,\infty)$,
%with $N_0=S_0+E_0+I_0+R_0$, and $t\in (0,\infty)$.
See the schametric diagram in Figure \ref{fig:2.1}.
\noindent The Tables \ref{tab:2.1} and \ref{tab:2.2} include a listing of the parameters and the state variables, respectively.

\begin{table}[ht]
\centering
\caption{Model (\ref{Model_1})'s state variables.}
\label{tab:2.1}
\begin{tabular}{@{}ll@{}}
\toprule
State variable & Description                           \\ \midrule
$S$               & Susceptible population           \\
$E$               & Latent population                \\
$I_1$            & Aware infected population        \\
$I_2$             & Unaware infected population      \\
$T$               & HIV-infected population who receive ART treatment \\
$A$               & AIDS population                  \\ \bottomrule
\end{tabular}
\end{table}

\begin{table}[htbp]
\centering
\caption{Explanation of parameters in model (\ref{Model_1}).}
\label{tab:2.2}
\resizebox{\textwidth}{!}{%
\begin{tabular}{@{}cllcl@{}}
\toprule
\textbf{Parameter} & \textbf{Explanation} & \textbf{Unit} & \textbf{Value} & \textbf{Source} \\ \midrule
$\kappa$ & Susceptible population recruitment rate & Year$^{-1}$ & 1130 & Estimated \\
$\alpha$ & The rate of successful transmission of HIV through direct contact & Year$^{-1}$ & $1.7649e^{-5}$ & \cite{xue2023modelling} \\
$\beta$ & The rate at which people in a latent state transition to either \\
& an unaware or aware infected state & Year$^{-1}$ & 0.3333 & Estimated \\
$\epsilon$ & Diminished infectiousness of latent persons 
& No Dimension & 0.7691 & \cite{xue2023modelling} \\
$p_0$ & The minimal proportion of individuals who are aware of being \\
& infected & No Dimension & 0.0826 & Estimated \\
$p_{\text{max}}$ & The maximum proportion of persons who are aware of being\\
& affected & No Dimension & 0.65584 & Estimated  \\
$\delta_1$ & The rate at which people who are aware of the disease enrol \\
& in AIDS class & No Dimension & 0.0157 & Assumed \\
$\delta_2$ & The rate at which people who are unaware of the disease move \\
& to AIDS class & No Dimension & 0.0157 & Assumed \\
$\psi$ & Treatment rate for aware infected individuals & No Dimension & 0.3487 & Estimated  \\
$\gamma$ & The rate at which those who were previously unaware become \\
& aware & No Dimension & 0.0157 & Estimated \\
$\mu_0$ & Mortality rate of individuals with AIDS & Year$^{-1}$ & 0.2000 & \cite{xue2023modelling} \\
$\mu$ & The rate of deaths that occur naturally & Year$^{-1}$ & 0.0143 & \cite{xue2023modelling} \\ \bottomrule
\end{tabular}
}
\end{table}

\noindent To facilitate calculation and analysis, let \( k_1 = \beta + \mu \), \( k_2 = \mu + \xi + \delta_1 \), \( k_3 = \gamma + \mu + \delta_2 \), and \( k_4 = \mu + \xi \). %, then model (\ref{Model_1}) is changed as follows:
Let us define $\mathbb{X}=(X_1,X_2,X_3,X_4,X_5,X_6)=(S,E,I_1,I_2,T,A)$. Then the mathematical model (\ref{Model_1}) is given in the compact form as below,
\begin{equation}\label{Model_2}
	\mathbb{X}'=G(X,t) , \  \mathbb{X}(0)=\mathbb{X}_0,
	%	\label{compact}
\end{equation}
where, the initial conditions, $\mathbb{X}_0=(S_0,E_0,I_{10},I_{20}, T_0,A_0),$  and,
\begin{align}
G(X,t)&=[\kappa - (\alpha \epsilon E + \alpha I_2)S - \mu S,(\alpha \epsilon E + \alpha I_2)S - k_1 E,p\beta E - k_2 I_1 + \gamma I_2,\nonumber \\ 
&(1 - p)\beta E - k_3 I_2,\psi I_1 - k_4 T, \delta_1 I_1 + \delta_2 I_2 + \xi T - \mu_0 A - \mu A].\nonumber
\end{align}

Note that the basic reproduction number, $\mathcal{R}_0$ of \eqref{Model_2} is
\begin{equation*}
	\displaystyle \mathcal{R}_0 = \rho(FV^{-1}) = \frac{\alpha \varepsilon S_0}{k_1} + \frac{\beta S_0 \alpha (1 - p)}{k_1 k_3}
\end{equation*}
The detailed procedure for determining $\mathcal{R}_0$ can be found in Appendix \ref{sec:mathematicalanalysis}.
 
 \subsection*{The Feasible Area}
 \begin{lemma}\label{Lem1}
 	Assume the initial state variables in model (\ref{Model_2}) are positive, specifically,\\ \(S(0), E(0), I_1(0), I_2(0), T(0), A(0)>0\). Then, the functions \(S(t), E(t), I_1(t), I_2(t), T(t),\) and \(A(t)\) remain non-negative for all \(t \geq 0\) in model (\ref{Model_2}), and
 	\[
 	\lim_{t \to \infty} \sup N(t) \leq \frac{\kappa}{\mu}.
 	\]
 \end{lemma}
 %\textbf{Lemma 3.1:} Assume the initial state variables in model (\ref{Model_2}) are positive, specifically, \(S(0), E(0), I_1(0), I_2(0), T(0), A(0)>0\). Then, the functions \(S(t), E(t), I_1(t), I_2(t), T(t),\) and \(A(t)\) remain non-negative for all \(t \geq 0\) in model (\ref{Model_2}), and
 %\[
 %\lim_{t \to \infty} \sup N(t) \leq \frac{\kappa}{\mu}.
 %\]
 \begin{proof}
 	Consider, \(M(t)\) as the minimum value of all state variables, i.e.,
 	\begin{align*}
 		M(t) = \min\{S(t), E(t), I_1(t), I_2(t), T(t), A(t)\}  
 	\end{align*}
 	Now, if \(t > 0\), then it is obvious that \(M(t) > 0\). We assume that there exists a minimum time \(t_1\) satisfying that \(M(t_1) = 0\). Suppose \(M(t_1) = S(t_1)\), then \(E(t)\), \(I_1(t)\), \(I_2(t)\), \(T(t)\), \(A(t) > 0\) \(\forall \) $t$ \(\in [0, t_1]\). When \(t \in [0, t_1]\), we get,
 	\begin{align*}
 		\frac{dS}{dt} &= \kappa - \lambda(t)S - \mu S \\
 		&\geq -\lambda(t)S - \mu S 	= -(\alpha \varepsilon E + \alpha I_2)S - \mu S.
 	\end{align*}
 	After integrating both sides of the equation from 0 to \(t\), the resulting expression is as follows:
 	\begin{align*}
 		S(t) &\geq S(0)e^{-\int_{0}^{t}(\alpha \varepsilon E(\tau)+\alpha I_2(\tau)+\mu)d\tau} > 0, \quad \forall \: t \in [0, t_1],
 	\end{align*}
 	which contradicts with \(M(t_1) = S(t_1) = 0\). Therefore, \(S(t)\), \(E(t)\), \(I_1(t)\), \(I_2(t)\), \(T(t)\), \(A(t)>0\) when, \(t \geq 0\).
 	Next, we determine the limits for the state variables. Given \(N(t) = S(t) + E(t) + I_1(t) + I_2(t) + T(t) + A(t)\), by summing the equations in model (\ref{Model_2}), we obtain,
 	\begin{align*}
 		\frac{dN(t)}{dt} &= \kappa - \mu(S(t) + E(t) + I_1(t) + I_2(t) + T(t)) - \mu_0 A(t) - \mu A(t).
 	\end{align*}
 	As, \(\mu_0 \geq 0\) and
 	\begin{align*}
 		\frac{dN(t)}{dt} &= \kappa - \mu N(t) - \mu_0 A(t) \leq \kappa - \mu N(t),
 	\end{align*}
 	using the Comparison Principle Theorem \cite{mukandavire2009modelling}, we get,
 	\begin{align*}
 		N(t) &\leq N(0)e^{-\mu t} + \frac{\kappa}{\mu}(1 - e^{-\mu t}),
 	\end{align*}
 	where \(N(0) = S(0) + E(0) + I_1(0) + I_2(0) + T(0) + A(0)\). Hence,
 	\[
 	\lim_{t \to \infty} \sup N(t) \leq \frac{\kappa}{\mu}.
 	\]
 \end{proof}
 \noindent So, the feasible area $\Omega$ of model (\ref{Model_2}) can be defined by,
 \begin{equation}
 	\Omega = \left\{ (S(t), E(t), I_1(t), I_2(t), T(t), A(t)) \in \mathbb{R}_+^6 : 0 \leq S, E, I_1, I_2, T, A \leq N, N(t) \leq \frac{\kappa}{\mu} \right\} \tag{5}
 	\label{lemma:3.2} 
 \end{equation}

 The %positivity and boundedness of solutions, 
 fixed points, basic reproduction number, and stability analysis of disease-free equilibrium  and endemic equilibrium state are discussed in Appendix \ref{sec:mathematicalanalysis}.

%\begin{equation}
%\left\{
%\begin{aligned}
%\frac{dS}{dt} &= \kappa - (\alpha \epsilon E + \alpha I_2)S - \mu S, \\
%\frac{dE}{dt} &= (\alpha \epsilon E + \alpha I_2)S - k_1 E, \\
%\frac{dI_1}{dt} &= p\beta E - k_2 I_1 + \gamma I_2, \\
%\frac{dI_2}{dt} &= (1 - p)\beta E - k_3 I_2, \\
%\frac{dT}{dt} &= \psi I_1 - k_4 T, \\
%\frac{dA}{dt} &= \delta_1 I_1 + \delta_2 I_2 + \xi T - \mu_0 A - \mu A
%\end{aligned}
%\right.
%\tag{2}
%\label{Model_2}
%\end{equation}
%===============================================omitted the uppder paprt=====================================
%============================================================================================================
\section{Computational Analysis of Numerical Data}
\label{sec:simulations}
UNAIDS \cite{unaids2023}, which stands for the Joint United Nations Programme on HIV/AIDS and Centers for Disease Control and Prevention (CDC) \cite{unaids2023data} %\cite{CDC_HIV_Surveillance} 
official websites are the source of the HIV case statistics. We took into consideration the transmission of HIV among individuals between the ages of 15 and 60. Then, we use these numbers to model (\ref{Model_2}). We utilized the annual data on new HIV cases in North America from 2001 to 2024 with the Markov Chain Monte Carlo (MCMC) method to estimate the unknown parameters of model (\ref{Model_2}), as presented in Table \ref{tab:2.2}. Additionally, we employed Partial Rank Correlation Coefficients (PRCC) to assess the global sensitivity of the parameters in model (\ref{Model_2}). The objective is to find the most critical characteristic influencing HIV transmission. Figure \ref{fig:sensitivity_analysis} shows the PRCCs of the parameters with regard to $\mathcal{R}_0$. Our study results show that the parameters \(\varepsilon\) and \(\alpha\) have a positive correlation with \(\mathcal{R}_0\), while the parameters \(\psi\), \(\delta\), \(\gamma\), and \(p\) are negatively correlated. Moreover, the parameters \(p\) and \(\alpha\) exhibit the highest sensitivity to \(\mathcal{R}_0\). Therefore, reducing \(\alpha\) and increasing \(p\) could effectively lower HIV transmission.
\begin{figure}[H]
  \centering
  \includegraphics[width=0.4\linewidth]{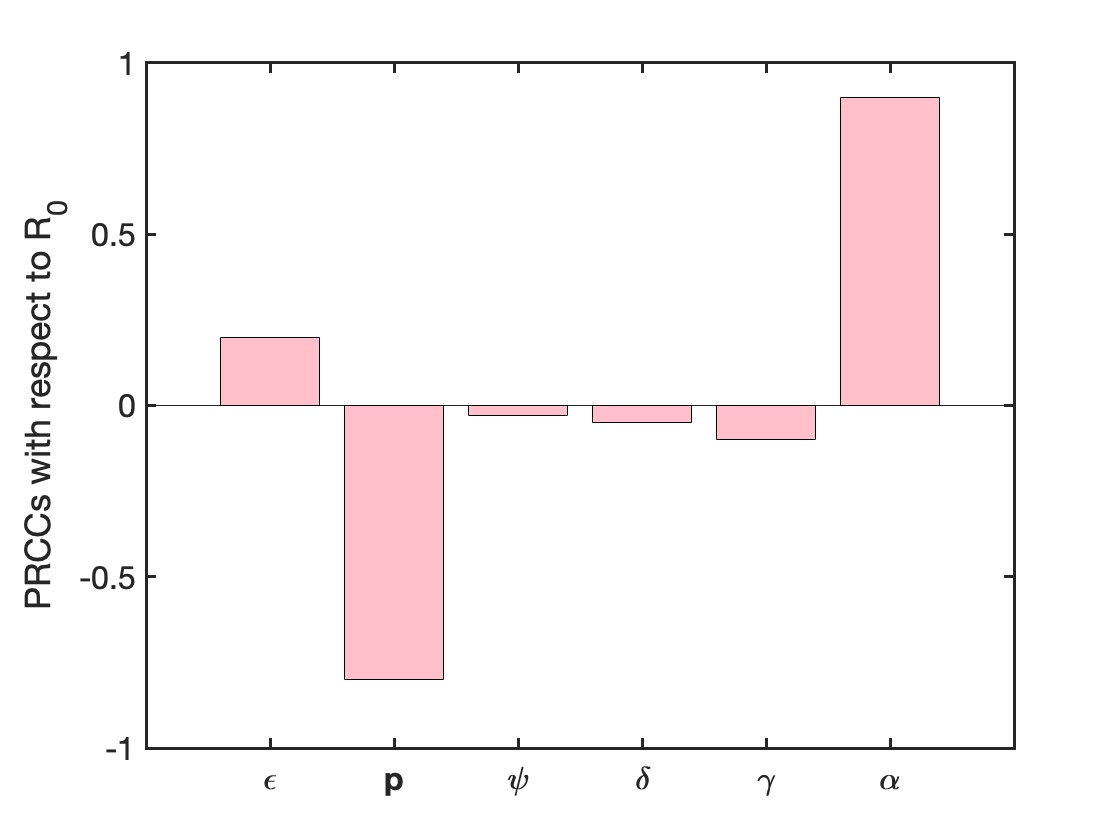}
  \caption{Global sensitivity analysis for the parameters \(\varepsilon, p, \psi, \delta, \gamma, \alpha\).}
  \label{fig:sensitivity_analysis}
\end{figure}

\noindent In our simulations, to model the fraction of latent individuals who become aware of their condition, we define the function \(p = p(t) = (p_0 - p_{\text{max}})e^{-at} + p_{\text{max}}\), where \(p_0\) is the initial fraction of aware infected individuals, \(p_{\text{max}}\) is the maximum fraction achievable, and \(a\) is a constant describing the rate of increase.\\
Evidence supporting the stability of the model is presented in figures \ref{fig:stability_endemic} and \ref{fig:stability_disease_free}. Figures \ref{fig:a}, \ref{fig:b}, \ref{fig:c}, \ref{fig:d} and \ref{fig:e} illustrate the stability at the endemic equilibrium for susceptible, latent, aware infected, unaware infected and treated individual respectively which tends to be unstable when \(R_0 > 1\). Initially, the number of susceptible individuals, \(S\), increases and subsequently stabilizes. The populations of \(E\), \(I_1\), and \(I_2\) also exhibit initial increases before reaching a steady state. The number of treated individuals, \(T\), consistently rises. These groups all show pronounced initial increases followed by stabilization. In contrast, figures \ref{fig:aa}, \ref{fig:bb}, \ref{fig:cc}, \ref{fig:dd} and \ref{fig:ee} show that the disease-free equilibrium point is generally stable when \(R_0 < 1\). The numbers of susceptible people ($S$), tend to climb to a greater level, after which they tend to remain steady. On the other hand, the numbers of individuals who are exposed ($E$), aware of infection ($I_1$), and unaware of infection ($I_2$) tend to decline initially, and then they tend to remain constant after that. Over the course of time, the people who have been treated ($T$) tend to develop.
\begin{figure}[H]
	\centering
	\subfigure[The variation in the number of susceptible individuals ($S$) over time.]{\includegraphics[width=0.3\linewidth]{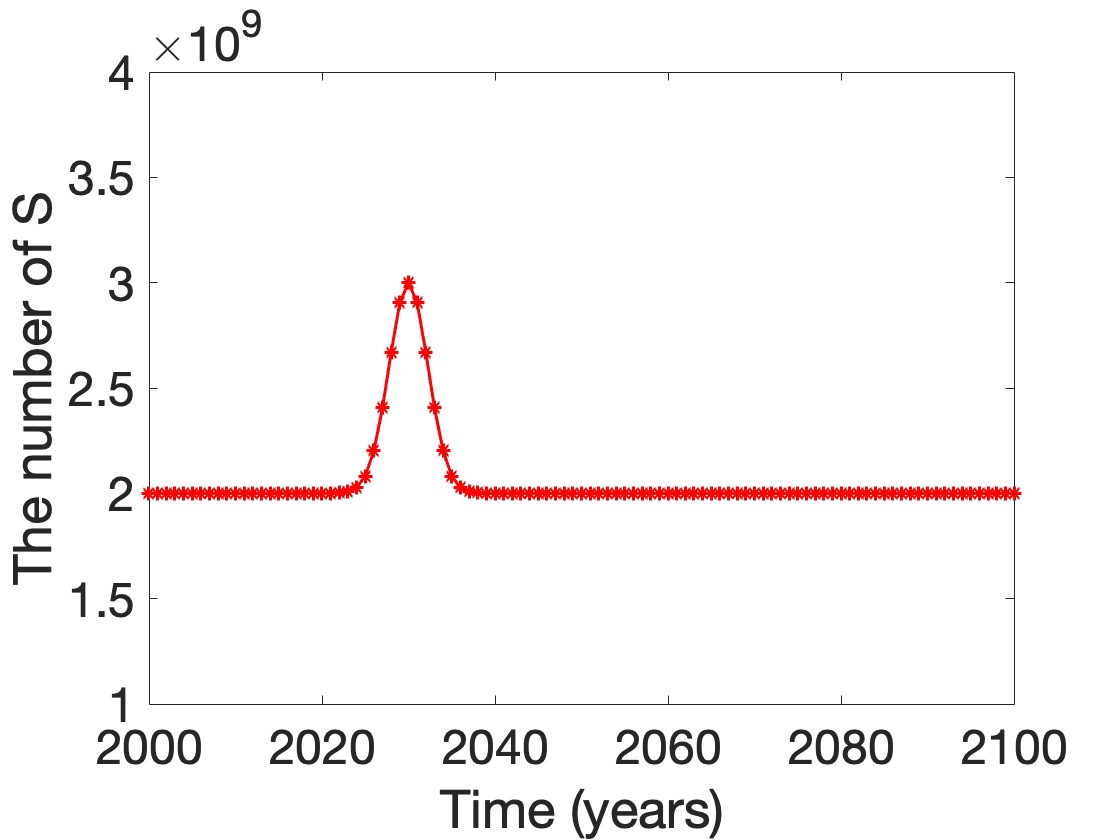}\label{fig:a}}
	\subfigure[The variation in the number of latent individuals ($E$) over time.]{\includegraphics[width=0.3\linewidth]{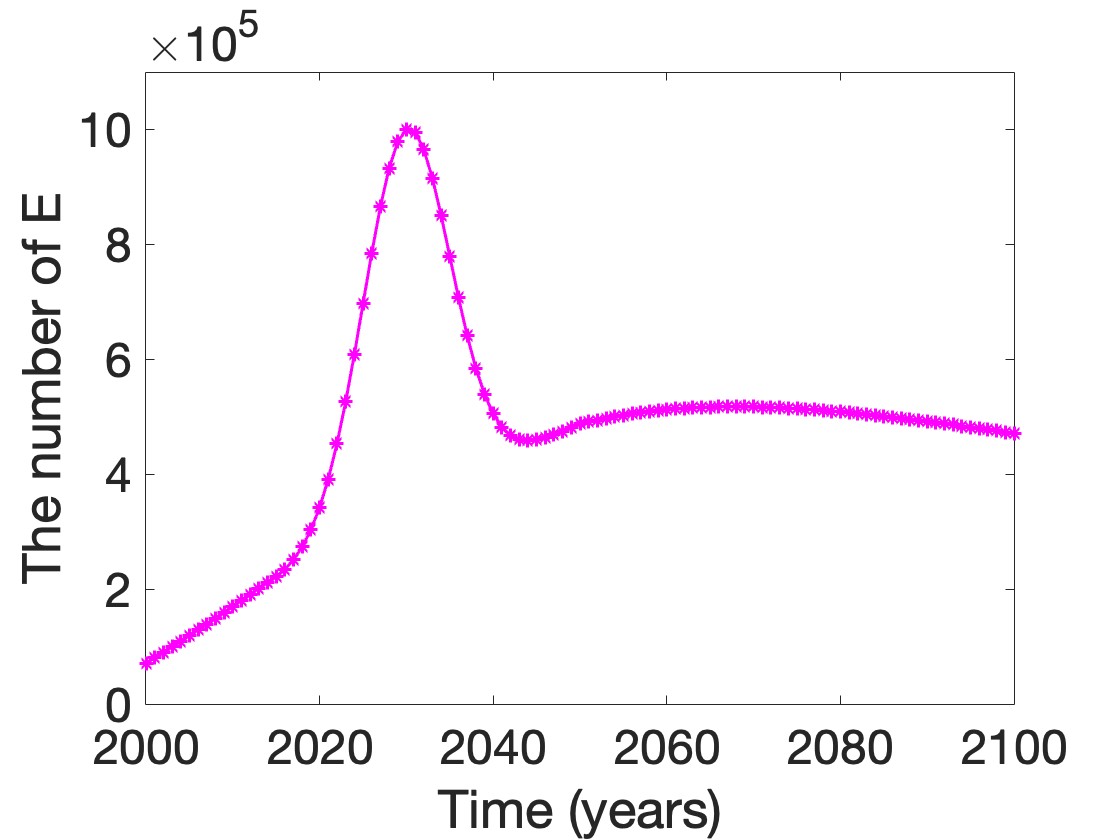} \label{fig:b}}
	\subfigure[The variation in the number of aware infected individuals ($I_1$) over time.]{\includegraphics[width=0.3\linewidth]{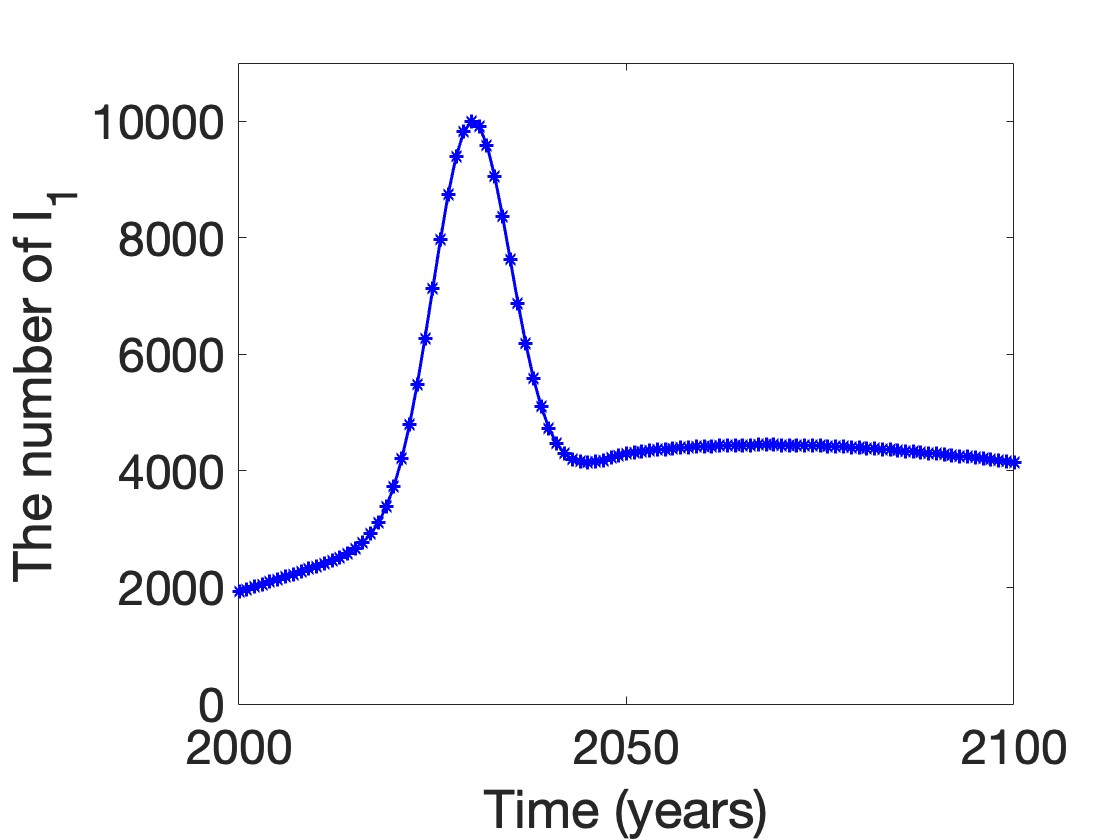}\label{fig:c}}
	\subfigure[The variation in the number of unaware infected individuals ($I_2$) over time.]{\includegraphics[width=0.3\linewidth]{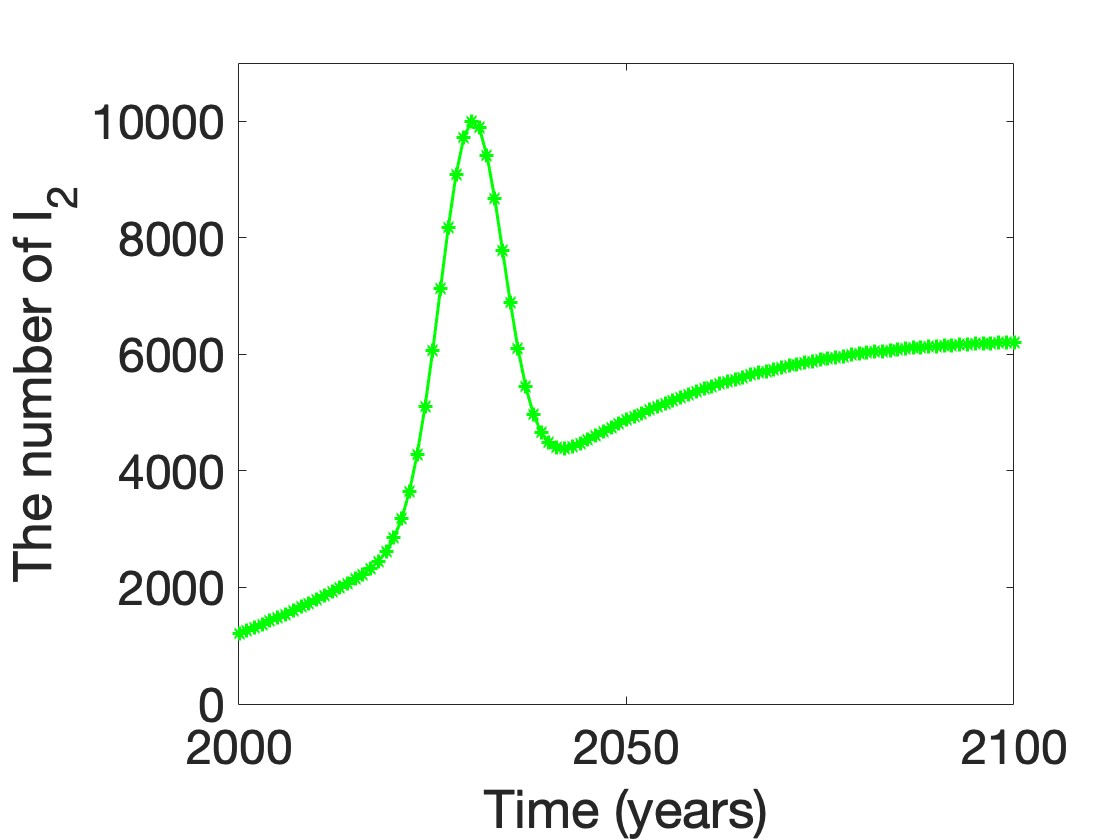}\label{fig:d}}
	\subfigure[The variation in the number of individuals who receive ART treatment ($T$) over time.]{\includegraphics[width=0.3\linewidth]{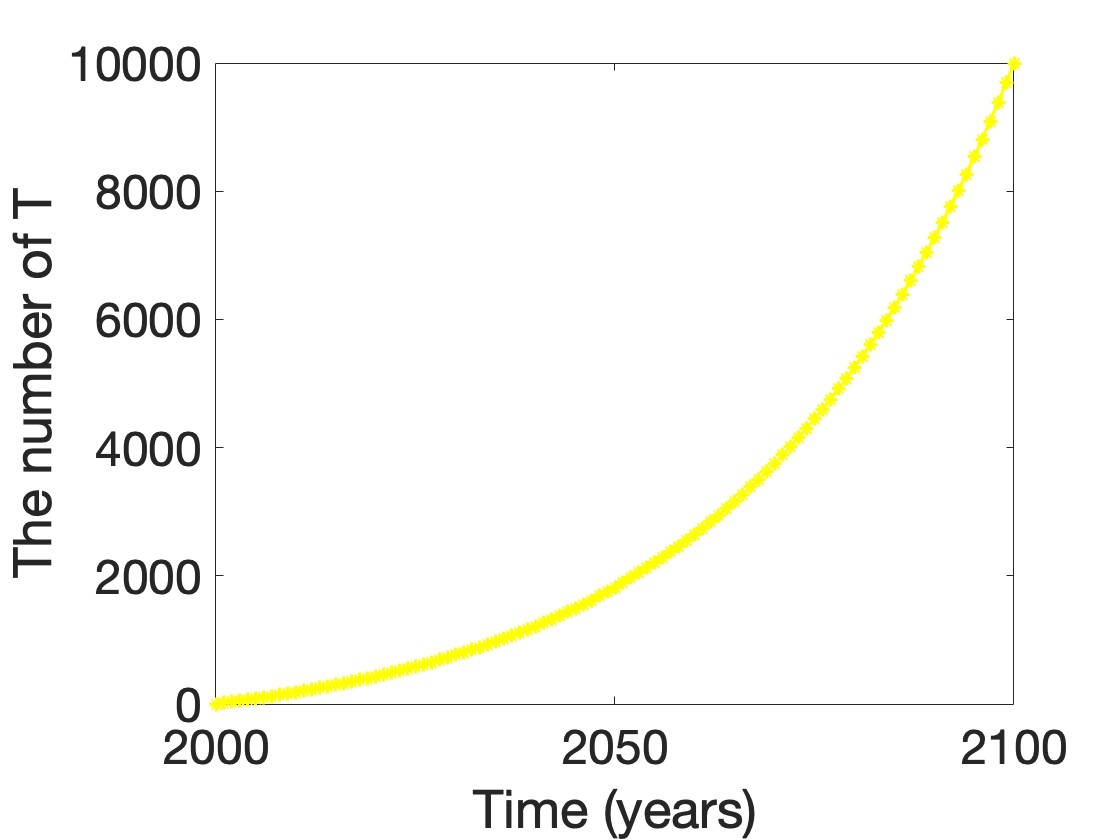}\label{fig:e}}
	\caption{A stability analysis with the following parameter values: \(\alpha = 1.7649 \times 10^{-6}\), \(\gamma = 0.015697\), \(\delta = 0.056044\), \(\psi = 0.34877\), \(p = 0.022591\), \(a = 0.055961\), \(p_{\text{max}} = 0.45584\), and \(R_0 = 1.7261 > 1\), indicating an endemic equilibrium.}
	\label{fig:stability_endemic}
\end{figure}

\begin{figure}[H]
\centering
\subfigure[The variation in the number of susceptible individuals ($S$) over time.]{\includegraphics[width=0.3\linewidth]{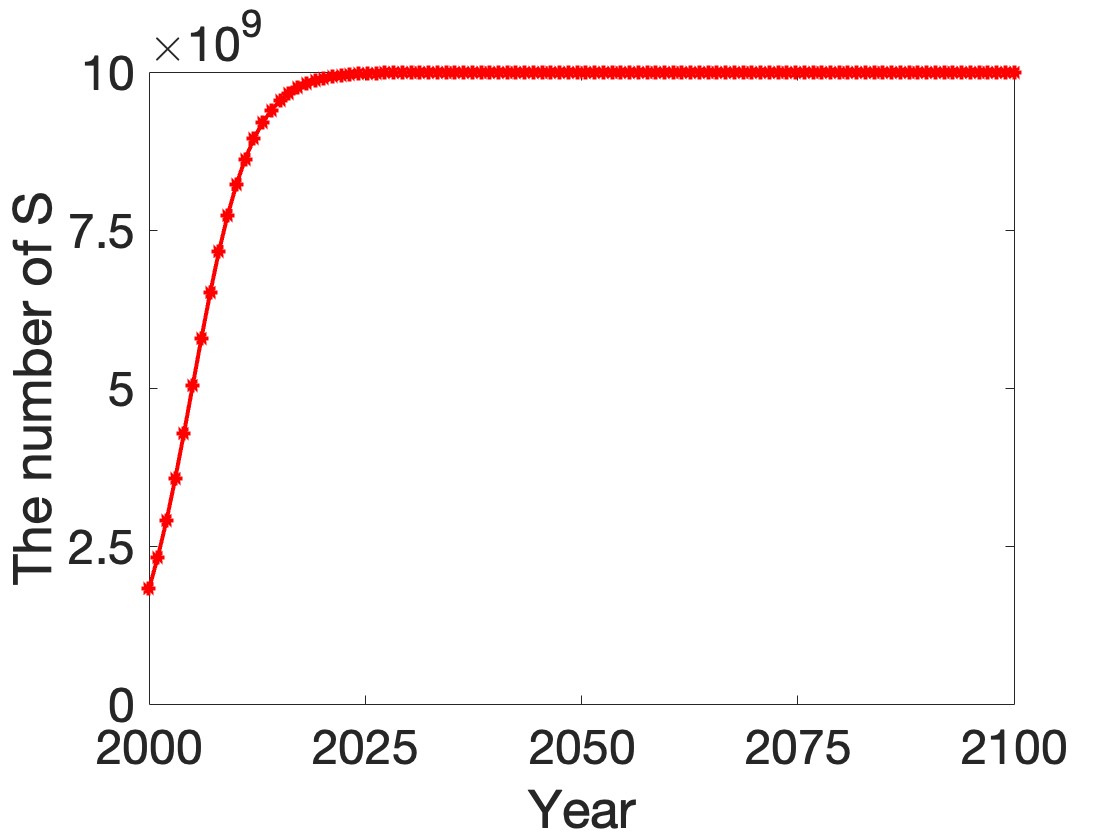}\label{fig:aa}}
\subfigure[The variation in the number of latent individuals ($E$) over time.]{\includegraphics[width=0.3\linewidth]{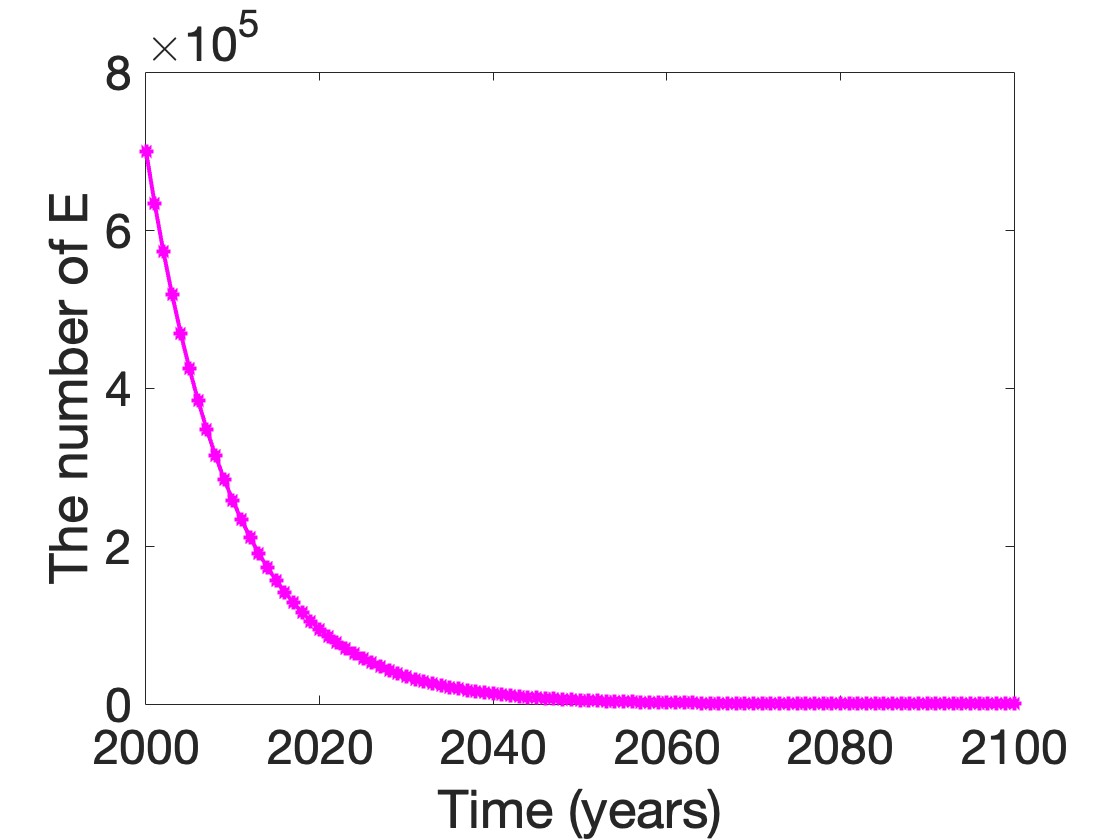}\label{fig:bb}}
\subfigure[The variation in the number of aware infected individuals ($I_1$) over time.]{\includegraphics[width=0.3\linewidth]{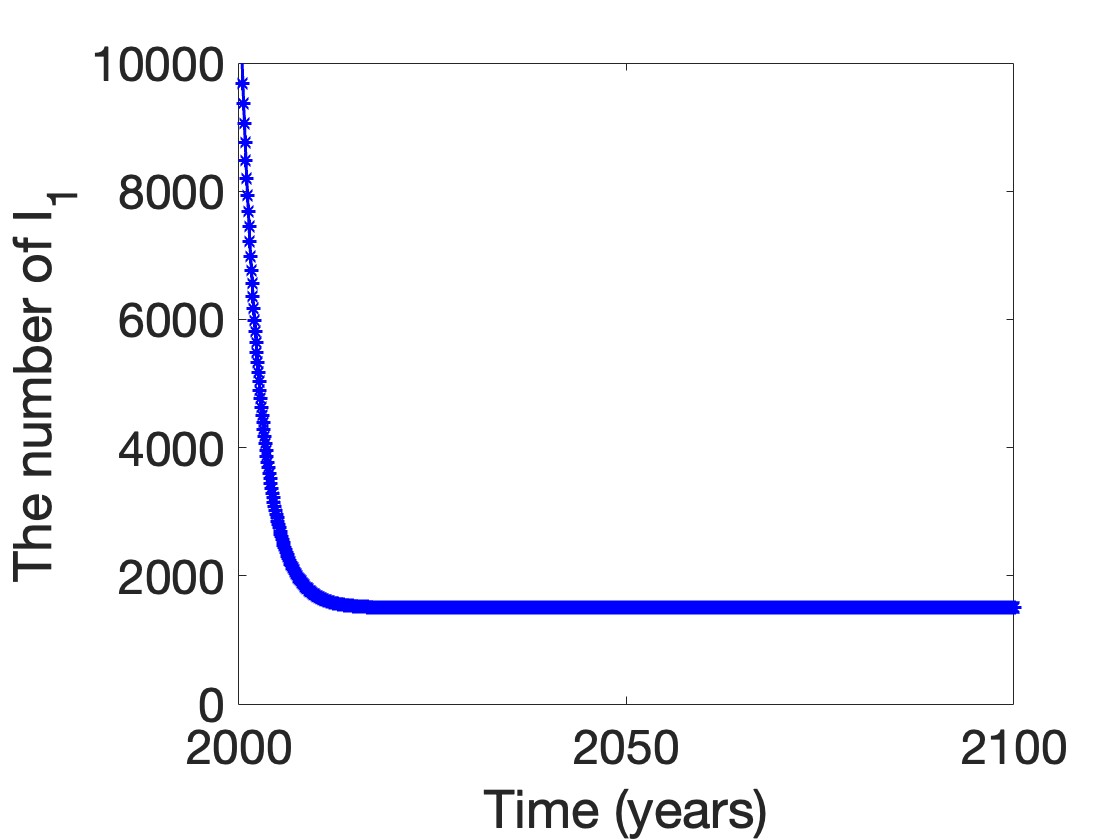}\label{fig:cc}}
\subfigure[The variation in the number of unaware infected individuals ($I_2$) over time.]{\includegraphics[width=0.3\linewidth]{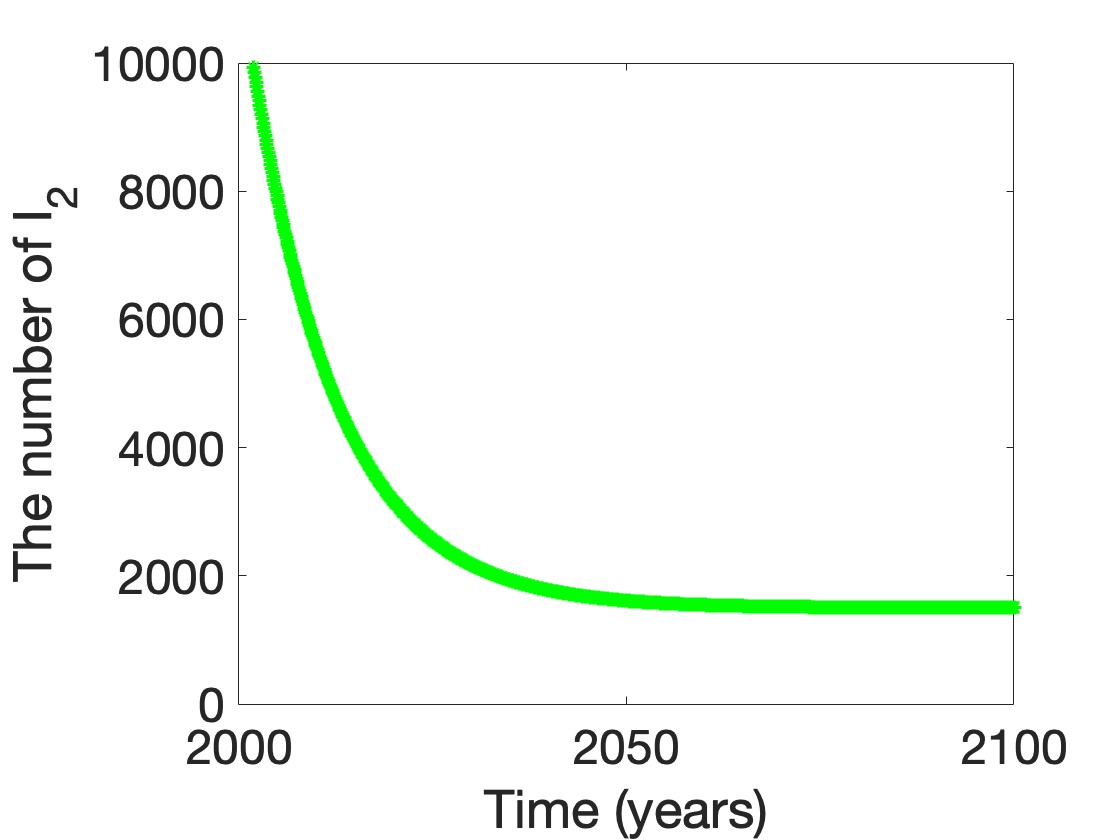}\label{fig:dd}}
\subfigure[The variation in the number of individuals who receive ART treatment ($T$) over time.]{\includegraphics[width=0.3\linewidth]{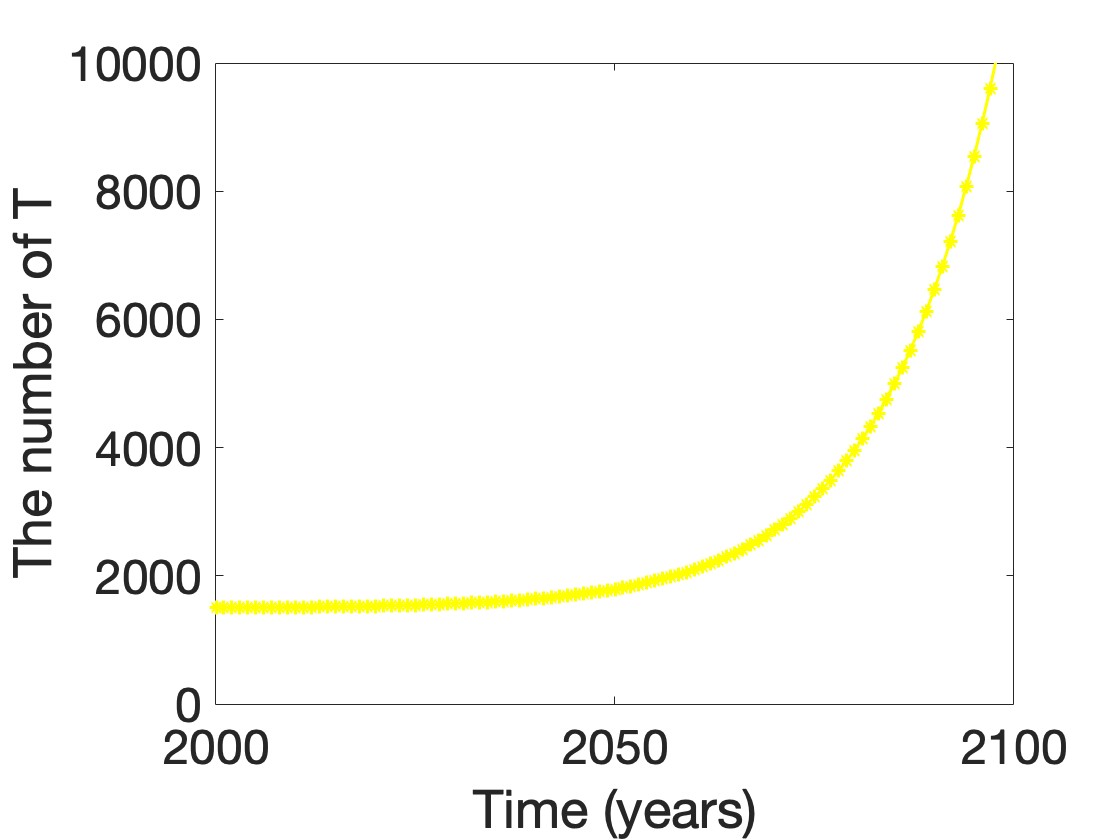}\label{fig:ee}}
\caption{A stability analysiswith the following parameter values: $\alpha = 1.007e^{-6}$, $\gamma = 0.01$, $\delta = 0.15$, $\psi = 0.34877$, $p = 0.022591$, $a = 0.055961$, $p_{\text{max}} = 0.4554$, and $R_0 = 0.7589 < 1$, indicating a disease-free equilibrium.}
\label{fig:stability_disease_free}
\end{figure}

\noindent Following that, we will examine the influence that a single parameter has had on the AIDS pandemic. Specifically, we are looking at the impact of the transmission rate $\alpha$. It is possible to efficiently restrict the number of infected persons during the latent stage of AIDS, as shown in figures \ref{fig:alpha} and \ref{fig:alpha1}, by lowering the transmission rate. Also, we can increase the amount of treated people by doing the same. The epidemic ends faster when the rate of transmission is lower. This suggests that in order to limit the spread of HIV infection, public health officials must keep enhancing control measures. The variations in the transmission rate show that this is necessary. As a result, it is essential to raise the level of awareness among those who are not yet aware. Figures \ref{fig:bottomLeft} and \ref{fig:bottomLeft1} show that the impact of $p_0$ is very insignificant. Nonetheless, in figures \ref{fig:bottomRight} and  \ref{fig:bottomRight1}, we can see that, there exists a robust association between the value of $p_{max}$ and the quantity of latent individuals; as the $p_{max}$ value increases, the number of treated people decreases.

\begin{figure}[H]
\centering
\subfigure[The number of infected individuals (\(I_1+I_2\)), varying with \(p_0\).]{\includegraphics[width=0.45\linewidth]{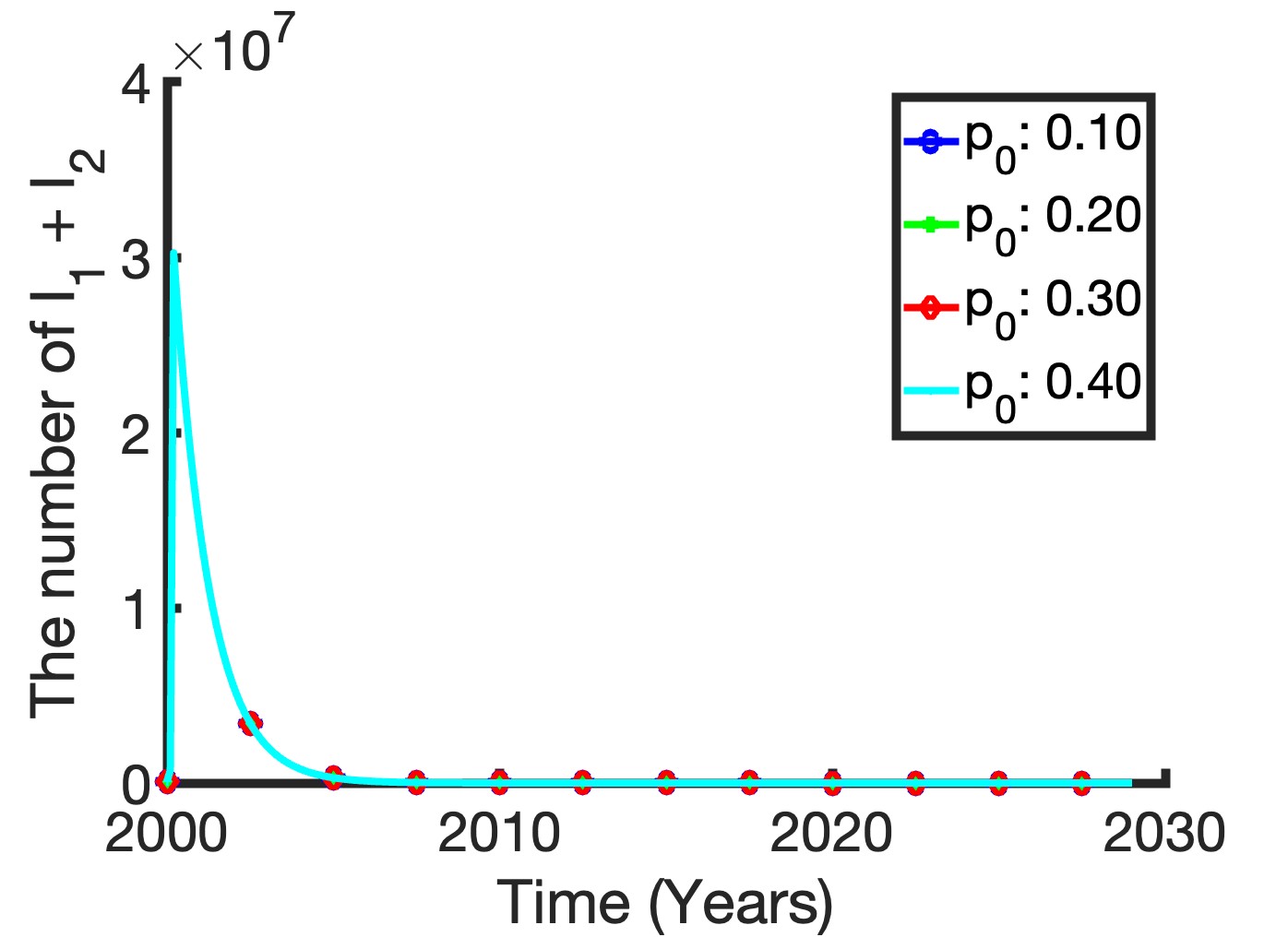}\label{fig:bottomLeft}}
\subfigure[The number of treated individuals (\(T\)), varying with \(p_0\).]{\includegraphics[width=0.45\linewidth]{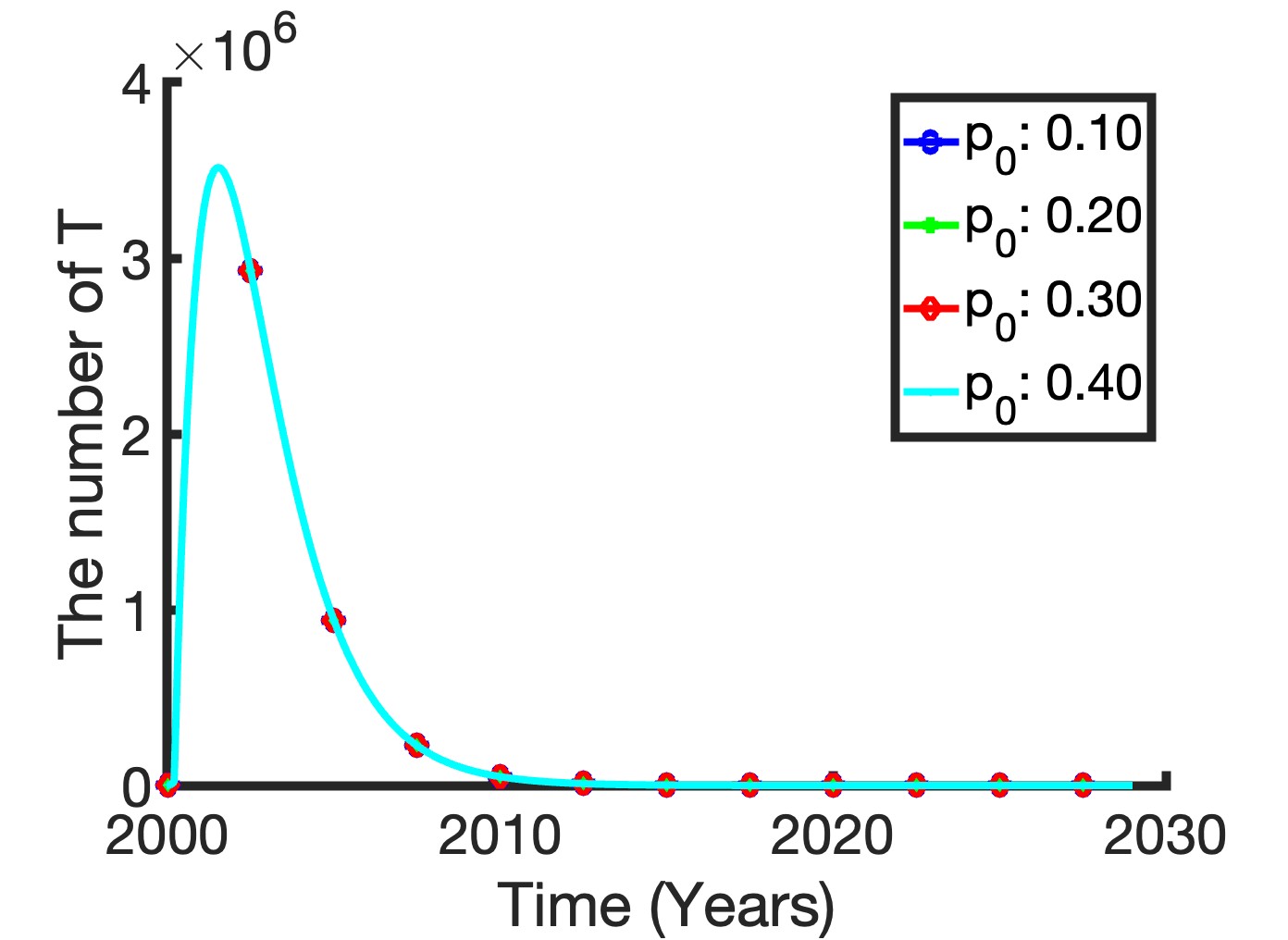}\label{fig:bottomLeft1}}
\subfigure[The number of infected individuals (\(I_1+I_2\)), varying with \(p_{max}\).]{\includegraphics[width=0.45\linewidth]{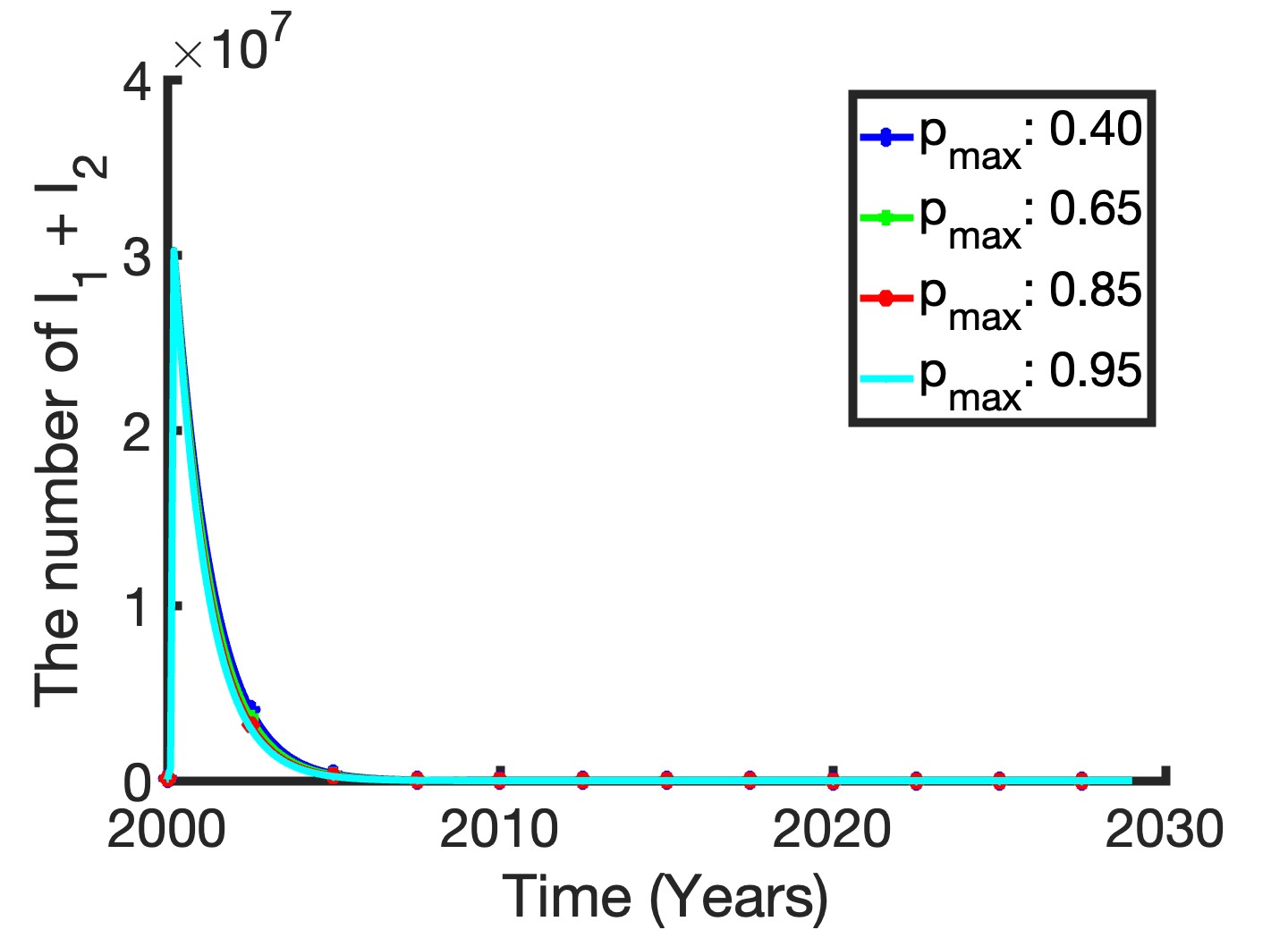}\label{fig:bottomRight}}
\subfigure[The number of treated individuals (\(T\)), varying with \(p_{max}\).]{\includegraphics[width=0.45\linewidth]{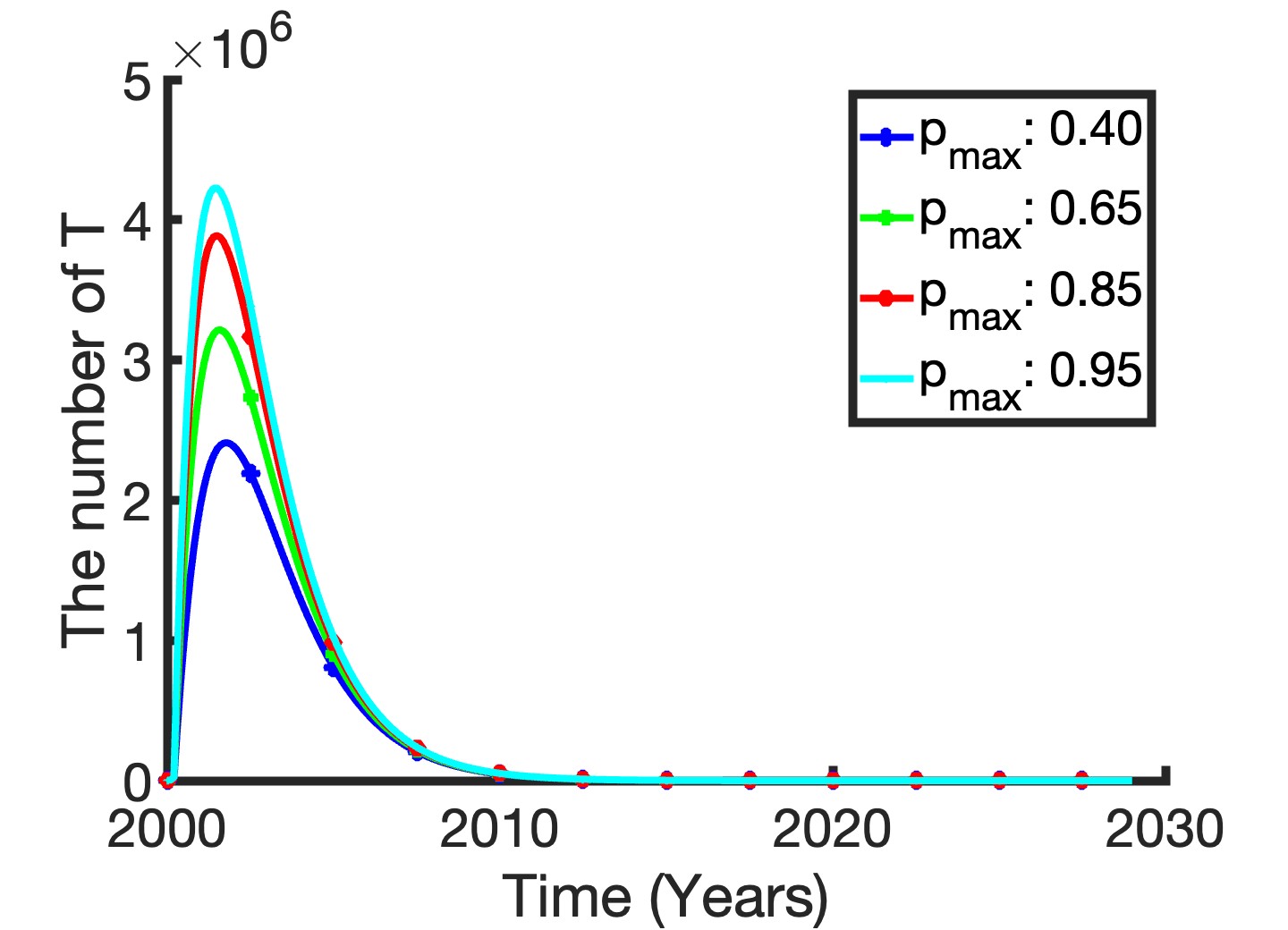}\label{fig:bottomRight1}}
\subfigure[The number of infected individuals (\(I_1+I_2\)), varying with \(\alpha\).]{\includegraphics[width=0.45\linewidth]{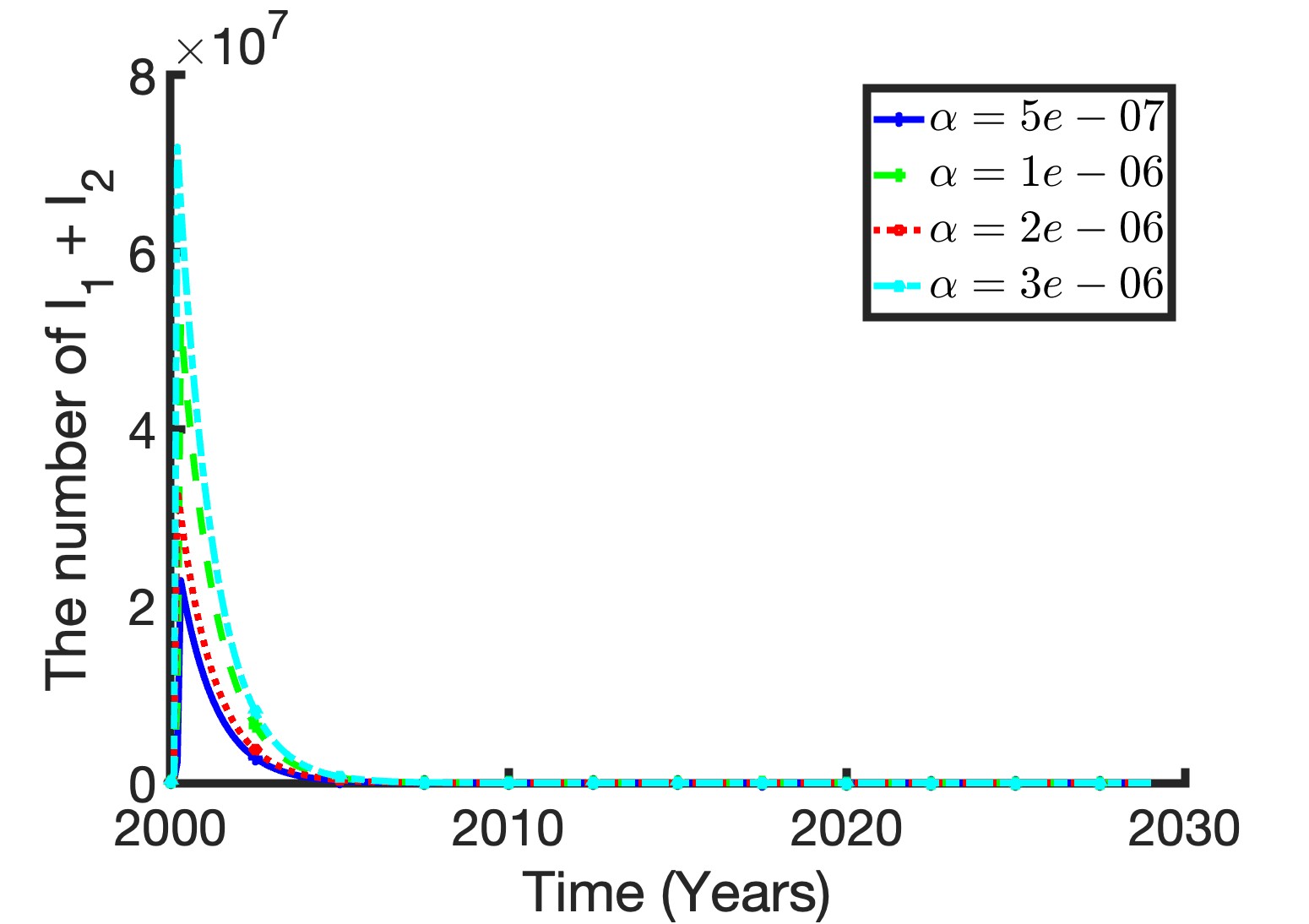}\label{fig:alpha}}
\subfigure[The number of treated individuals (\(T\)), varying with \(\alpha\).]{\includegraphics[width=0.45\linewidth]{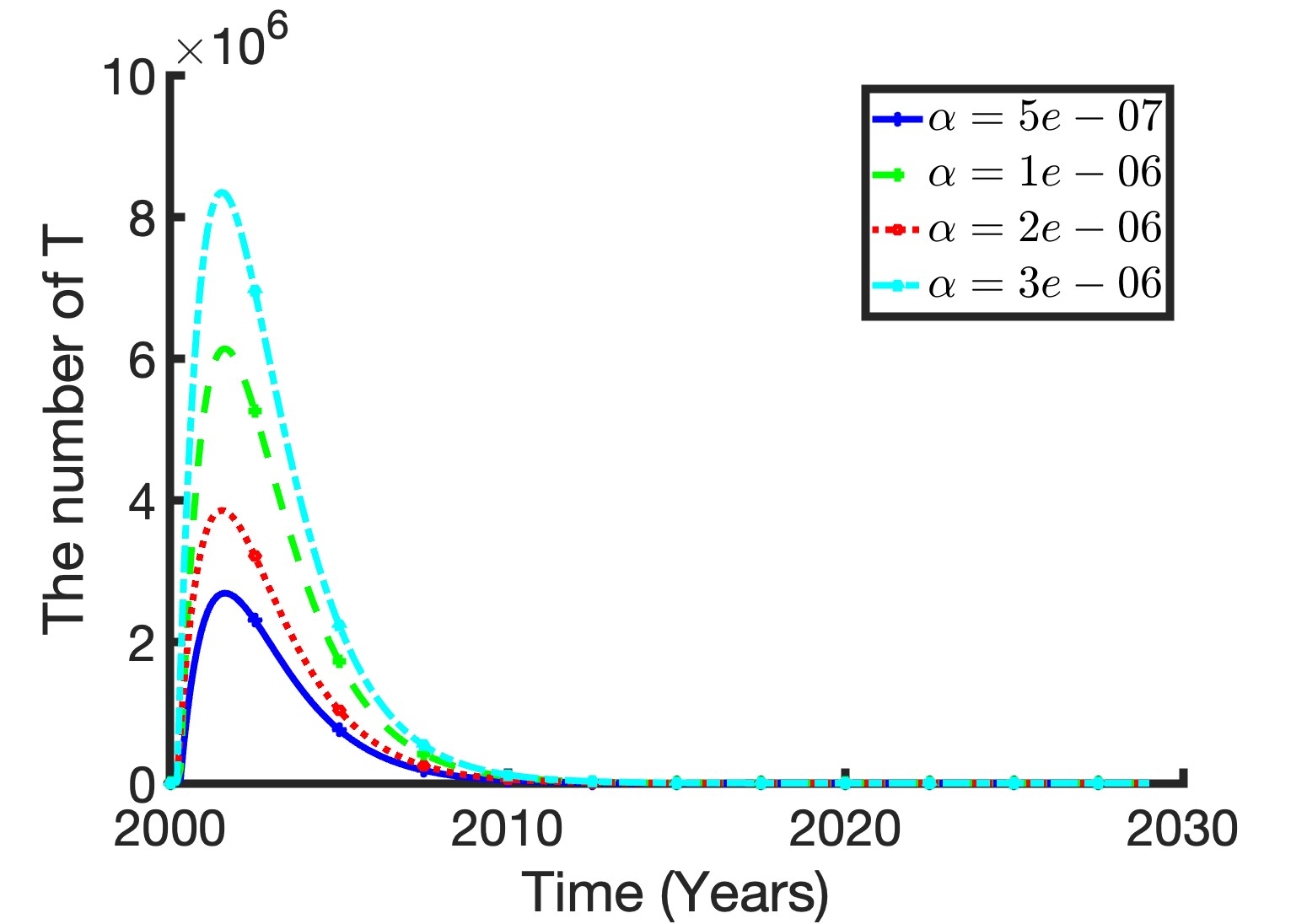}\label{fig:alpha1}}
\caption{Comparative analysis of the quantity of infected and treated individuals.}
\label{fig:images}
\end{figure}

\begin{figure}[H]
	\centering
	\subfigure[$\epsilon = 0.05$]{\includegraphics[width=0.45\linewidth]{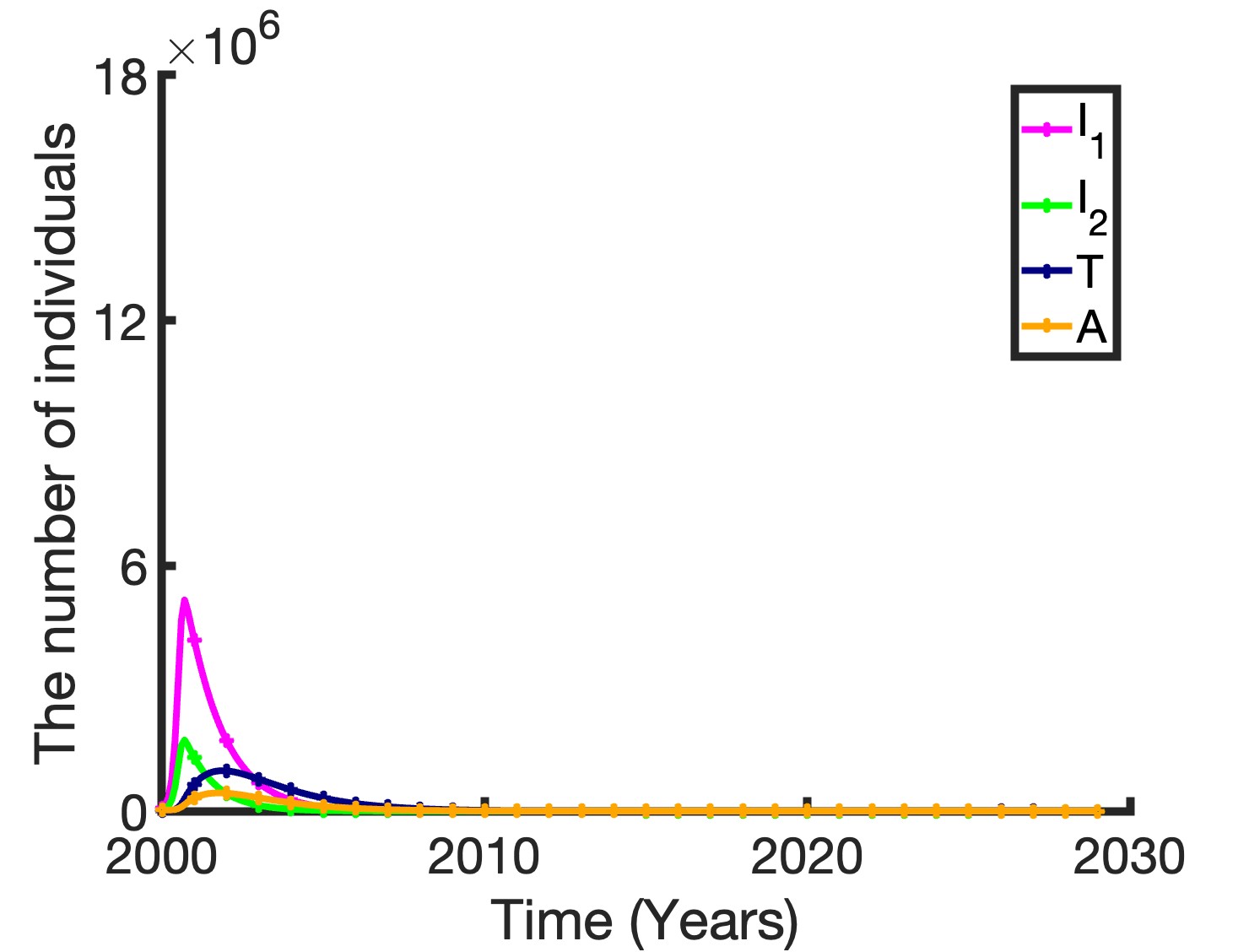}\label{fig:ep05}}
	\subfigure[$\epsilon = 0.10$]{\includegraphics[width=0.45\linewidth]{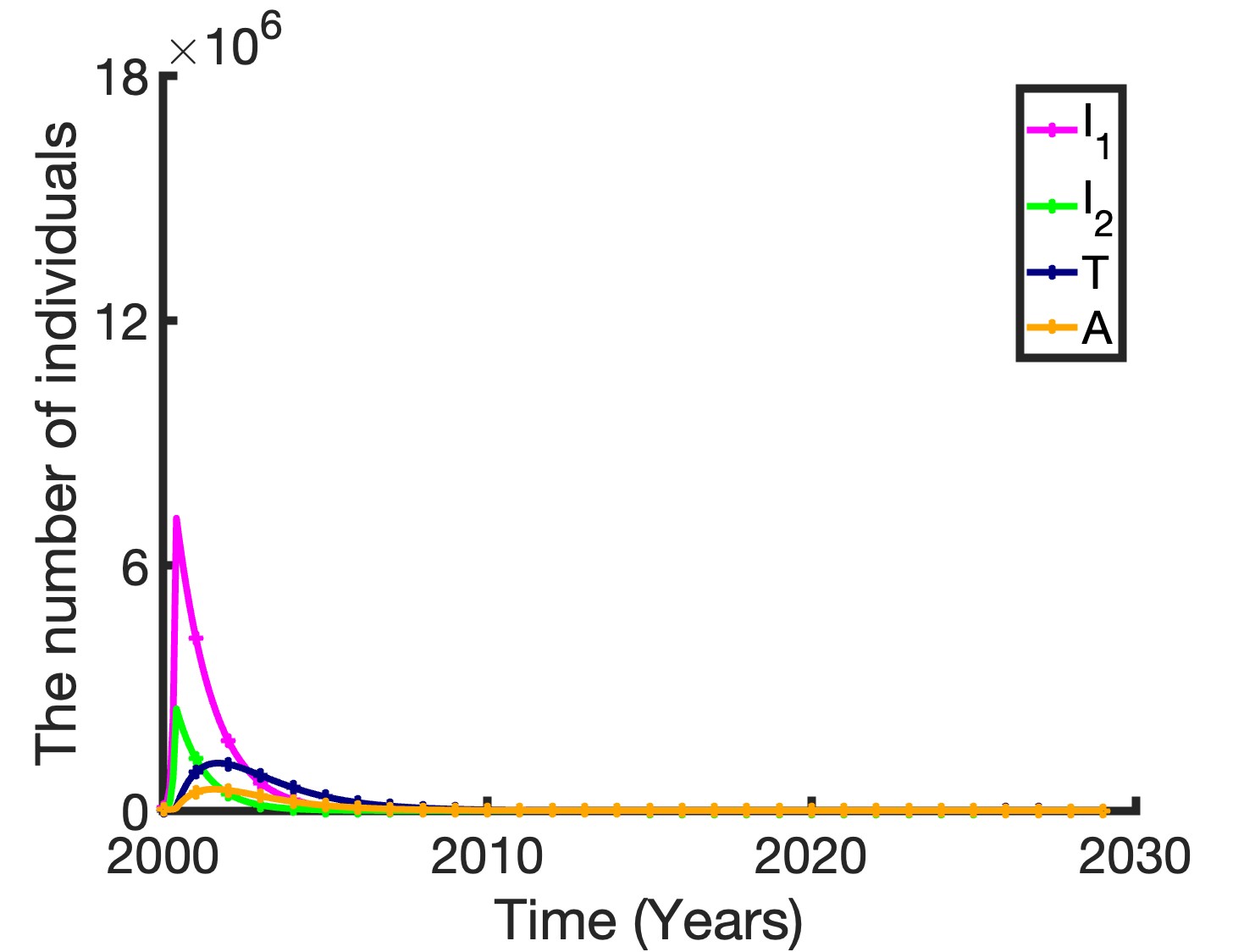}\label{fig:ep10}}
	\subfigure[$\epsilon = 0.15$]{\includegraphics[width=0.45\linewidth]{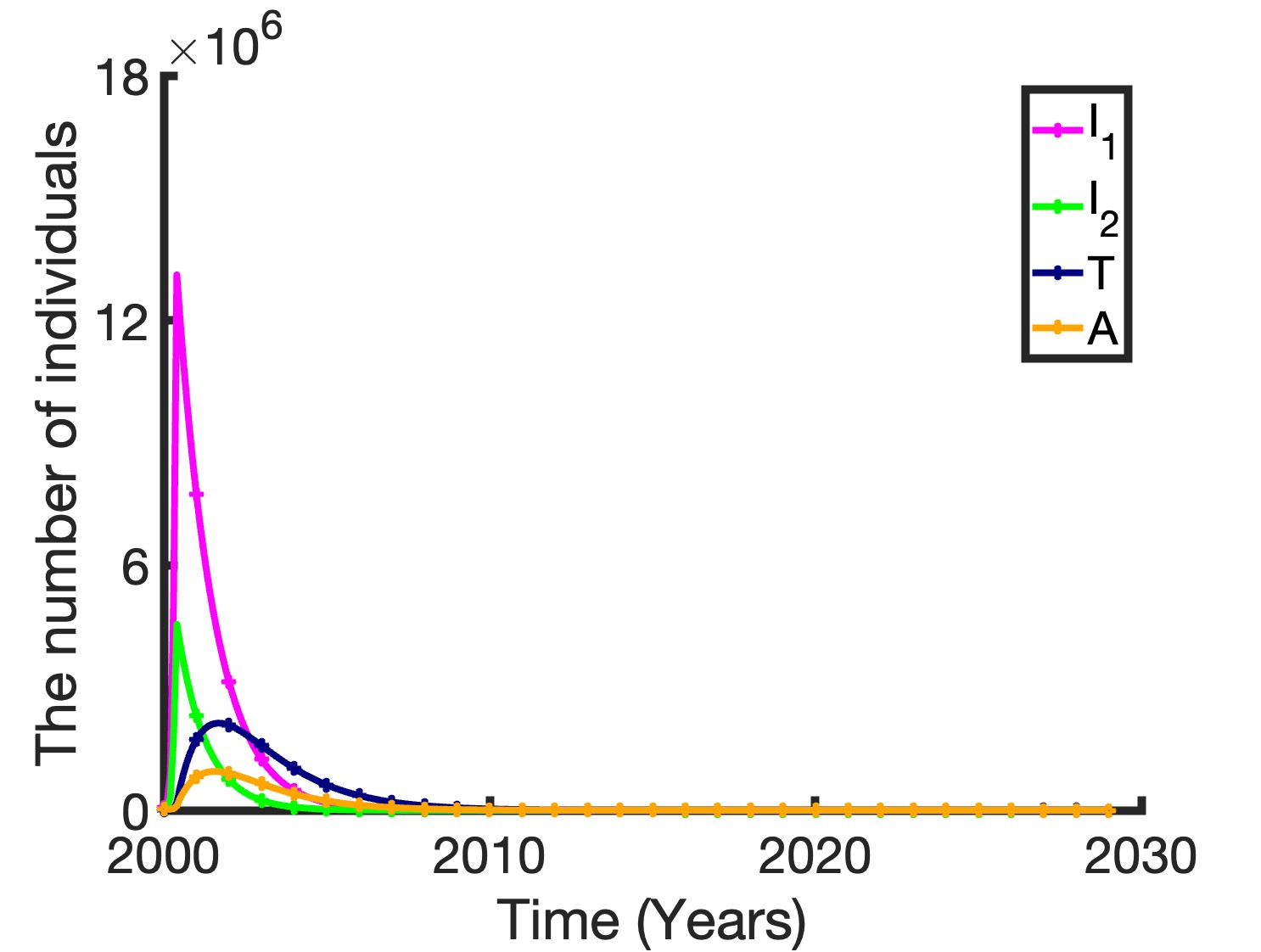}\label{fig:ep15}}
	\subfigure[$\epsilon = 0.25$]{\includegraphics[width=0.45\linewidth]{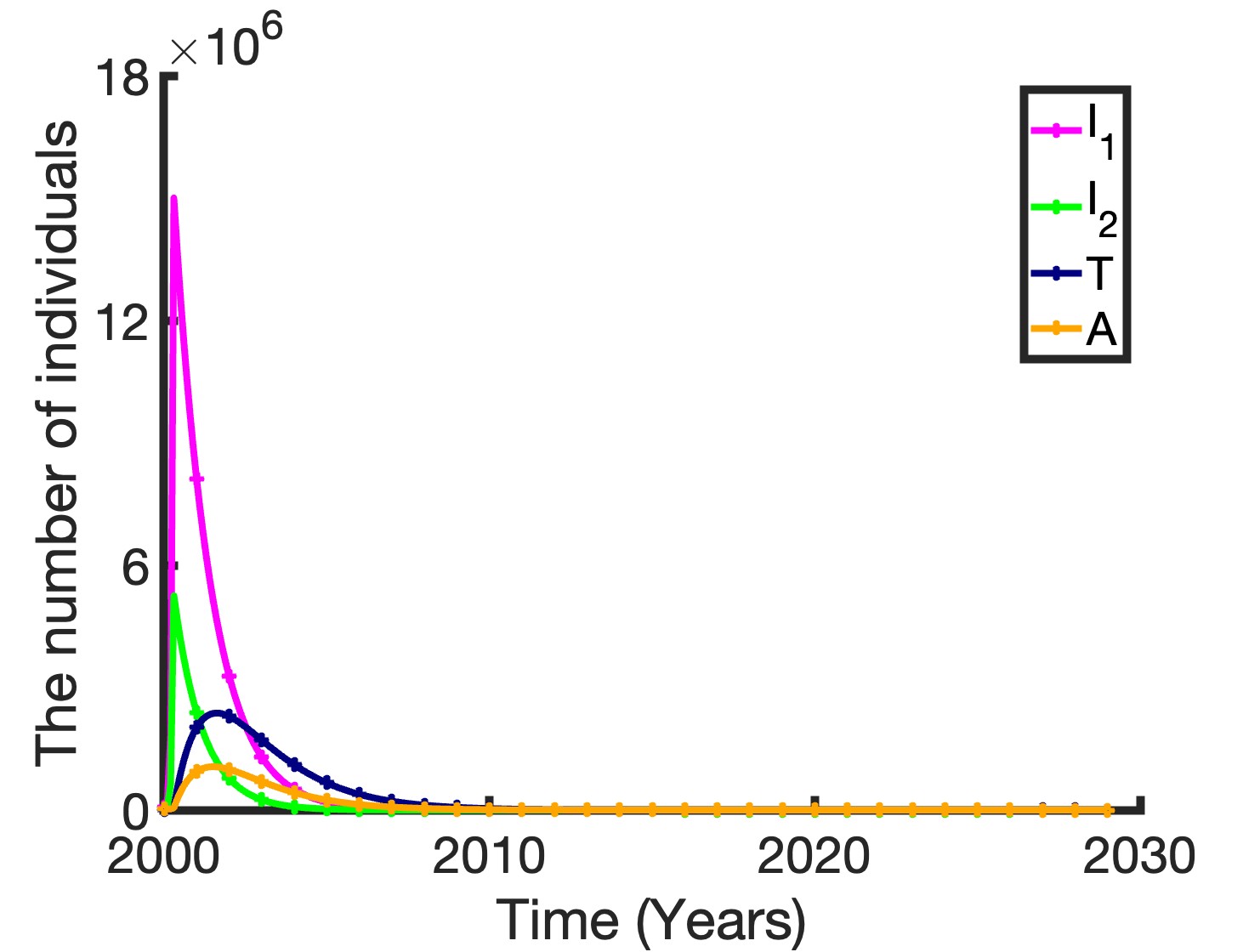}\label{fig:ep25}}
	\caption{The effect of $\epsilon$, the reduced infectiousness of individuals.}
	\label{fig:epsilon-effects}
\end{figure}
 
\noindent In figures \ref{fig:ep05}, \ref{fig:ep10}, \ref{fig:ep15} and \ref{fig:ep25}, we see the effects of $\epsilon$ on individuals. By the reduced number of infectiousness, we observe that the number of infected individuals decreases.  Figures \ref{fig:gamma001}, \ref{fig:gamma202}, \ref{fig:gamma657} and \ref{fig:gamma999} investigate the effect of the rate at which people who were unaware before become conscious, denoted as $\gamma$, on the total number of persons who are infected. Our findings indicate that when the value of $\gamma$ is raised from $0.001$ to $0.999$, more individuals are becoming aware at a faster rate which directly impacts the number of treated population. We notice that the number of treated individuals are increasing when we can raise awareness.

\begin{figure}[H]
\centering
\subfigure[$\gamma = 0.001$]{\includegraphics[width=0.49\linewidth]{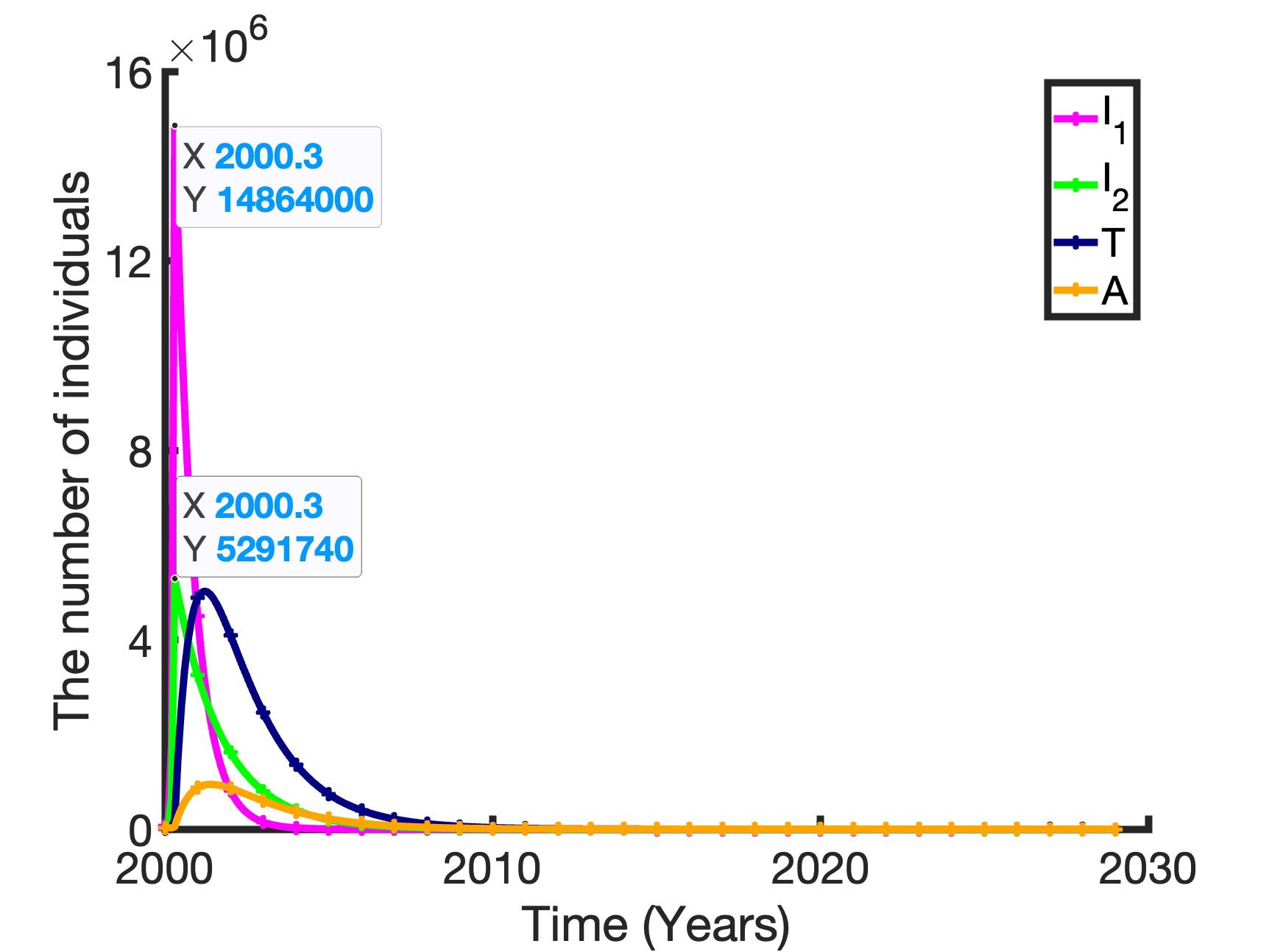}\label{fig:gamma001}}
\subfigure[$\gamma = 0.202$]{\includegraphics[width=0.49\linewidth]{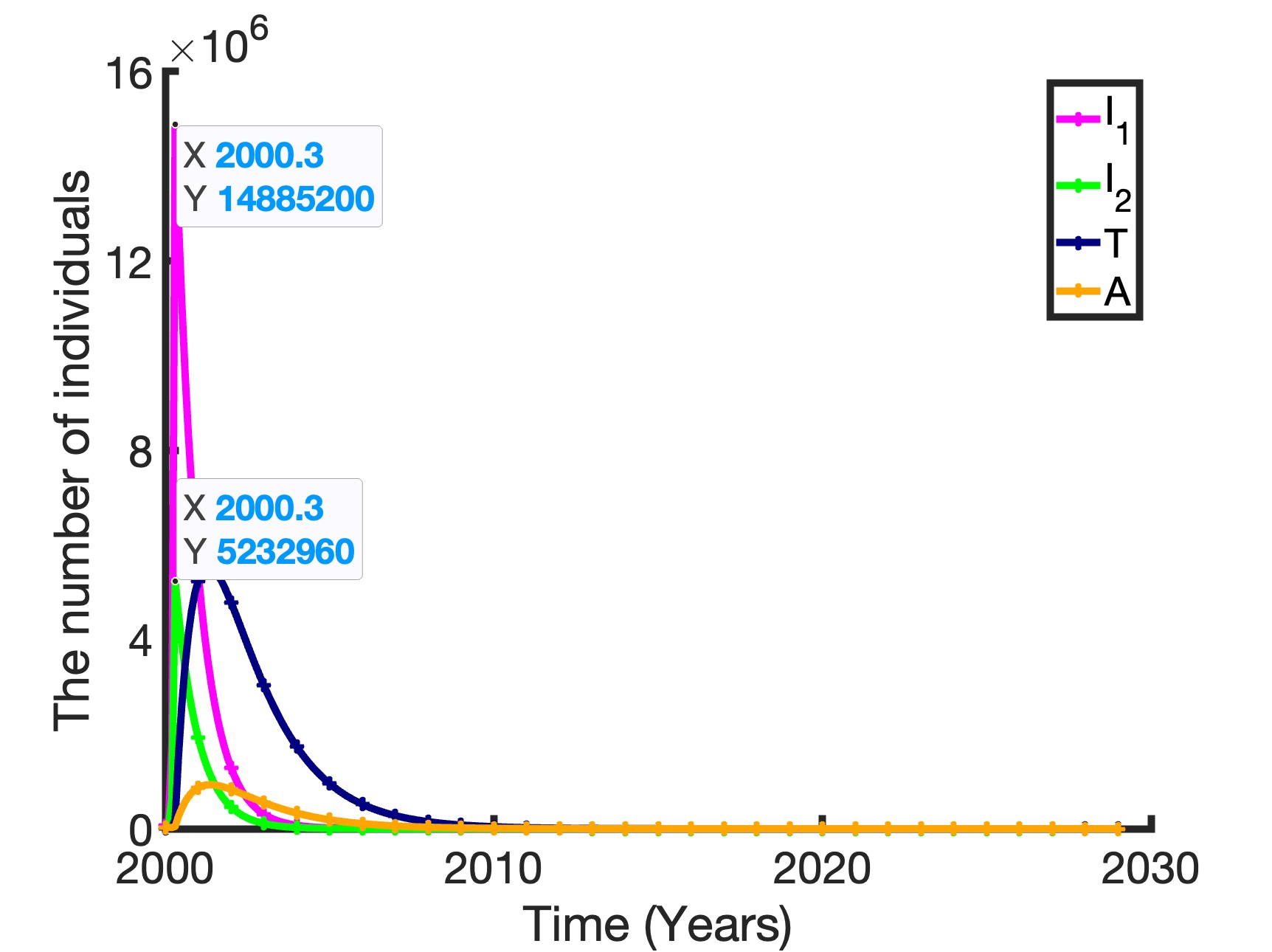}\label{fig:gamma202}}
\subfigure[$\gamma = 0.657$]{\includegraphics[width=0.49\linewidth]{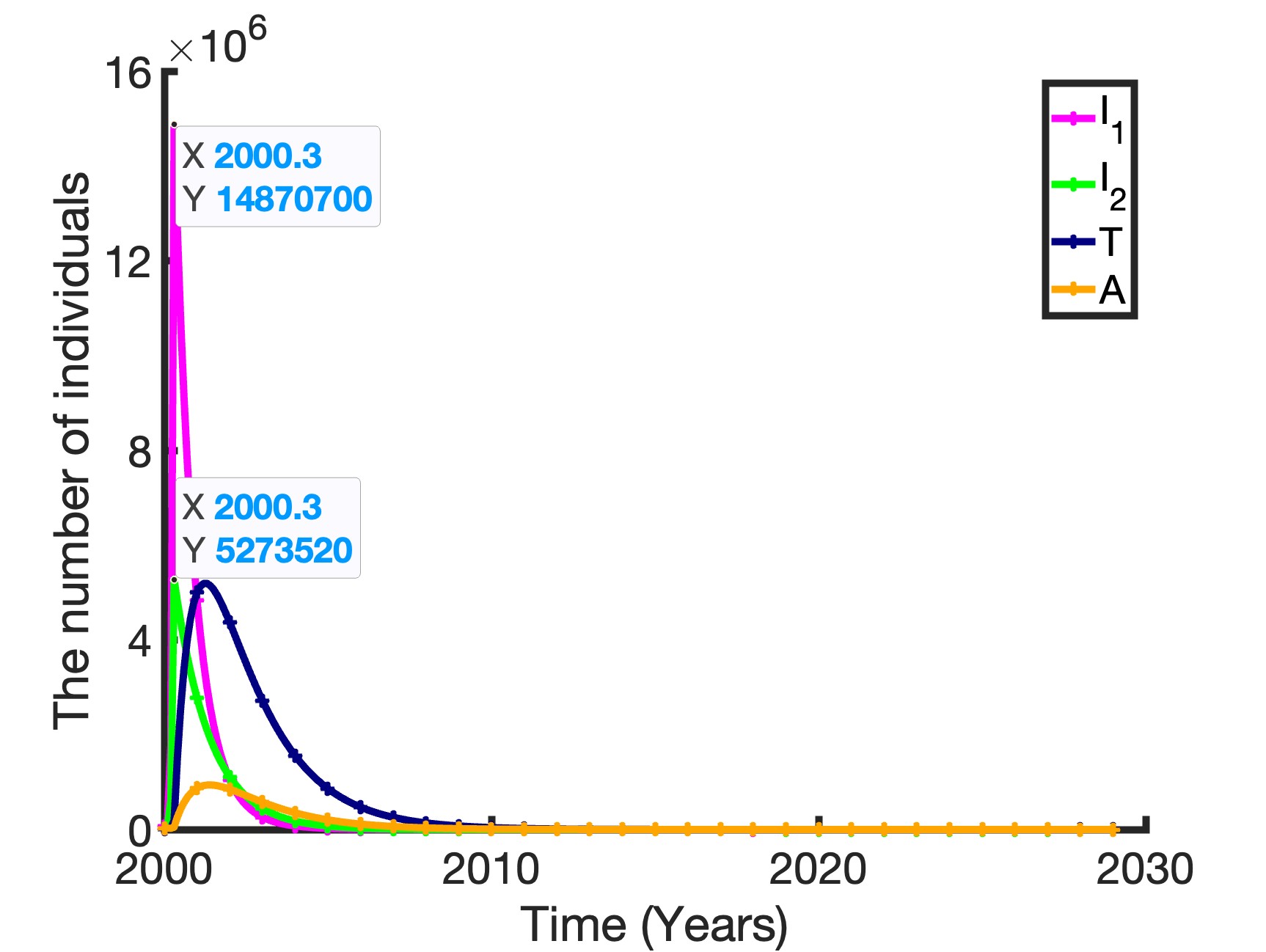}\label{fig:gamma657}}
\subfigure[$\gamma = 0.999$]{\includegraphics[width=0.49\linewidth]{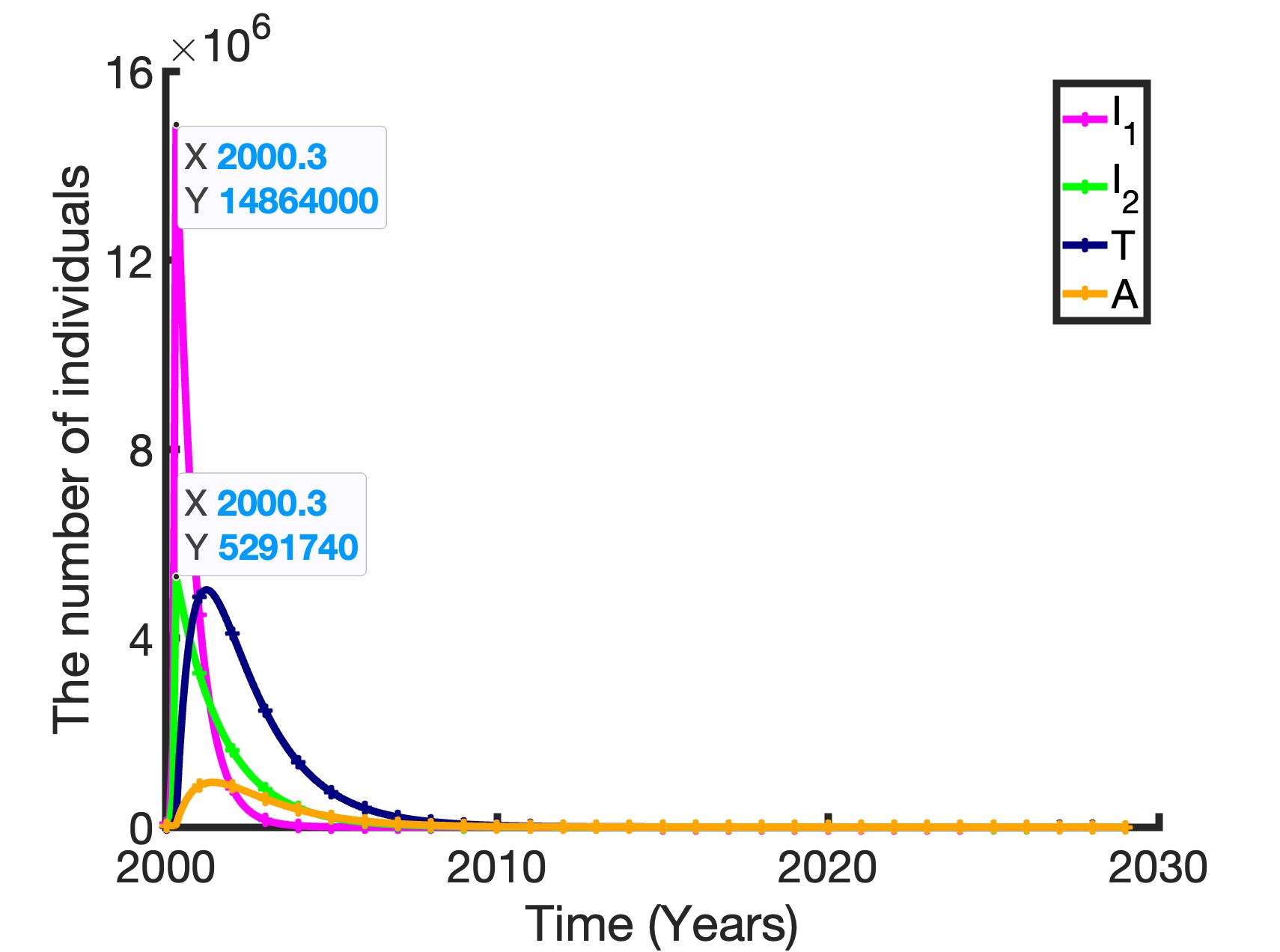}\label{fig:gamma999}}
\caption{The effect of $\gamma$, the rate at which infected individuals who are unconscious become aware.}
\label{fig:gamma-effects}
\end{figure}

\noindent According to figure \ref{fig:stability_endemic}, if the existing control measures are not altered, HIV outbreaks will continue to occur and will gradually get closer to the point where the endemic equilibrium is reached. The percentage of latent persons who transition to the class of aware individuals, denoted as $p_{max}$, and the rate at which unaware infected people become aware, denoted as $\gamma$, are found to have a significant influence on the reduction of the number of individuals who are infected. There is an increase in the number of people receiving antiretroviral therapy (ART) as more people who are HIV-positive are screened or become aware of their condition. Consequently, this results in a substantial decrease in the number of people who are infected with HIV, which in turn lessens the stress placed on the healthcare system \cite{granich2009universal}. In order to raise awareness about HIV infection, it is vital to increase the number of people who are tested for HIV and to improve the media efforts that are already being run.

\section{Optimal Control and Discussion}
\label{sec:optimal}
Through the use of a variety of control mechanisms, we expand our model. The existence, boundedness, and uniqueness of the optimal solution are demonstrated, and then the optimal solution is computed using Pontryagin's Maximum Principle \cite{pontryagin2018mathematical} by applying this principle. In addition, we evaluate the effectiveness of the ways for implementing control as well as the expenses associated with them, and we provide the control approach that is the most cost-effective.

\subsection{Model of Control}
The following model \eqref{Model_3} comprises six state variables: \(S, E, I_1, I_2, T, \text{ and } A\), along with three control variables: \(v_1(t), v_2(t), \text{ and } v_3(t)\), where, \( v_1(t) \) symbolizes the level of screening for latent individuals, \( v_2(t) \) depicts the level of education for people who are unaware of being affected, and \( v_3(t) \) constitutes the therapy for persons who are aware that they are infected. The control set is:
\[
\Omega = \{ (v_1, v_2, v_3) | v_i(t) \in L^{\infty} [0, t_f], 0 \leq v_i(t) \leq c_i, 0 < c_i \leq 1, i = 1,2,3 \},
\]
where \( t_f \) is the point at which the implementation of controls comes to an end.
The following is a description of the optimum control model:
\begin{equation}\label{Model_3}
%\left\{
\begin{cases}
\begin{aligned}
\frac{dS}{dt} &= \kappa - (\alpha \epsilon E + \alpha I_2)S - \mu S \\
\frac{dE}{dt} &= (\alpha \epsilon E + \alpha I_2)S - (\beta + \mu) E - v_1(t) E \\
\frac{dI_1}{dt} &= p \beta E + \gamma I_2 - (\mu + \psi + \delta_1) I_1 + v_1(t) E + v_2(t) I_2 - v_3(t) I_1 \\
\frac{dI_2}{dt} &= (1 - p) \beta E - (\gamma + \mu + \delta_2) I_2 - v_2(t) I_2 \\
\frac{dT}{dt} &= \psi I_1 - (\mu + \xi) T + v_3(t) I_1 \\
\frac{dA}{dt} &= \delta_1 I_1 + \delta_2 I_2 + \xi T - \mu_0 A - \mu A.
\end{aligned}
\end{cases}
%\right.
%\tag{6}
\end{equation}
The initial conditions satisfy,
\begin{equation}
S(0) \geq 0, \, E(0) \geq 0, \, I_1(0) \geq 0, \, I_2(0) \geq 0, \, T(0) \geq 0, \, A(0) \geq 0. \label{equ:initial}    %\tag{7} 
\end{equation}
\noindent In addition to lowering the costs that are connected with the adoption of screening, education, and treatment programmes, our goal is to cut down on the number of people who are infected with the virus. As a result, the objective function may be expressed here:
\[
K(v_1, v_2, v_3) = \int_{0}^{t_f} \left[ p_1 E + p_2 I_1 + p_3 I_2 + \frac{q_1}{2} v_1^2(t) + \frac{q_2}{2} v_2^2(t) + \frac{q_3}{2} v_3^2(t) \right] dt, \label{equ:objective function} %\tag{8} 
\]
where \( p_1 \), \( p_2 \), and \( p_3 \) indicate the relative weights of the numbers of persons who are latent, individuals who are aware of their infection, and those who are unaware of of their infection, respectively. The weights \( q_1 \), \( q_2 \), and \( q_3 \) evaluate the expenses associated with control variables. \( v_1 \), \( v_2 \), and \( v_3 \), respectively.

Initially, we establish the presence of an optimum solution.
\begin{theorem}
An optimal solution \(v^* = \{v_1^*, v_2^*, v_3^*\}\) exists within \(\Omega\) for the objective function \(K(v)\), such that:
\[
K(v^*) = K(v_1^*, v_2^*, v_3^*) = \min_{v \in \Omega} K(v_1, v_2, v_3).
\]
\end{theorem}
\begin{proof}
The findings demonstrate that there is a strategy \( v^* \) that is optimum. Closed and convex are the two characteristics that define the control set \( \Omega \). The function \( K \) is concave on the set \( \Omega \) when integrated. The control system is bounded, indicating that the optimal control is compact. To add to this, \( \zeta > 1 \), as well as positive values \( D_1 \) and \( D_2 \), exist such that:
\[
K(v_1, v_2, v_3) \geq D_1 \left( |v_1(t)|^2 + |v_2(t)|^2 + |v_3(t)|^2 \right)^{\frac{\zeta}{2}} - D_2,
\]
which demonstrates the existence of the optimum control.
\end{proof}

\begin{theorem} In model \eqref{Model_3}, when the control variables \((v_1^*, v_2^*, v_3^*)\) ensure that \(H(t, \phi, v^*) < H(t, \phi, v)\), then the state variables \(S^{**}\), \(E^{**}\), \(I_1^{**}\), \(I_2^{**}\), and \(A^{**}\) represent the solutions. Accordingly, adjoint variables \(\lambda_i(t)\) for \(i = 1, 2, \dots, 6\) are confirmed to satisfy the conditions:
\begin{align}\label{model_4}
\begin{cases}
lambda_1' &= \lambda_1(\alpha E + \alpha I_2 + \mu) - \lambda_2(\alpha \epsilon E + \alpha I_2), \\
\lambda_2' &= -p_1 + \lambda_1 \epsilon \alpha S - \lambda_2(\alpha \epsilon S - (\beta + \mu) - u_1(t)) - \lambda_3(p \beta + u_1(t)) - \lambda_4(1 - p)\beta, \\
\lambda_3' &= -p_2 + \lambda_3(\mu + \psi + \delta + u_3(t)) - \lambda_5(\psi + u_3(t)) - \lambda_6\delta, \\
\lambda_4' &= -p_3 + \lambda_1\alpha S - \lambda_2\alpha S - \lambda_3(\gamma + u_2(t)) + \lambda_4(\delta + \gamma + \mu + u_2(t)) - \lambda_6\delta, \\
\lambda_5' &= \lambda_5(\xi + \mu) - \lambda_6\xi, \\
\lambda_6' &= \lambda_6(\mu_0 + \mu). %\tag{9}
\end{cases}
\end{align}
The boundary conditions are
\[
\lambda_i(t_f) = 0, \quad i = 1, 2, \ldots, 6.
\]
The optimal controls $v_1^*$, $v_2^*$, and $v_3^*$ are
\[
\begin{aligned}
\begin{cases}
v_1^*(t) &= \min\{\max\{0, v_1^c\}, 1\}, \\
v_2^*(t) &= \min\{\max\{0, v_2^c\}, 1\}, \\
v_3^*(t) &= \min\{\max\{0, v_3^c\}, 1\}, 
\end{cases}
\end{aligned}
\]
where,
\[
\begin{aligned}
v_1^c &= \frac{(\lambda_2 - \lambda_3)E^{**}}{q_1}, \quad
v_2^c = \frac{(\lambda_4 - \lambda_3)I_2^{**}}{q_2}, \quad
v_3^c = \frac{(\lambda_3 - \lambda_5)I_1^{**}}{q_3}.
\end{aligned}
\]
We found the optimal solution:
\[
v_i^* = 
\begin{cases}
0, & \text{if } v_i^c \leq 0, \\
v_i^c, & \text{if } 0 < v_i^c < 1, \\
1, & \text{if } v_i^c \geq 1,
\end{cases}
\]
where, $i = 1, 2, 3.$
\end{theorem}
\begin{proof}
\noindent To identify the optimal solution, we will construct the Lagrange function \(L\) and the Hamiltonian function \(H\).
\[
L(t, \phi, v) = p_1 E + p_2 I_1 + p_3 I_2 - \frac{q_1}{2} v_1^2(t) - \frac{q_2}{2} v_2^2(t) - \frac{q_3}{2} v_3^2(t),   
\]
\[
\begin{aligned}
H(t, \phi, v) &= p_1 E + p_2 I_1 + p_3 I_2 + \frac{q_1}{2} v_1^2(t) + \frac{q_2}{2} v_2^2(t) + \frac{q_3}{2} v_3^2(t) \\
&\quad + \lambda_1 (\kappa - (\alpha E + \alpha I_2)S - \mu S) + \lambda_2 ((\alpha \epsilon E + \alpha I_2)S - (\beta + \mu)E - u_1(t)E) \\
&\quad + \lambda_3 (p \beta E + \gamma I_2 - (\mu + \psi + \delta)I_1 + u_1(t)E + u_2(t)I_2 - u_3(t)I_1) \\
&\quad + \lambda_4 ((1 - p) \beta E - (\gamma + \delta + \mu)I_2 - u_2(t)I_2) \\
&\quad + \lambda_5 (\psi I_1 - (\xi + \mu)T + u_3(t)I_1) + \lambda_6 (\delta I_1 + \delta I_2 + \xi T - \mu_0 A - \mu A),
\end{aligned}
\]

\noindent where \(\phi = (S, E, I_1, I_2, T, A)^T\), \(v = (v_1, v_2, v_3)^T\), and \(\lambda_i(t)\), \(i = 1, 2, \ldots, 6\), are adjoint variables.\\
The optimal answer is then found by using Pontryagin's Maximum Principle \cite{pontryagin2018mathematical}. We express the system of differential equations for the Hamiltonian function \(H\) as follows, taking into account the principle's guarantee of the availability of optimal control solutions:

\[
\begin{aligned}
\frac{d\lambda_1}{dt} &= -\frac{\partial H}{\partial S} = \lambda_1(\alpha \epsilon E + \alpha I_2 + \mu) - \lambda_2(\alpha \epsilon E + \alpha I_2), \\
\frac{d\lambda_2}{dt} &= -\frac{\partial H}{\partial E} = -p_1 + \lambda_1\alpha \epsilon S - \lambda_2(\alpha \epsilon S - (\beta + \mu) - u_1(t))- \lambda_3 (p \beta + u_1(t)) \\
&\quad - \lambda_4(1 - p)\beta, \\
\frac{d\lambda_3}{dt} &= -\frac{\partial H}{\partial I_1} = -p_2 + \lambda_3(\mu + \psi + \delta + u_3(t)) - \lambda_5(\psi + u_3(t)) - \lambda_6\delta, \\
\frac{d\lambda_4}{dt} &= -\frac{\partial H}{\partial I_2} = -p_3 + \lambda_1\alpha S - \lambda_2\alpha S - \lambda_3(\gamma + u_2(t)) + \lambda_4(\delta + \gamma + \mu + u_2(t))\\
&\quad  - \lambda_6\delta, \\
\frac{d\lambda_5}{dt} &= -\frac{\partial H}{\partial T} = \lambda_5(\xi + \mu) - \lambda_6\xi, \\
\frac{d\lambda_6}{dt} &= -\frac{\partial H}{\partial A} = \lambda_6(\mu_0 + \mu),
\end{aligned}
\]
and $\lambda_i(t_f) = 0, i = 1, 2, \ldots, 6.$
After that, we can derive the partial derivatives of the Hamiltonian function with respect to the control variables \( v_1, v_2, \) and \( v_3 \) in the control set \( \Omega \) as follows:
\[
\frac{\partial H}{\partial v_1} = q_1 v_1(t) + (\lambda_3 - \lambda_2)E^{**},
\]
\[
\frac{\partial H}{\partial v_2} = q_2 v_2(t) + (\lambda_3 - \lambda_4)I_2^{**},
\]
\[
\frac{\partial H}{\partial v_3} = q_3 v_3(t) + (\lambda_5 - \lambda_3)I_1^{**}.
\]
Let \( \frac{\partial H}{\partial v_1} = 0, \frac{\partial H}{\partial v_2} = 0, \) and \( \frac{\partial H}{\partial v_3} = 0, \) we obtain,
\[
\begin{aligned}
v_1^*(t) &= \frac{(\lambda_2 - \lambda_3)E^{**}}{q_1}, \\
v_2^*(t) &= \frac{(\lambda_4 - \lambda_3)I_2^{**}}{q_2}, \\
v_3^*(t) &= \frac{(\lambda_3 - \lambda_5)I_1^{**}}{q_3}.
\end{aligned}
\]
Since the control variable \( 0 \leq v_i \leq 1, i = 1, 2, 3, \) the optimal controls are:
\[
\begin{aligned}
v_1^*(t) &= \min \left\{ \max \left\{ 0, \frac{(\lambda_2 - \lambda_3)E^{**}}{q_1} \right\}, 1 \right\}, \\
v_2^*(t) &= \min \left\{ \max \left\{ 0, \frac{(\lambda_4 - \lambda_3)I_2^{**}}{q_2} \right\}, 1 \right\}, \\
v_3^*(t) &= \min \left\{ \max \left\{ 0, \frac{(\lambda_3 - \lambda_5)I_1^{**}}{q_3} \right\}, 1 \right\}.
\end{aligned}
\]
\end{proof}
\noindent Using the method explained in references \cite{ghosh2021mathematical} and \cite{roy2021control}, it has been demonstrated that the solution to model (\ref{Model_3}), satisfying the initial condition \eqref{equ:initial}, remains within a finite range.
\begin{theorem}
The solution set \((S(t), E(t), I_1(t), I_2(t), T(t), A(t))\) for model (\ref{Model_3}) is absolutely continuous over the complete interval \([0, t_f]\) for all permissible control variables \((v_1(t), v_2(t), v_3(t))\). These state variables \(S(t), E(t), I_1(t), I_2(t), T(t), A(t)\) adhere to the following inequalities:
\begin{align*}
0 \leq S(t) &\leq \frac{\kappa}{\mu}, \quad 0 \leq E(t) \leq \frac{\kappa}{\mu}, \quad 0 \leq I_1(t) \leq \frac{\kappa}{\mu}, \quad 0 \leq I_2(t) \leq \frac{\kappa}{\mu}, \\
0 \leq T(t) &\leq \frac{\kappa}{\mu}, \quad 0 \leq A(t) \leq \frac{\kappa}{\mu}.
\end{align*}
\end{theorem} 
\begin{proof}
We will assume that the solution \((S(t), E(t), I_1(t), I_2(t), T(t), A(t))\) for model (\ref{Model_3}) is defined over the interval \([0, t^*]\), representing the maximal interval of existence for this model. For the sake of simplicity, let's assume that, $t^* \leq t_f$. The solutions $(S(t), E(t), I_1(t),\\ I_2(t), T(t), A(t))$ are not negative according to model (\ref{Model_3}) and the condition $0 \leq v_i \leq 1$ for $i = 1, 2, 3$. 
Through the use of Lemma \ref{Lem1}, we may concurrently ascertain the maximum value of the solution for model (\ref{Model_3}), which are:
\[
0 \leq S(t), E(t), I_1(t), I_2(t), T(t), A(t) \leq \frac{\kappa}{\mu}.
\]
As a result, the responses of model (\ref{Model_3}) that meet the initial condition \eqref{equ:initial} are considered to be bounded.
\end{proof} 
\noindent We will utilize the technique outlined in reference \cite{roy2015effect} to showcase the distinctiveness of the optimum control.
\begin{theorem}
The optimal system (\ref{Model_3}) possesses a unique solution within the adequately brief interval \([t_0, t_f]\).
\end{theorem}

\begin{proof}
We suppose that $(S, E, I_1, I_2, T, A, \lambda_1, \lambda_2, \lambda_3, \lambda_4, \lambda_5, \lambda_6)$ 
and $(\hat{S}, \hat{E}, \hat{I}_1, \hat{I}_2, \hat{T}, \hat{A}, \hat{\lambda}_1, \hat{\lambda}_2,\\ \hat{\lambda}_3, \hat{\lambda}_4, \hat{\lambda}_5, \hat{\lambda}_6)$ 
are two solutions to models (\ref{Model_3}) and (\ref{model_4}). 
Let $S = e^{\phi t} m_1, E = e^{\phi t} m_2, I_1 = e^{\phi t} m_3, I_2 = e^{\phi t} m_4, T = e^{\phi t} m_5, A = e^{\phi t} m_6, \lambda_1 = e^{-\phi t} n_1, \lambda_2 = e^{-\phi t} n_2, \lambda_3 = e^{-\phi t} n_3, \lambda_4 = e^{-\phi t} n_4, \lambda_5 = e^{-\phi t} n_5, \lambda_6 = e^{-\phi t} n_6, \hat{S}=e^{\phi t} \hat{m_1}, \hat{E}=e^{\phi t} \hat{m_2}, \hat{I_1}=e^{\phi t} \hat{m_3}, \hat{I_2}=e^{\phi t} \hat{m_4}, \hat{T}=e^{\phi t} \hat{m_5}, \hat{A}=e^{\phi t} \hat{m_6}, \hat{\lambda_1} = e^{-\phi t} \hat{q_1}, \hat{\lambda_2} = e^{-\phi t} \hat{q_2}, \hat{\lambda_3} = e^{-\phi t} \hat{q_3}, \hat{\lambda_4} = e^{-\phi t} \hat{q_4}, \hat{\lambda_5} = e^{-\phi t} \hat{n_5}, \hat{\lambda_6} = e^{-\phi t} \hat{n_6}$, 
where $\phi > 0$ is to be determined, and also $m_i, n_i, \hat{m}_i$ and $\hat{n}_i$ ($i = 1, 2, \ldots, 6$) are variables about $t$.
On top of that, we have:
\[
v_1^*(t) = \min \left\{ \max \left[ 0, \frac{(n_2 - n_3)m_2}{q_1} \right], 1 \right\}
\]
\[
v_2^*(t) = \min \left\{ \max \left[ 0, \frac{(n_4 - n_3)m_4}{q_2} \right], 1 \right\}
\]
\[
v_3^*(t) = \min \left\{ \max \left[ 0, \frac{(n_3 - n_5)m_3}{q_3} \right], 1 \right\}
\]
and
\[
\hat{v}_1^*(t) = \min \left\{ \max \left[ 0, \frac{(\hat{n}_2 - \hat{n}_3)\hat{m}_2}{q_1} \right], 1 \right\}
\]
\[
\hat{v}_2^*(t) = \min \left\{ \max \left[ 0, \frac{(\hat{n}_4 - \hat{n}_3)\hat{m}_4}{q_2} \right], 1 \right\}
\]
\[
\hat{v}_3^*(t) = \min \left\{ \max \left[ 0, \frac{(\hat{n}_3 - \hat{n}_5)\hat{m}_3}{q_3} \right], 1 \right\}
\]
Following that, we will replace the values of $S, E, I_1, I_2, T, A, \lambda_1, \lambda_2, \lambda_3, \lambda_4, \lambda_5, \lambda_6$ into \eqref{Model_3} and \eqref{model_4}. Then, we have,
\begin{align*}
m_1' + \phi m_1 &= \kappa e^{-\phi t} - (\alpha \epsilon m_2 + \alpha m_4)m_1e^{\phi t} - \mu m_1, \\
m_2' + \phi m_2 &= (\alpha \epsilon m_2 + \alpha m_4)m_1e^{\phi t} - (\beta + \mu)m_2 - \frac{(n_2 - n_3)m_2}{q_1}m_2, \\
m_3' + \phi m_3 &= m \alpha m_2 + \gamma m_4 - (\mu + \psi + \delta)m_3 + \frac{(n_2 - n_3)m_2}{q_1}m_2 + \frac{(n_4-n_3)m_4}{q_2}m_4 - \frac{(n_3-n_5)m_3}{q_1}m_3, \\
m_4' + \phi m_4 &= (1 - m) \beta m_2 - (\delta + \gamma + \mu)m_4 - \frac{(n_4 - n_3)m_4}{q_2}m_4, \\
m_5' + \phi m_5 &= \psi m_3 - (\xi + \mu)m_5 + \frac{(n_3 - n_5)m_3}{q_1}m_3, \\
m_6' + \phi m_6 &= \delta m_3 + \delta m_4 + \xi m_5 - \mu_0 m_6 - \mu m_6, \\
n_1' - \phi n_1 &= n_1(\alpha \epsilon m_2e^{\phi t} + \alpha m_4e^{\phi t} + \mu) - n_2e^{\phi t}(\alpha m_2 + \alpha m_4), \\
n_2' - \phi n_2 &= -p_1e^{\phi t} + n_1m_1\alpha \epsilon e^{\phi t} - n_2(\alpha \epsilon m_1e^{\phi t} - (\beta + \mu) - \frac{(n_2 - n_3)m_2}{q_1}m_2) \\
&\quad - n_3(m\beta + \frac{(n_3 - n_5)m_2}{q_1}m_2) - n_4(1 - m)\beta, \\
n_3' - \phi n_3 &= -p_2e^{\phi t} + n_3(\mu + \psi + \delta + \frac{(n_3 - n_5)m_3}{q_1}m_3) - n_5(\psi + \frac{(n_3 - n_5)m_3}{q_1}m_3) - \delta n_6, \\
n_4' - \phi n_4 &= -p_3e^{\phi t} - n_3(\gamma + \frac{(n_4 - n_3)m_4}{q_2} m_4) + n_4(\delta + \gamma + \mu + \frac{(n_4 - n_3)m_4}{q_2}m_4) - \delta n_6, \\
n_5' - \phi n_5 &= n_5(\xi + \mu) - \eta n_6, \\
n_6' - \phi n_6 &= n_6(\mu_0 + \mu).
\end{align*} 
The subtraction of the equations for \(S\) and \(\hat{S}\), \(E\) and \(\hat{E}\), \(I_1\) and \(\hat{I}_1\), \(I_2\) and \(\hat{I}_2\), \(T\) and \(\hat{T}\), \(A\) and \(\hat{A}\) are performed. Next, each equation is multiplied by a function that corresponds to it, and the resulting product is integrated across the time range ranging from \(t_0\) to \(t_f\). In light of this,
\begin{align*}
\int_{t_0}^{t_f} (v_1^* - \hat{v}_1^*)^2 \, dt &\leq B_1 e^{2\phi t_f} \int_{t_0}^{t_f} \left[ |m_2 - \hat{m}_2|^2 + |n_2 - \hat{n}_2|^2 + |n_3 - \hat{n}_3|^2 \right] \, dt, \\
\int_{t_0}^{t_f} (v_2^* - \hat{v}_2^*)^2 \, dt &\leq B_2 e^{2\phi t_f} \int_{t_0}^{t_f} \left[ |m_4 - \hat{m}_4|^2 + |n_3 - \hat{n}_3|^2 + |n_4 - \hat{n}_4|^2 \right] \, dt, \\
\int_{t_0}^{t_f} (v_3^* - \hat{v}_3^*)^2 \, dt &\leq B_3 e^{2\phi t_f} \int_{t_0}^{t_f} \left[ |m_3 - \hat{m}_3|^2 + |n_5 - \hat{n}_5|^2 + |n_5 - \hat{n}_5|^2 \right] \, dt,
\end{align*}
where $B_1$, $B_2$, and $B_3$ are constants.
Therefore, we obtain,
\begin{align*}
&\frac{1}{2}(m_1 - \hat{m}_1)^2(t_f) + \phi \int_{t_0}^{t_f} |m_1 - \hat{m}_1|^2 \mathrm{d}t \\
&\leq \alpha \epsilon e^{-\phi t_f} \int_{t_0}^{t_f} \left[|m_1 - \hat{m}_1|^2 + |m_2 - \hat{m}_2|^2\right] \mathrm{d}t + \alpha e^{-\phi t_f} \int_{t_0}^{t_f} |m_1 - \hat{m}_1|^2 \mathrm{d}t \\
&\quad + |m_4 -\hat{m}_4|^2 \mathrm{d}t + \mu \int_{t_0}^{t_f} |m_1 - \hat{m}_1|^2 \mathrm{d}t,
\end{align*}
and
\begin{align*}
&\frac{1}{2}(n_1 - \hat{n}_1)^2(t_f) + \phi \int_{t_0}^{t_f} |n_1 - \hat{n}_1|^2 \mathrm{d}t \leq \alpha \epsilon e^{2\phi t_f} \int_{t_0}^{t_f} \left[|m_2 - \hat{m}_2|^2 + |n_1-\hat{n}_1|^2\right] \mathrm{d}t \\
&+ \alpha e^{2\phi t_f} \int_{t_0}^{t_f} \left[|m_4 - \hat{m}_4|^2 + |n_1 - \hat{n}_1|^2\right] \mathrm{d}t + \alpha e^{2\phi t_f} \int_{t_0}^{t_f} \left[|m_2 - \hat{m}_2|^2 + |n_2 - \hat{n}_2|^2\right] \mathrm{d}t \\
&+ \alpha e^{2\phi t_f} \int_{t_0}^{t_f} \left[|m_4 - \hat{m}_4|^2 + |n_1 - \hat{n}_1|^2\right] \mathrm{d}t + \mu e^{\phi t_f} \int_{t_0}^{t_f} \left[|m_1 - \hat{m}_1|^2 + |n_1 - \hat{n}_1|^2\right] \mathrm{d}t. 
\end{align*}
Similarly, we derive the inequalities for $m_2$ and $\hat{m}_2$, $m_3$ and $\hat{m}_3$, $m_4$ and $\hat{m}_4$, $m_5$ and $\hat{m}_5$, $m_6$ and $\hat{m}_6$, $n_2$ and $\hat{n}_2$, $n_3$ and $\hat{n}_3$, $n_4$ and $\hat{n}_4$, $n_5$ and $\hat{n}_5$, $n_6$ and $\hat{n}_6$. Summing all the inequalities, we have,
\begin{align*}
	&\frac{1}{2} \big[(m_1 - \hat{m}_1)^2(t_f) + (m_2 - \hat{m}_2)^2(t_f) + (m_3 - \hat{m}_3)^2(t_f) + (m_4 - \hat{m}_4)^2(t_f) + (m_5 - \hat{m}_5)^2(t_f) \\
	& + (m_6 - \hat{m}_6)^2(t_f) + (n_1 - \hat{n}_1)^2(t_0) + (n_2 - \hat{n}_2)^2(t_0) + (n_3 - \hat{n}_3)^2(t_0) + (n_4 - \hat{n}_4)^2(t_0) \\
	& + (n_5 - \hat{n}_5)^2(t_0) + (n_6 - \hat{n}_6)^2(t_0) \big] + \phi \int_{t_0}^{t_f} \left[ |m_1 - \hat{m}_1|^2 + |m_2 - \hat{m}_2|^2 + |m_3 - \hat{m}_3|^2 \right. \\
	& \left. + |m_4 - \hat{m}_4|^2 + |m_5 - \hat{m}_5|^2 + |m_6 - \hat{m}_6|^2 + |n_1 - \hat{n}_1|^2 + |n_2 - \hat{n}_2|^2 + |n_3 - \hat{n}_3|^2 \right. \\
	& \left. + |n_4 - \hat{n}_4|^2 + |n_5 - \hat{n}_5|^2 + |n_6 - \hat{n}_6|^2 \right] \,dt \\
	& \leq (\tilde{K_1} + \tilde{K_2} e^{3 \phi t_f}) \int_{t_0}^{t_f} \left[ |m_1 - \hat{m}_1|^2 + |m_2 - \hat{m}_2|^2 + |m_3 - \hat{m}_3|^2 + |m_4 - \hat{m}_4|^2 + |m_5 - \hat{m}_5|^2 \right. \\
	& \left. + |m_6 - \hat{m}_6|^2 + |n_1 - \hat{n}_1|^2 + |n_2 - \hat{n}_2|^2 + |n_3 - \hat{n}_3|^2 + |n_4 - \hat{n}_4|^2 + |n_5 - \hat{n}_5|^2 + |n_6 - \hat{n}_6|^2 \right] \,dt.
\end{align*}
From the above equation, we have,
\begin{align*}
    &(\phi - \tilde{K}_1 - \tilde{K}_2 e^{-3\phi t_f}) \int_{t_0}^{t_f} \big[|m_1 - \hat{m}_1|^2 + |m_2 - \hat{m}_2|^2 + |m_3 - \hat{m}_3|^2 + |m_4 - \hat{m}_4|^2 + |m_5 - \hat{m}_5|^2 \\
    &+ |m_6 - \hat{m}_6|^2 + |n_1 - \hat{n}_1|^2 + |n_2 - \hat{n}_2|^2 + |n_3 - \hat{n}_3|^2 + |n_4 - \hat{n}_4|^2 + |n_5 - \hat{n}_5|^2 + |n_6 - \hat{n}_6|^2 ] \,dt.\\
    &\leq 0,
\end{align*}
The coefficients and the bounds on $m_i$ and $n_i$ are dependent on the coefficients and the bounds on $\tilde{K}_1$ and $\tilde{K}_2$, where $i$ ranges from 1 to 6.\\
We choose $\phi$ such that $\phi > \tilde{K}_1 + \tilde{K}_2$ and $t_f < \frac{1}{3} \ln\left(\frac{\phi - \tilde{K}_1}{\tilde{K}_2}\right)$, then $m_1 = \hat{m}_1, m_2 = \hat{m}_2, m_3 = \hat{m}_3, m_4 = \hat{m}_4, m_5 = \hat{m}_5, m_6 = \hat{m}_6, n_1 = \hat{n}_1, n_2 = \hat{n}_2, n_3 = \hat{n}_3, n_4 = \hat{n}_4, n_5 = \hat{n}_5, n_6 = \hat{n}_6. $ \\
It may be concluded that model (\ref{Model_3}) is capable of providing a single optimal solution within a short period of time.
\end{proof}

\subsection{Control Strategies: Data Analysis and  Numerical Simulations}
In order to demonstrate that the theoretical conclusions and control methods are feasible, numerical simulations are carried out in order to investigate the influence that the optimum control technique has on the transmission of HIV.
\begin{figure}[H]
  \centering
  \includegraphics[width=0.6\textwidth]{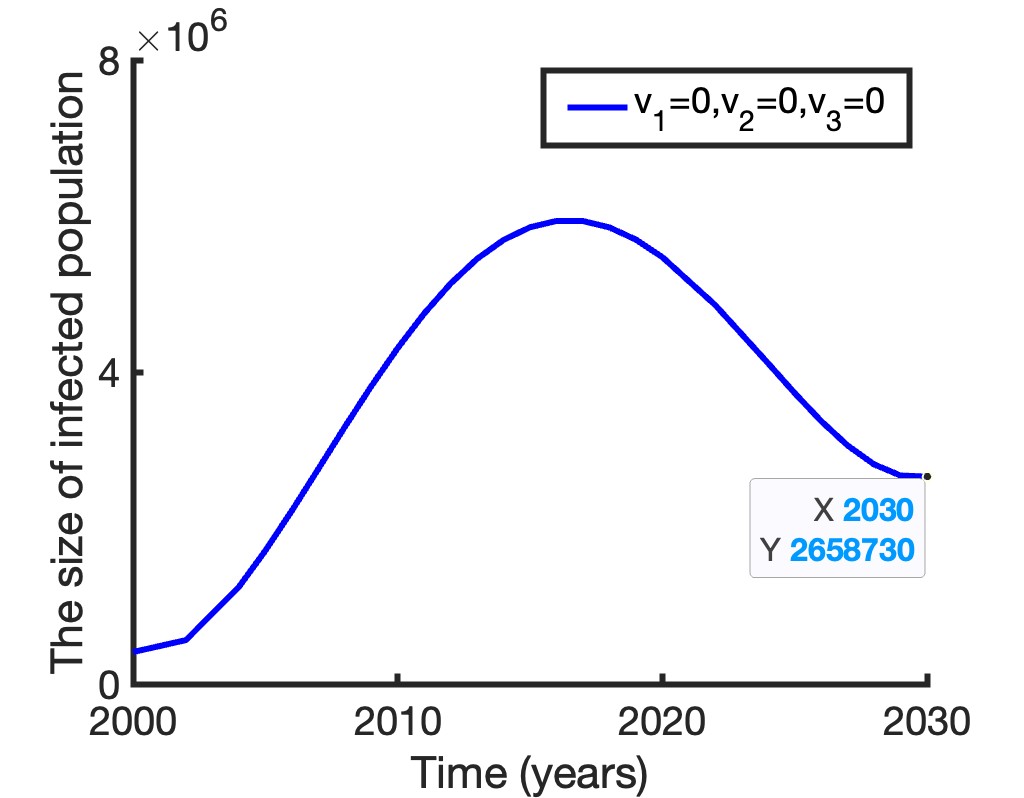}
  \caption{The amount of people with infections when $v_1=0, v_2=0$ and $v_3=0$.}
  \label{fig:v=0}
\end{figure}
\noindent In the Hamiltonian function, the weights assigned to state variables and control variables are set as follows: $p_1 = 20, p_2 = 30, p_3 = 30, q_1 = 45, q_2 = 25,$ and $q_3 = 10$. On the basis of the three control variables, $v_1(t),\; v_2(t)$, and $v_3(t)$, we offer the four control techniques that are listed below.\\
In order to evaluate the effectiveness of the control following the implementation of the control measures, we first simulate the amount of individuals who are infected but do not implement any control measures., as seen in figure \ref{fig:v=0}. Infected persons can be classified as either latent, aware, or unaware.The data shown in figure \ref{fig:v=0} demonstrates that the number of infected persons begins to rise from 2002 to 2017, but then begins to fall from 2017 to 2030. The number of people that are fatigued has expected to reach 26,58,730 by the end of the year 2030. Next, we will investigate the influence that control techniques have on the spread of HIV.\\

\noindent\textbf{Strategy 01: $v_1=0,\;v_2 \neq 0$ and $v_3 \neq 0$}\\
It is being proposed to implement a mix of education for individuals who are infected $v_2(t)$ and treatment for infected individuals who are aware of their infection $v_3(t)$. That is to say, $v_1=0,\;v_2 \neq 0$ and $v_3 \neq 0$. Following that, to analyze control effects, we shall optimize the cost function $K$ in equation \eqref{equ:objective function}.
\begin{figure}[H]
\centering
\subfigure[Variation in the optimal control solutions $v_2$ and $v_3$.]{\includegraphics[width=0.45\linewidth]{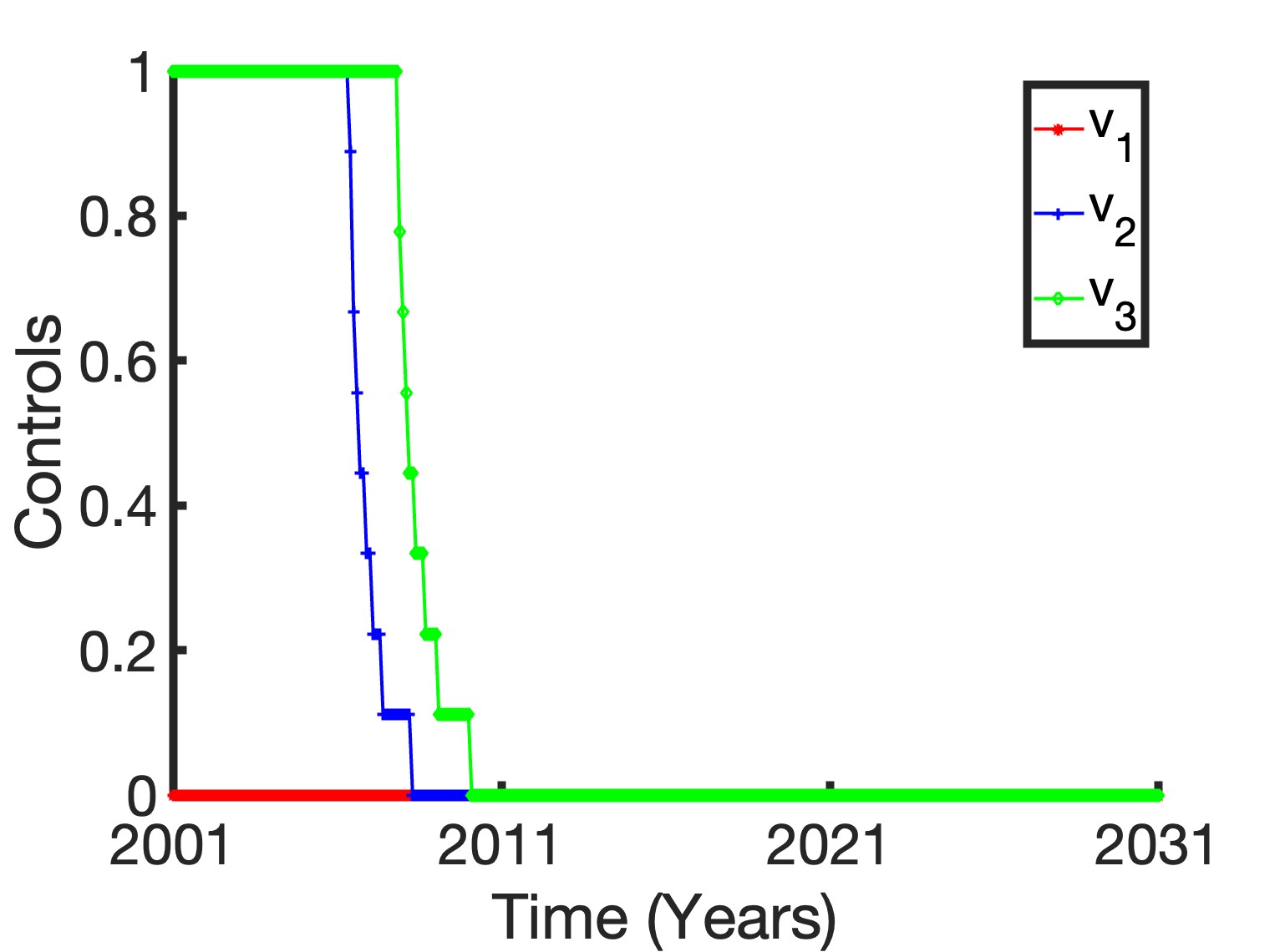}\label{fig:12a}}
\subfigure[Total size of infected population under controls $v_2$ and $v_3$.]{\includegraphics[width=0.45\linewidth]{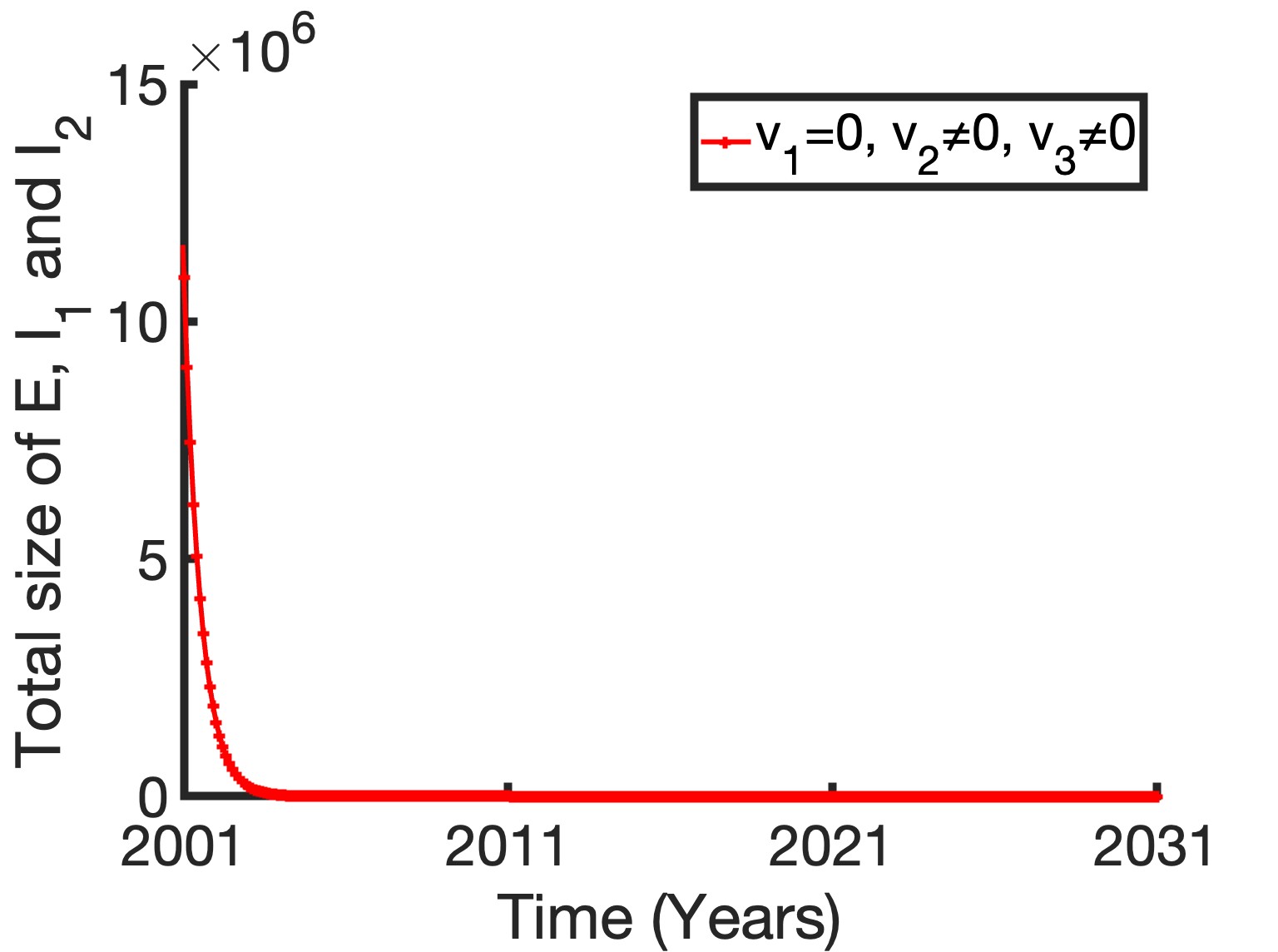}\label{fig:12b}}
\caption{Effects of control strategy implementation 1.}
\label{fig:12}
\end{figure}

\noindent Control techniques $v_2$ and $v_3$ are shown to have their effects depicted in figure \ref{fig:12a}. Both the controls $v_3$ and $v_2$ have reached the maximum control level of 100\% for a timeframe of 8 to 9 years, respectively, at the early stage of the control process. Then both of these drop gradually. Figure \ref{fig:12b} illustrates the impact that controls $v_2$ and $v_3$ have on the total number of persons who are living with the infection. When compared to the scenario where there was no control, the findings indicate that the overall amount of latent, aware, and unaware infected persons is significantly lower when the controls $v_2$ and $v_3$ are implemented (see figure \ref{fig:v=0}). Furthermore, with the implementation of controls $v_2$ and $v_3$, we discover that HIV is completely eradicated in the year 2004. Based on this evidence, it may be concluded that these restrictions are highly successful in reducing the spread of HIV.\\

\noindent\textbf{Strategy 02: $v_1\neq 0,\;v_2=0$ and $v_3\neq 0$}\\
The strategy under consideration involves the utilization of both screening for latent persons $v_1(t)$ and therapies for aware infected individuals $v_3(t)$. To clarify, the values are 
$v_1\neq 0,\;v_2=0$ and $v_3\neq 0$. Subsequently, the cost function $K$ is optimized in equation \eqref{equ:objective function} so that we may investigate the impact of the control.
\begin{figure}[H]
\centering
\subfigure[Variation in the optimal control solutions $v_1$ and $v_3$.]{\includegraphics[width=0.45\linewidth]{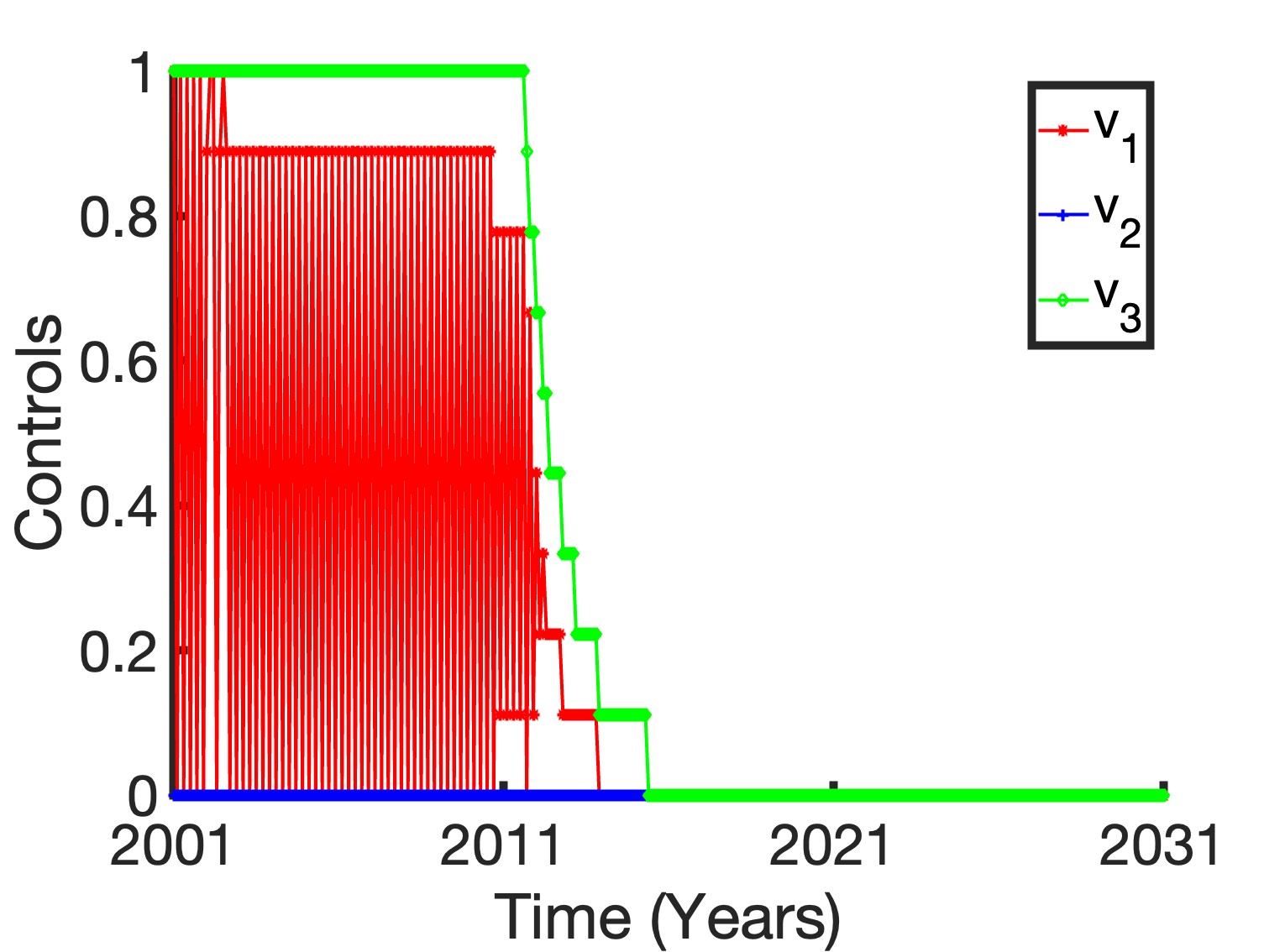}\label{fig:13a}}
\subfigure[Total size of infected population under controls $v_1$ and $v_3$.]{\includegraphics[width=0.45\linewidth]{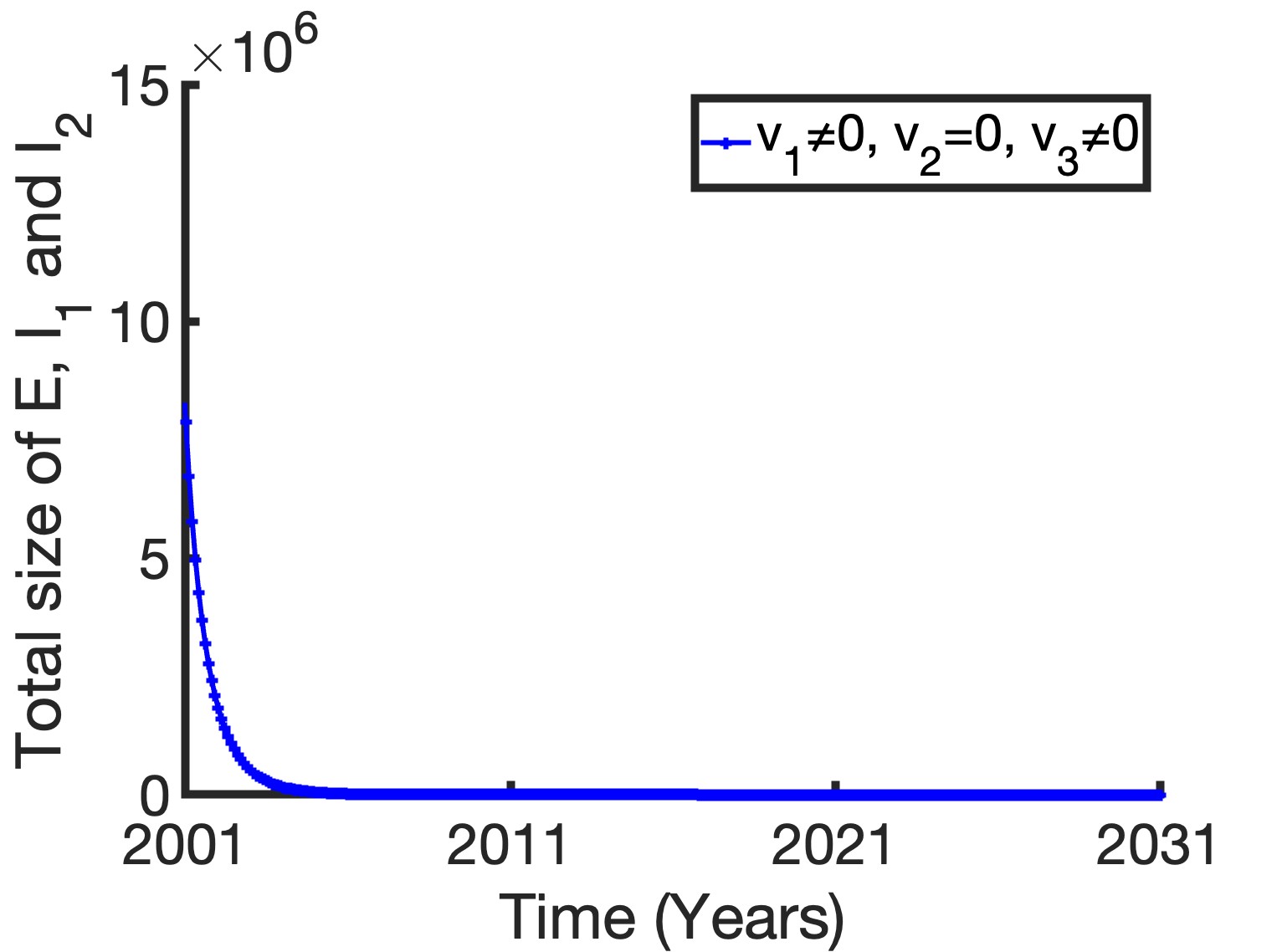}\label{fig:13b}}
\caption{Effects of control strategy implementation 2.}
\label{fig:13}
\end{figure}

\noindent The effects of controls $v_1$ and $v_3$ throughout the course of time are depicted in figure \ref{fig:13a}. As can be seen in the image, control $v_3$ has been operating at the highest possible control level of 100\% for a period of 12 years. After that, these control measure steadily drop.For the control $v_1$, it operates at its 100\% for first 2 years and then drops to 95\%. It continues to be like this for almost 8 years and then again drops to a level of 20\%. With a little fluctuation it is in action for further 2/3 years. In the end, the values of the controls $v_1$ and $v_3$ will be equal to zero. An illustration of the impact that controls $v_1$ and $v_3$ have on the total number of infected persons may be found in figure \ref{fig:13b}. The use of controls $v_1$ and $v_3$ has the potential to lower the overall number of persons who are latent, aware, and oblivious of their infection when compared to the situation where control methods were not implemented (see figure \ref{fig:v=0}). Furthermore, we discover that the utilization of controls $v_1$ and $v_3$ in conjunction with one another has the potential to reduce HIV by the year 2030. In particular, around 2005 is a year that the HIV virus can be eradicated. The screening of latent infected individuals and the treatment of infected individuals over an extended length of time are therefore two methods that can be utilized to achieve optimal control outcomes. Following that, we will be able to put a halt to the transmission of HIV by increasing the number of individuals who are aware of the disease and increasing the number of people who have the chance to get treatment.\\

\noindent\textbf{Strategy 03: $v_1\neq 0,\;v_2\neq 0$ and $v_3=0$}\\
In order to detect latent persons $v_1(t)$ and educate individuals $v_2(t)$ who are unaware that they are infected, a mix of screening and education is utilized. For example, the values are 
$v_1\neq 0,\;v_2\neq 0$ and $v_3=0$. In the subsequent step, we will proceed to optimize the function $K$ that is shown in equation \eqref{equ:objective function}.
\begin{figure}[H]
\centering
	\subfigure[Variation in the optimal control solutions $v_1$ and $v_2$.]{\includegraphics[width=0.45\linewidth]{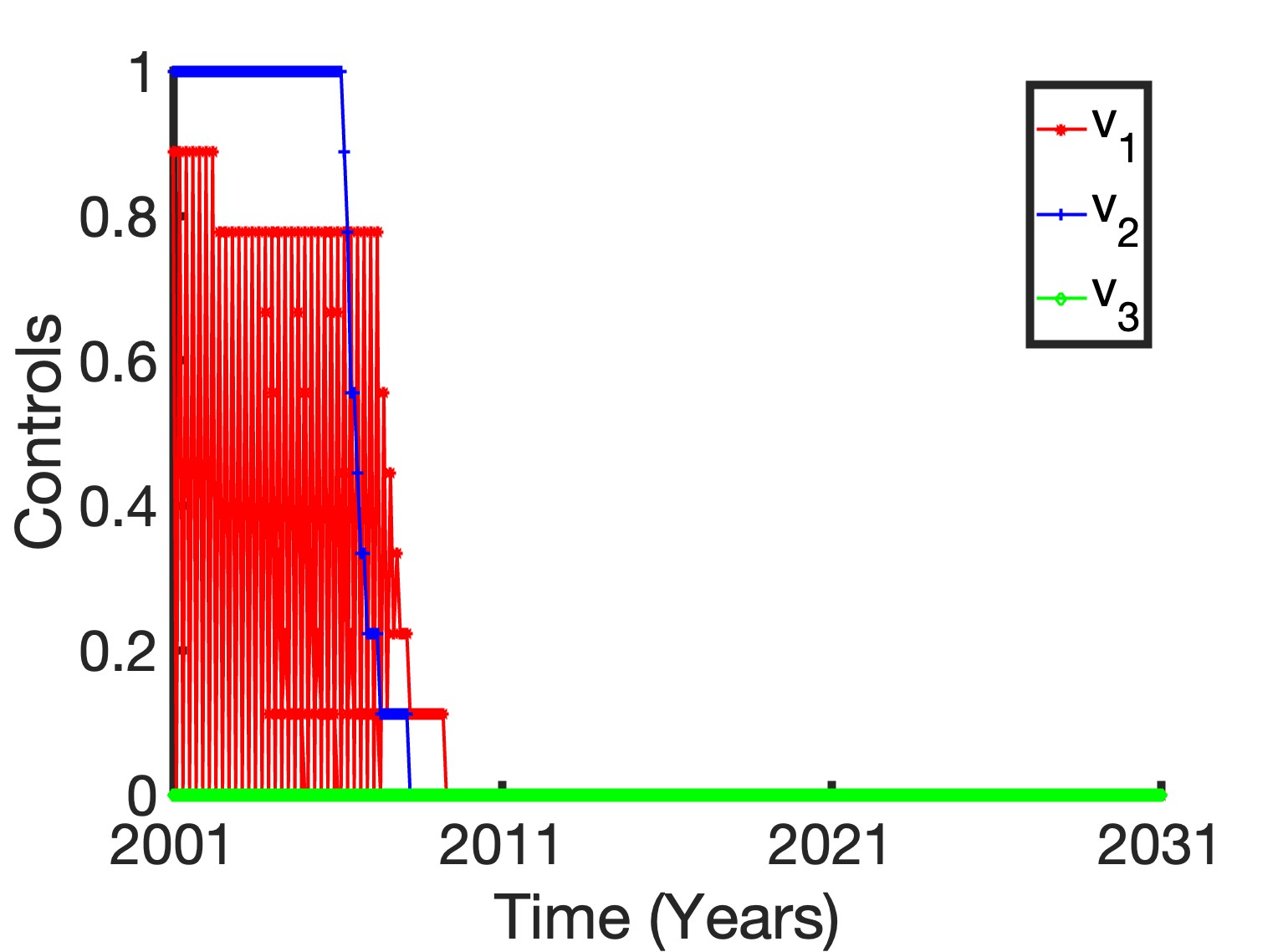}\label{fig:14a}}
	\subfigure[Total size of infected population under controls $v_1$ and $v_2$.]{\includegraphics[width=0.45\linewidth]{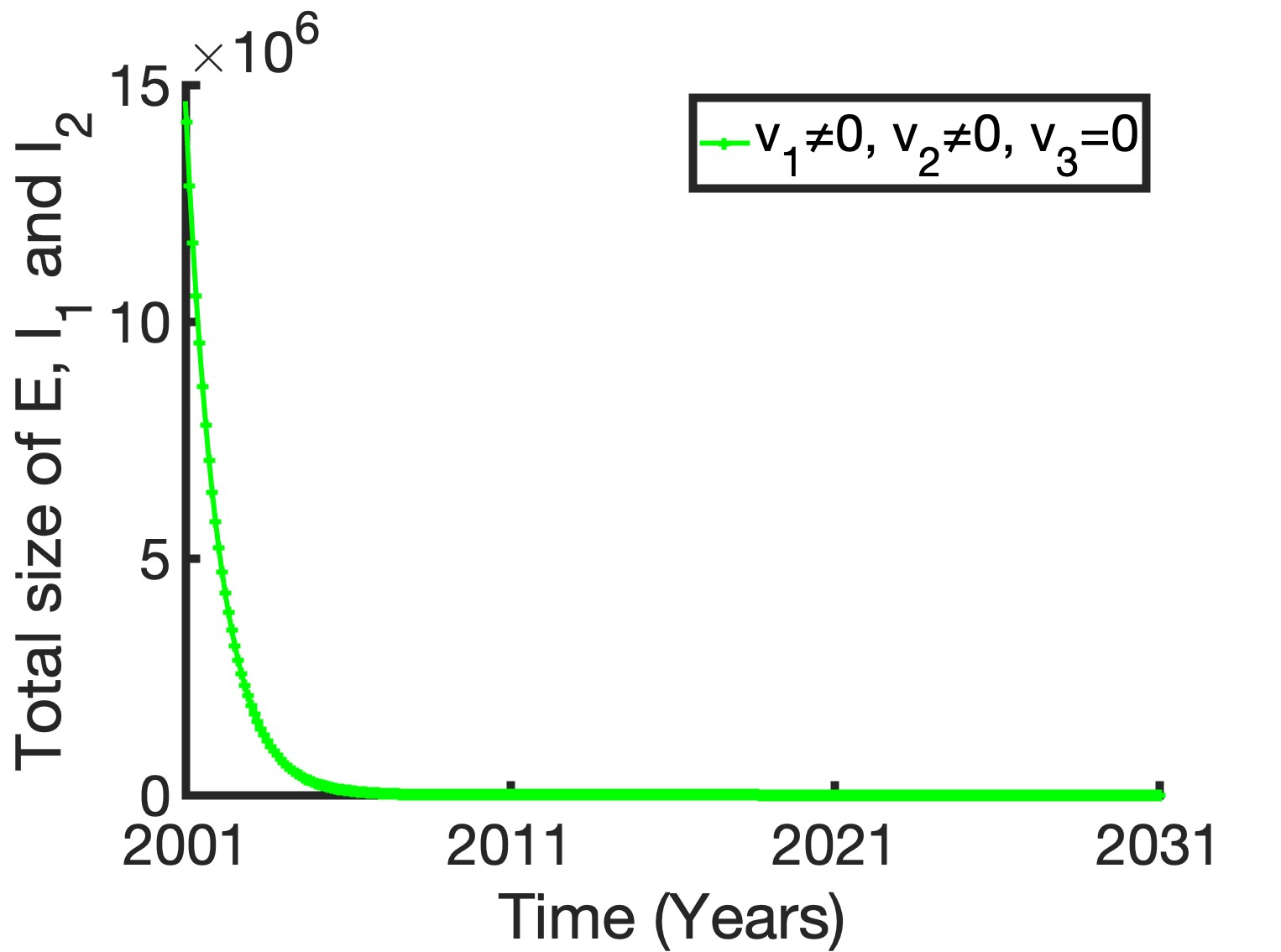}\label{fig:14b}}
\caption{Effects of control strategy implementation 3.}
\label{fig:14}
\end{figure}

\noindent The results of the control techniques $v_1$ and $v_2$ throughout the span of time are depicted in figure \ref{fig:14a}. This means that the maximum values of controls $v_2$ is maintained during the 8 years. Thereafter, the value of control $v_2$ faces a significant drop  and will not be needed after 2010. The value for control $v_1$ starts with being 90\% for 1 year and then dropped to 80\% and continues to be like that till 2009. In some consecutive years the up and down that are shown in the graph refers to the fact that for some years control value of $v_1$ will go up and down.  A comparison of the effects of controls $v_1$ and $v_2$ on the total number of infected persons is presented in figure \ref{fig:14b}. When compared to the scenario in which control techniques were not implemented (see figure \ref{fig:v=0}), the implementation of controls $v_1$ and $v_2$ has the potential to dramatically lower the number of persons that are infected. Specifically, the HIV virus can be eliminated in the year 2007.\\

\noindent\textbf{Strategy 04: $v_1\neq 0,v_2\neq 0$ and $v_3\neq 0$}\\
The proposed approach involves a mix of screening $v_1(t)$, education for those who are aware of being infected $v_2(t)$, and treatment for individuals who are aware of being infected $v_3(t)$. To clarify, the values are $v_1\neq 0,v_2\neq 0$ and $v_3\neq 0$. Subsequently, the cost function $K$ is optimized in equation \eqref{equ:objective function} in order to examine the control impact.\\
\noindent The change in the number of infected persons is depicted in figure \ref{fig:15}, which shows the situation when three controls are implemented simultaneously. Control $v_1$ is applied at 90\% at first and then decreased to 80\% for 7-8 years. Then, it keeps decreasing gradually. For the controls $v_2$ and $v_3$, both start at 100\% and gradually decrease throughout the year as shown in figure \ref{fig:15a}. By applying all the controls together, we can see that the goal could have been reached around 2004 (for details observe figure \ref{fig:15b}). The findings indicate that, in comparison to the scenario in which control techniques were not implemented (as depicted in figure \ref{fig:v=0}), controls $v_1, v_2$, and $v_3$ have the ability to decrease the overall number of persons who are latent, aware, and unaware of their infection. This indicates that the implementation of control techniques leads to a higher number of infected persons receiving treatments, as well as an increase in the number of latent and unaware infected individuals becoming conscious.
\begin{figure}[H]
\centering
\subfigure[Variation in the optimal control solutions $v_1, v_2$ and $v_3$.]{\includegraphics[width=0.45\linewidth]{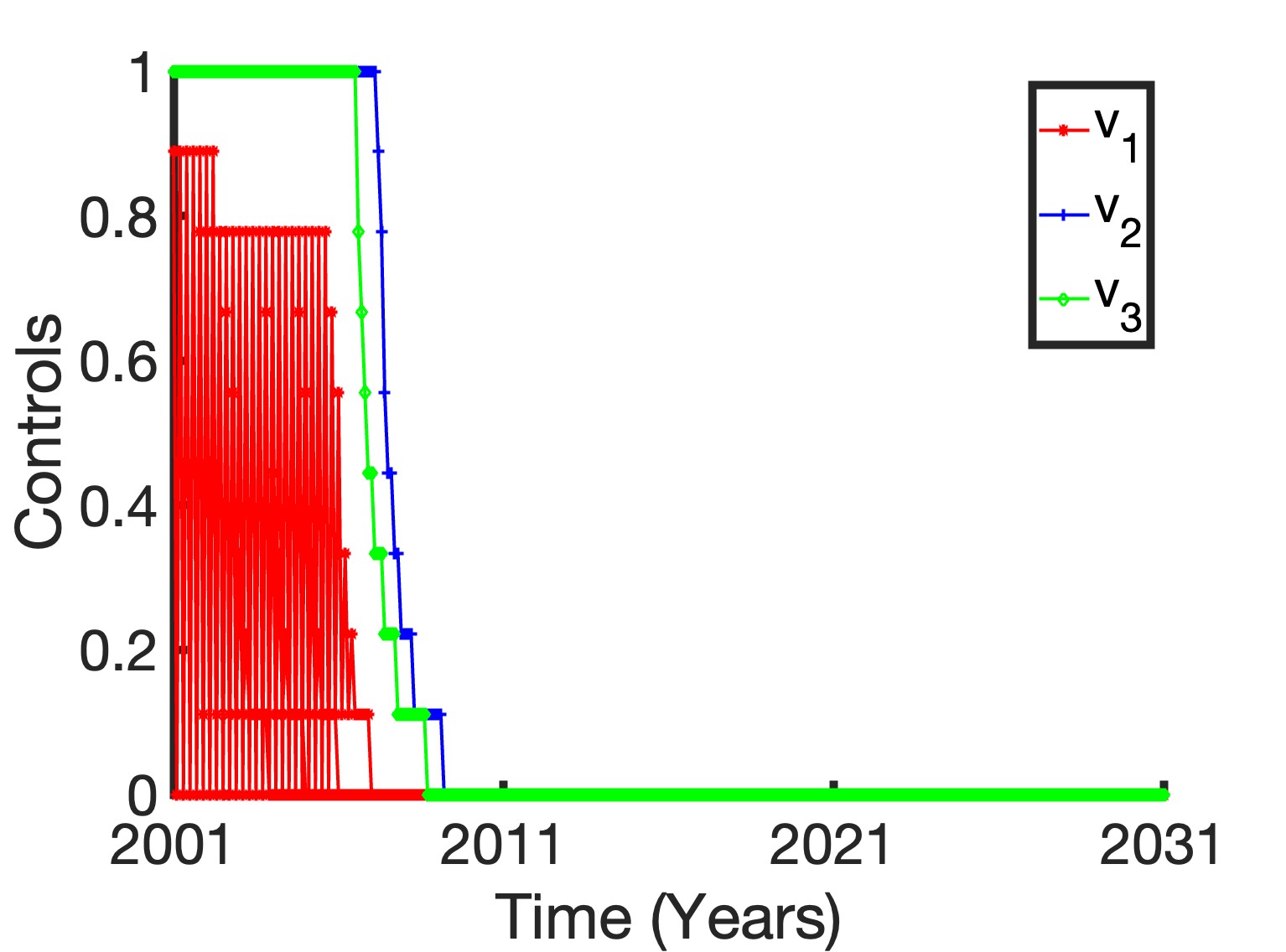}\label{fig:15a}}
\subfigure[Total size of infected population under controls $v_1, v_2$ and $v_3$.]{\includegraphics[width=0.45\linewidth]{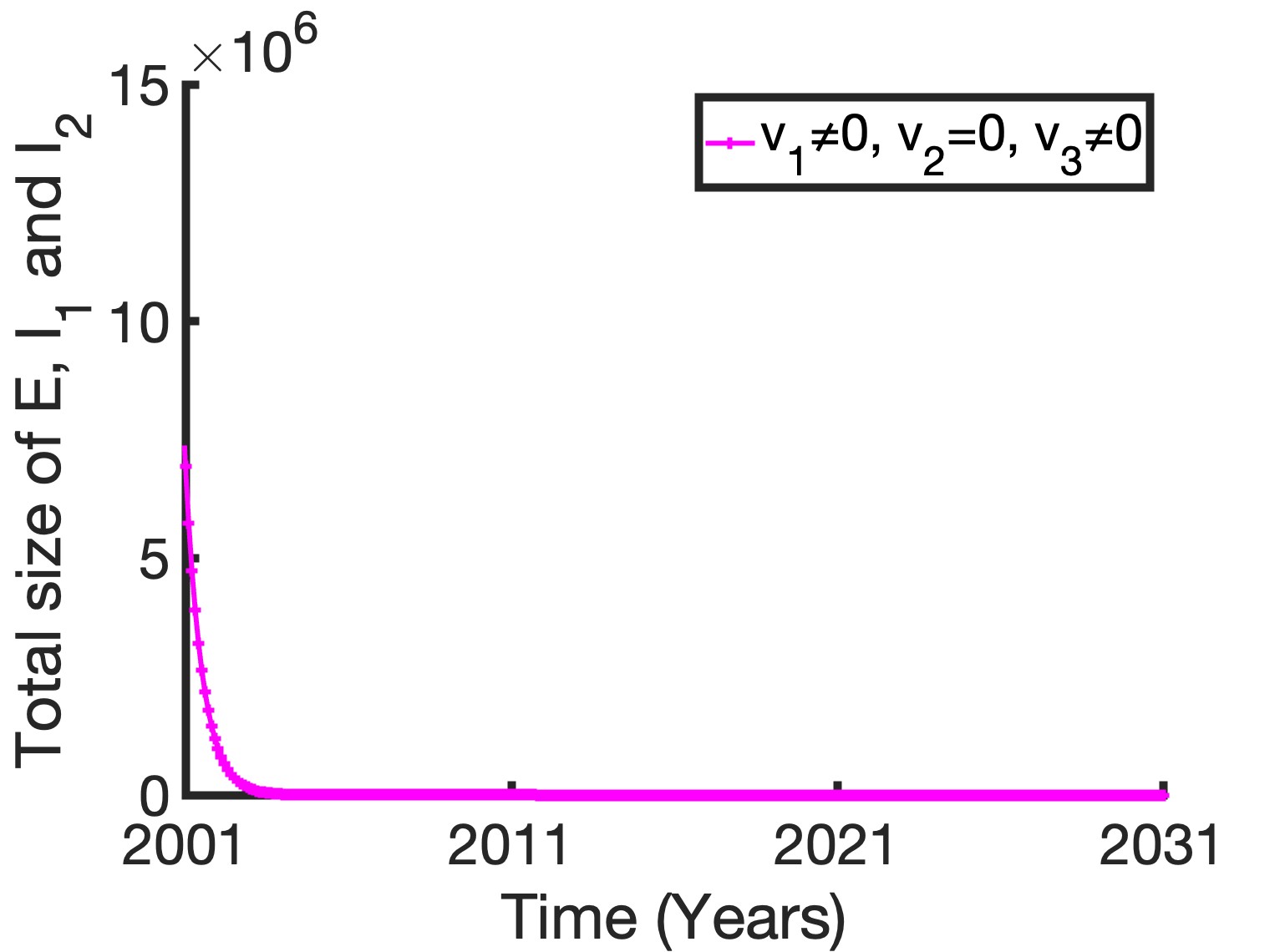}\label{fig:15b}}
\caption{Effects of control strategy implementation 4.}
\label{fig:15}
\end{figure}

\noindent Furthermore, as depicted in figure \ref{fig:all_controls}, we evaluate the impact of four distinct control methods on the number of individuals who are infected. These individuals include those who are unaware of their infection, those who are aware of their illness, and those who are experiencing latent infection. According to our findings, \textbf{Strategy 04}: $v_1\neq 0,\;v_2\neq 0$ and $v_3\neq 0$ is the most effective plan. Strategy 4 has the potential to eliminate HIV in the shortest amount of time (about three years). There are other ways that require additional time in order to obtain the same level of control. We summarize the control outcomes in Table 4 in order to determine whether or not the four control measures can reach the three goals that have been suggested by UNAIDS, which are to reduce the burden of AIDS by 95\% and almost eliminate it by the year 2030. There are four different control techniques that have the potential to accomplish the control objectives that have been stated by the United Nations.
\begin{figure}[H]
  \centering
  \includegraphics[width=0.6\linewidth]{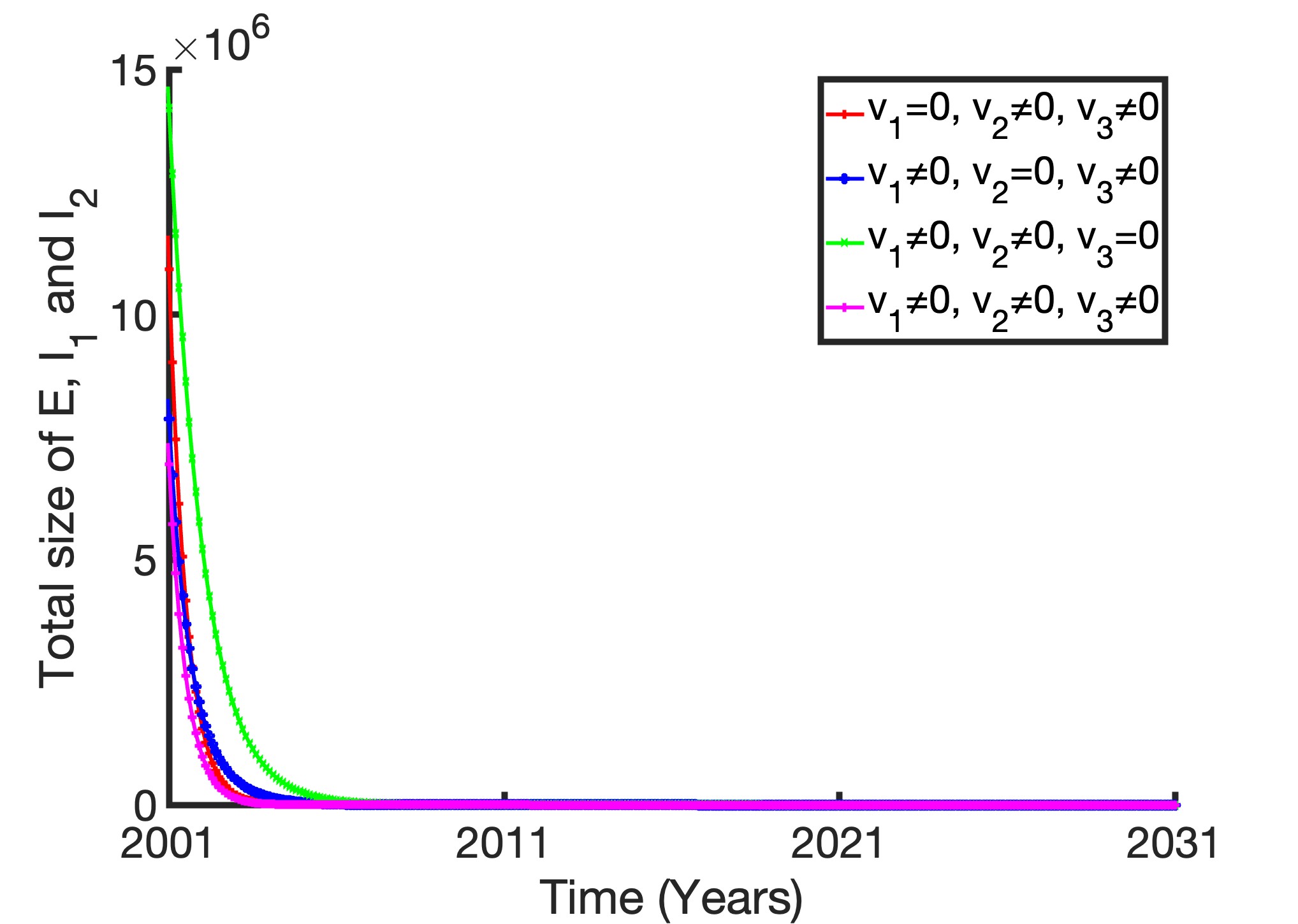}
  \caption{The influence of various control strategies on the frequency of infection among individuals.}
  \label{fig:all_controls}
\end{figure}

\noindent The following step is to investigate the impact that control weights have on Strategy 4. When $q_2 = 50$ and $q_3 = 20$, with $q_1 = 100,\; q_1 = 150,$ and $q_1 = 200$, we conduct an analysis of the optimal control curves as well as the number of infected individuals. In figure \ref{fig:str_4_q1}, the length of the maximum control level of $v_1$ decreases with an increase in $q_1$, yet the maximum control levels of $q_2$ and $q_3$ extend with time. The controls $v_2$ and $v_3$ soon achieve their maximum control level when $q_1$ is set to 100, and they remain at that level for about 7 to 8 years (see figure \ref{fig:q11}), after those controls have reached their maximum control level. Afterward, the control levels gradually decreases. The control $v_1$ starts at 90\% and gradually decreases to zero starting from 2004. Beginning in 2001 and continuing until 2008, the control level of control $v_2$ continues at a consistently high level. After that, it steadily diminishes until it reaches 0. The control $v_3$ experiences a decrease from 2008 to 2030. In figure \ref{fig:q12}, with $q_1$ set to 150, the controls $v_1,\; v_2$ and $v_3$ quickly reach their maximum control level and remain at that level for several years. After certain years, the controls $v_1,\; v_2,$ and $v_3$ gradually decrease to zero. In figure \ref{fig:q13}, with $q_1=200$, the controls $v_1,\; v_2$ and $v_3$ reach their maximum control level and remain at that level for approximately seven, nine and twelve months, respectively. Following certain years, the controls $v_1,\; v_2$ and $v_3$ gradually decrease to zero. Furthermore, we investigate how varying $q_1$ values impact the number of infected individuals. According to figure \ref{fig:q14}, achieving HIV eradication can be done sooner with a smaller weight $q_1$.
\begin{figure}[H]
\centering
\subfigure[Weight: $q_1=100, q_2=50, q_3=20$.]{\includegraphics[width=0.45\linewidth]{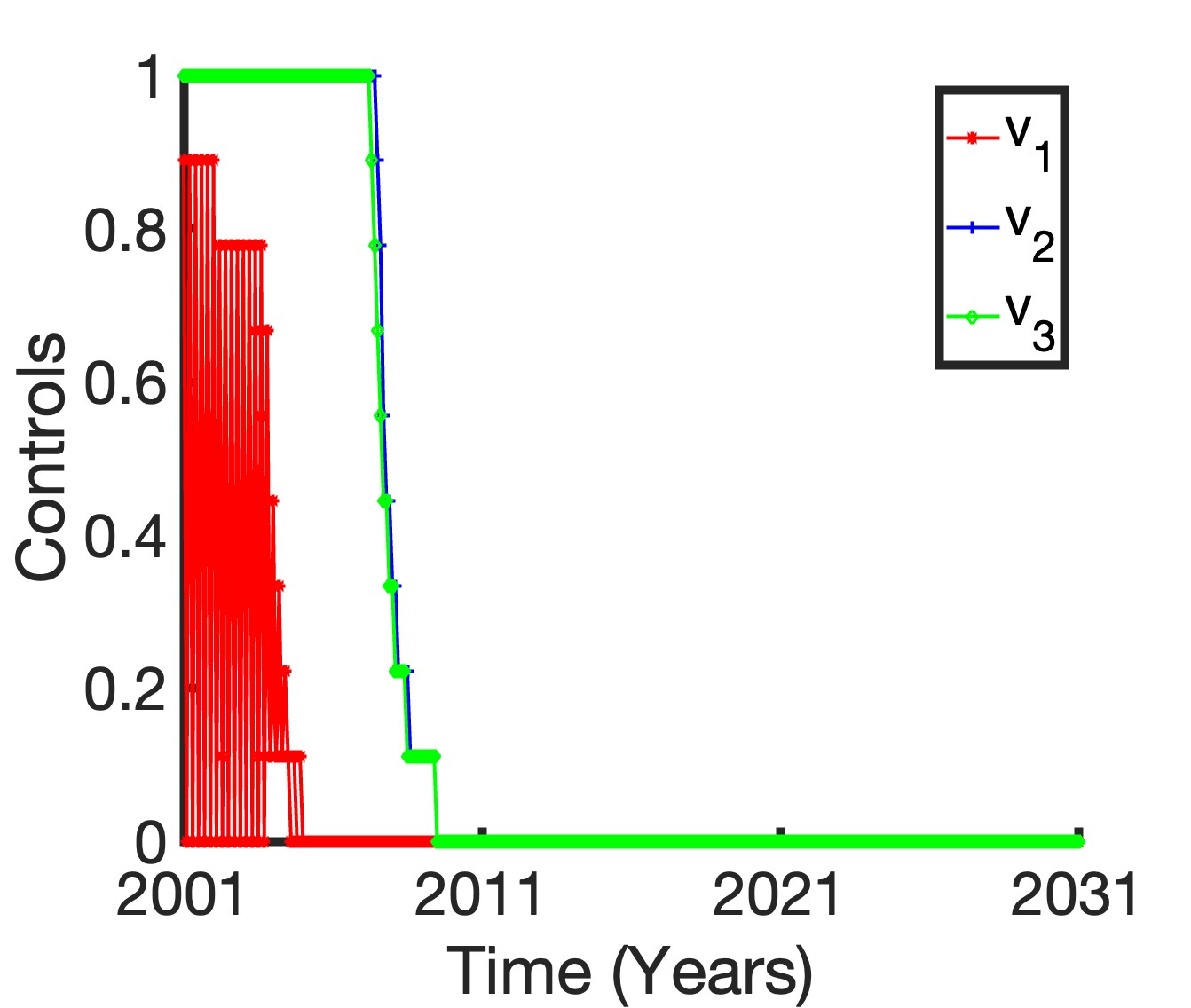}\label{fig:q11}}
\subfigure[Weight: $q_1=150, q_2=50, q_3=20$.]{\includegraphics[width=0.45\linewidth]{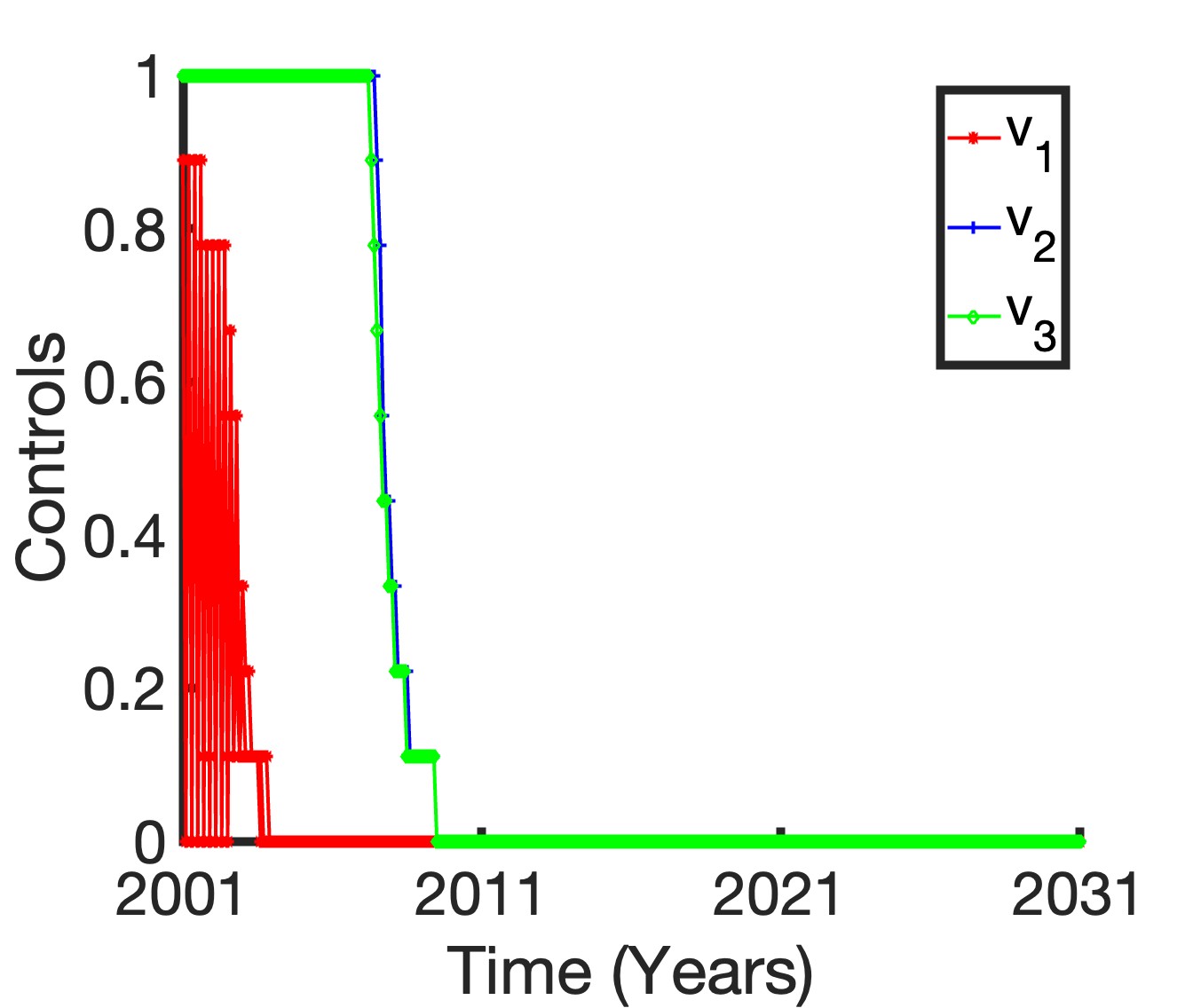}\label{fig:q12}}
\subfigure[Weight: $q_1=200, q_2=50, q_3=20$.]{\includegraphics[width=0.45\linewidth]{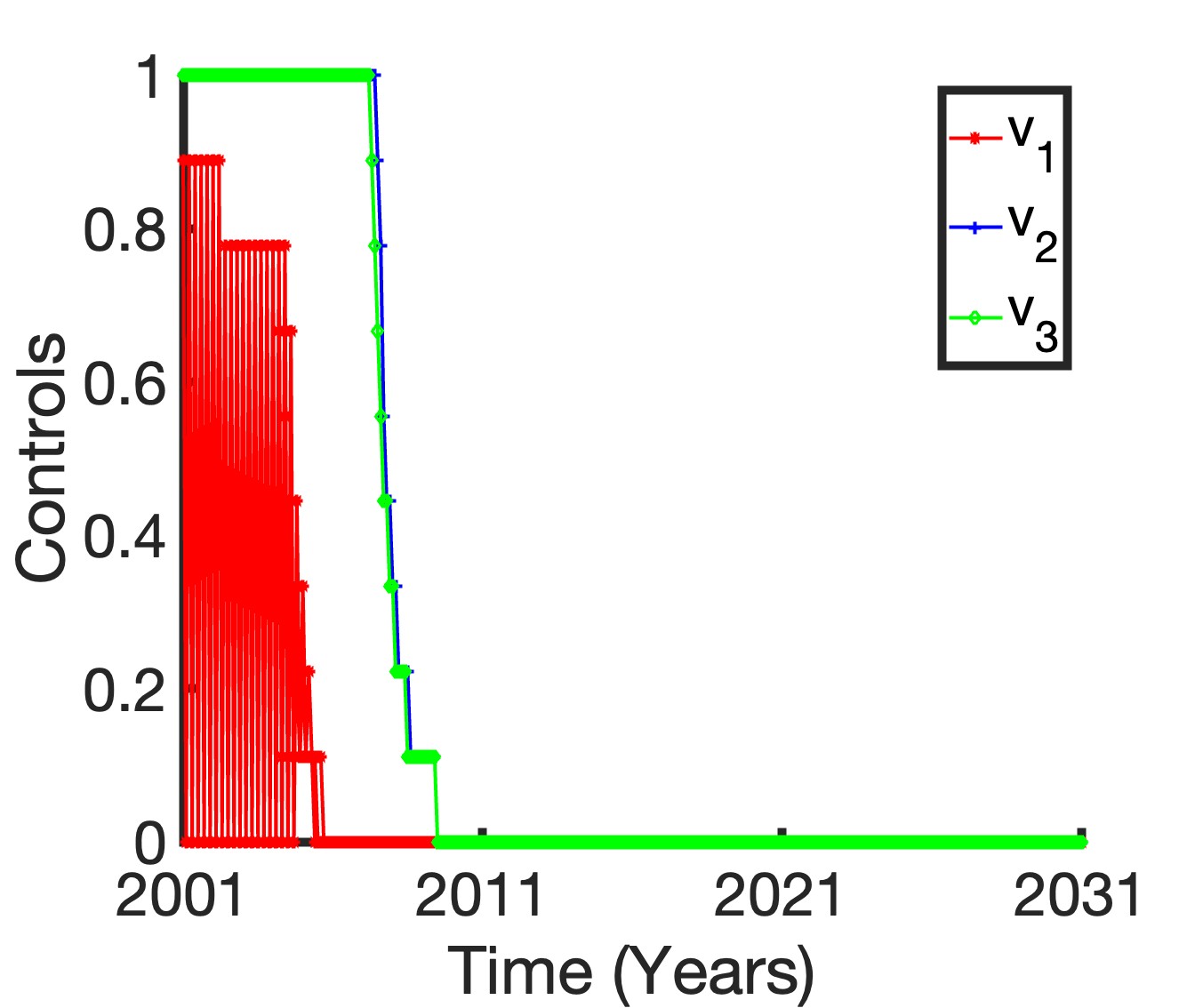}\label{fig:q13}}
\subfigure[The total number of infected people under various $q_1$ cost weights.]{\includegraphics[width=0.45\linewidth]{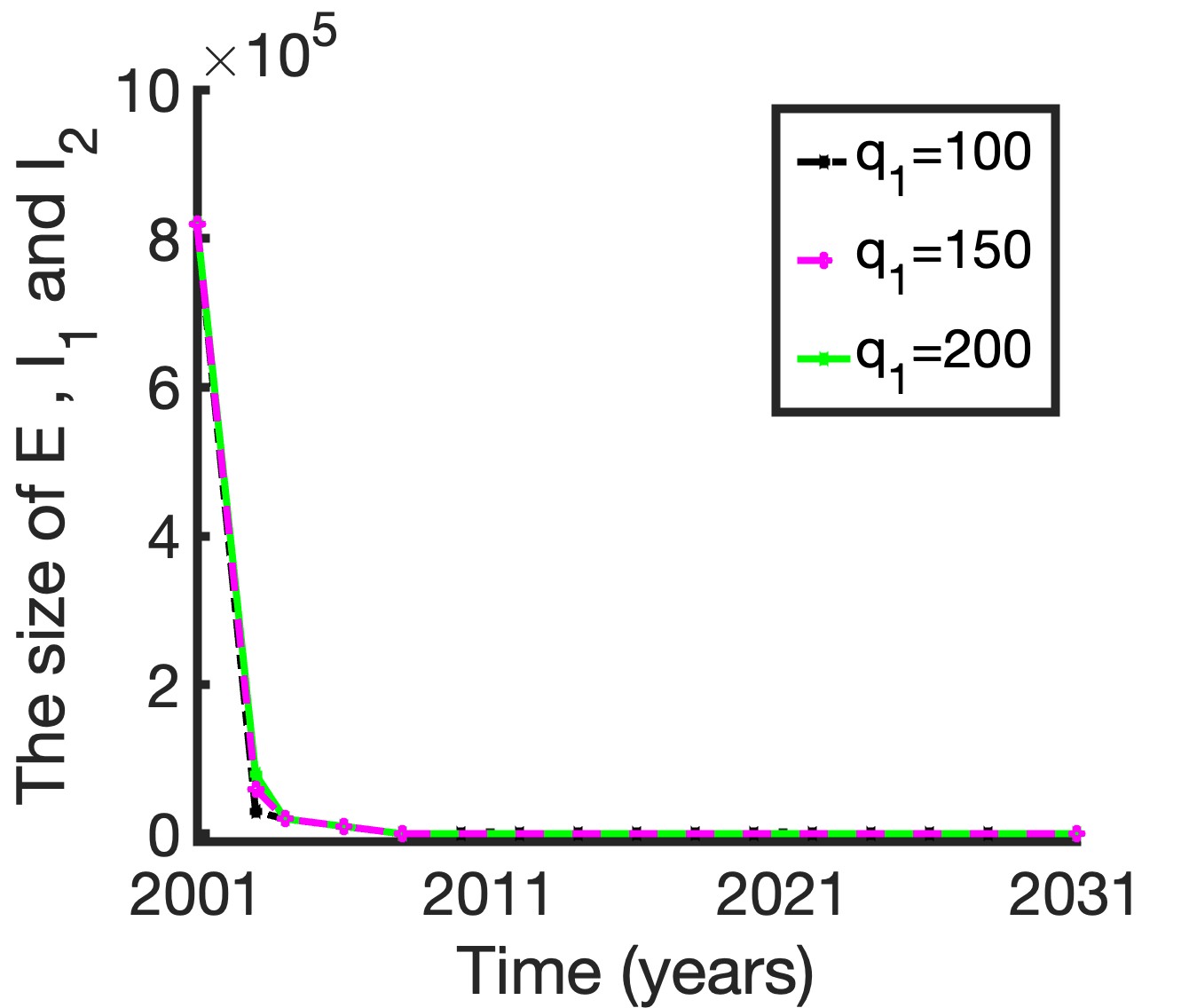}\label{fig:q14}}
\caption{Effects of modifications on the number of infected people and control procedures on the cost weight of control $v_1(t),\; q_1$.}
\label{fig:str_4_q1}
\end{figure}

\noindent The effect of weight $q_2$ on the control curves and the number of infected individuals for $q_2$ values of 30, 90, and 170 is shown in figure \ref{fig:str_4_q2}. As the weight $q_2$ increases, the maximum control level of control $v_2$ gradually decreases. In figure \ref{fig:q21}, with $q_2 = 30$, the control $v_1$ is activated from 90\% and then progressively decreases and reaches zero. For the controls $v_2$ and $v_3$, both start at 100\% and by 2010, decline to zero. In figures \ref{fig:q22} and \ref{fig:q23}, the controls $v_1$ and $v_3$ stay almost the same as the first one. But as we changed the $q_2$ value, only the control percentage of $v_2$ changes.

\begin{figure}[H]
\centering
\subfigure[Weight: $q_1=50, q_2=30, q_3=10$.]{\includegraphics[width=0.45\linewidth]{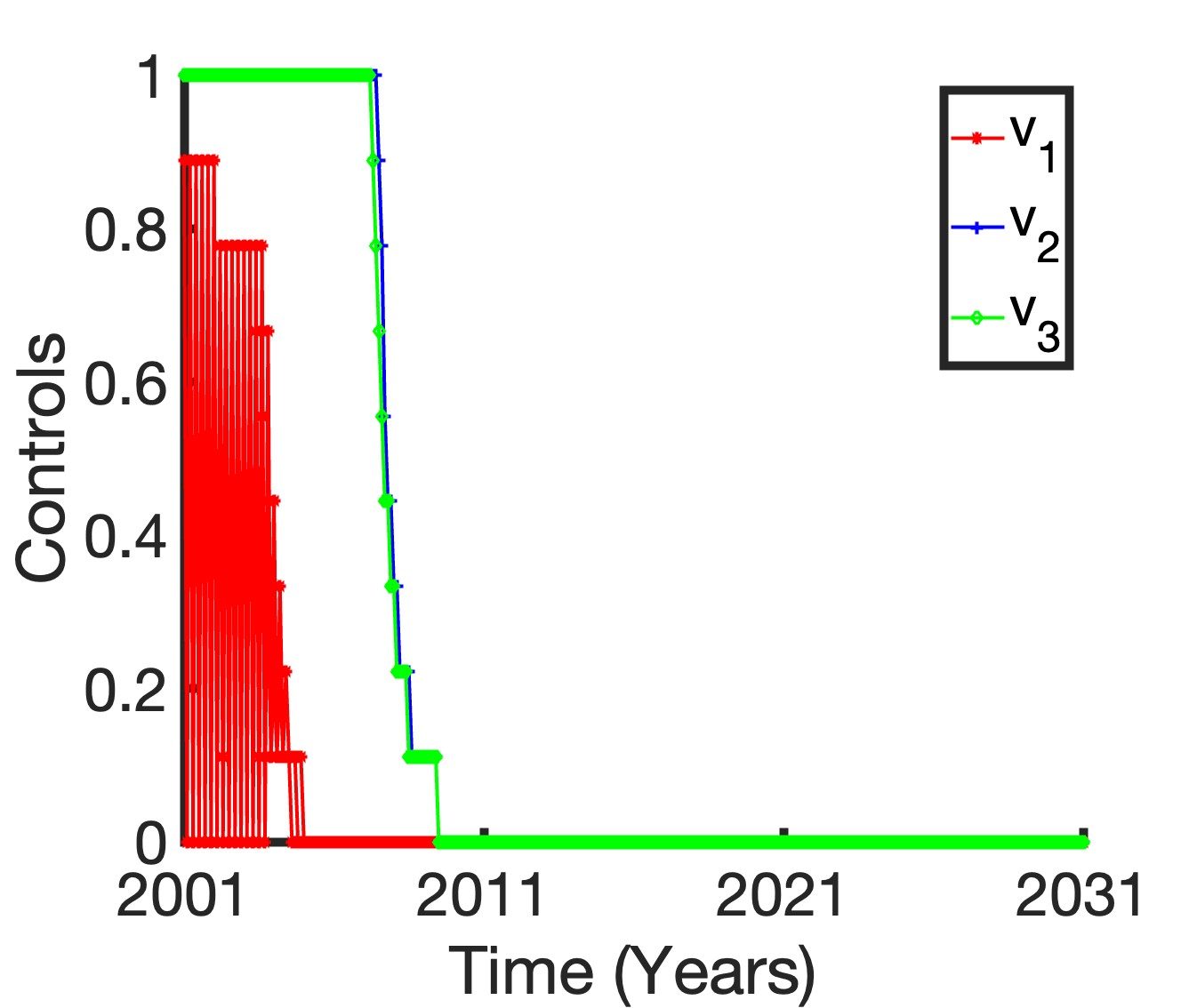}\label{fig:q21}}
\subfigure[Weight: $q_1=50, q_2=90, q_3=10$.]{\includegraphics[width=0.45\linewidth]{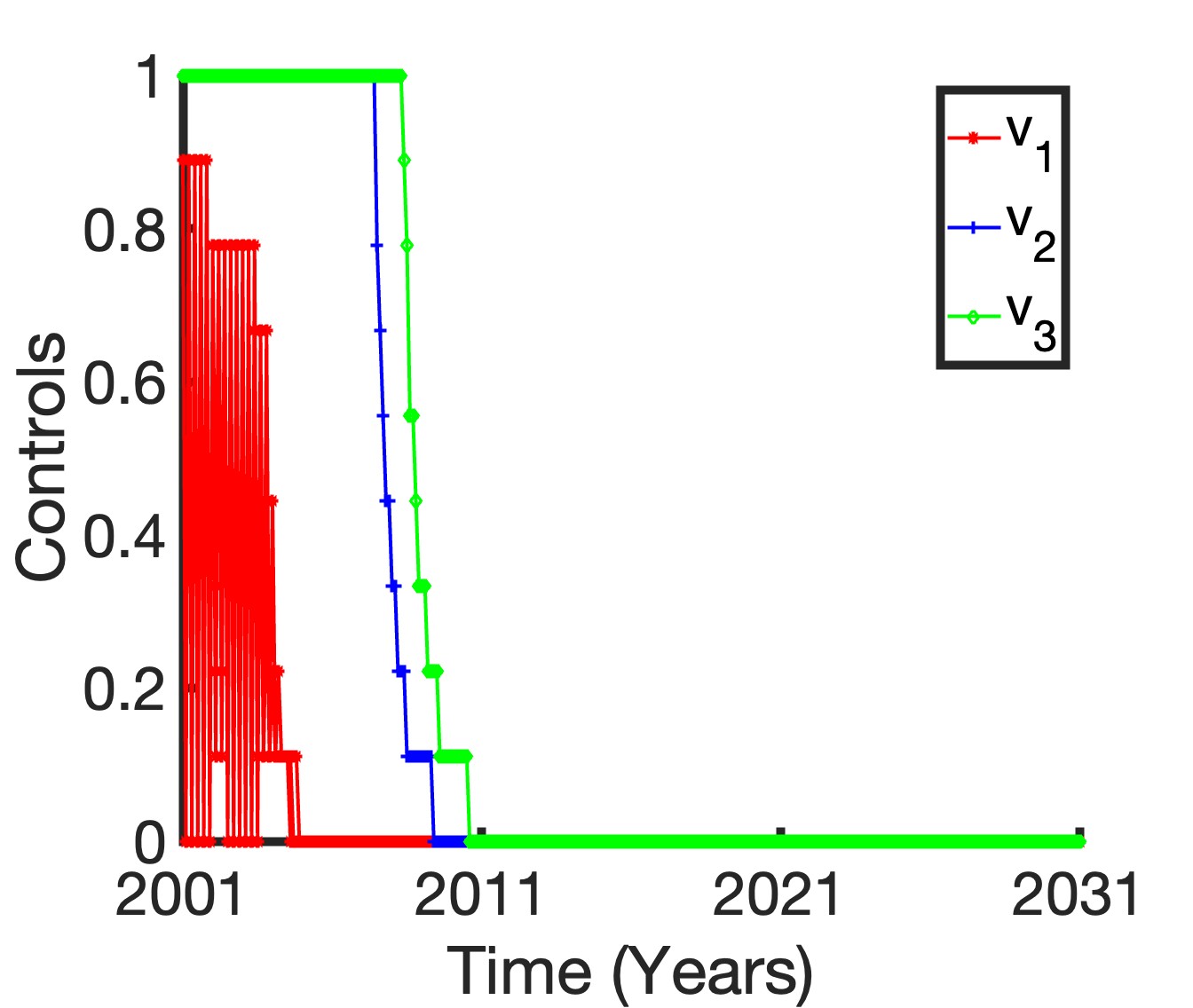}\label{fig:q22}}
\subfigure[Weight: $q_1=50, q_2=170, q_3=10$.]{\includegraphics[width=0.45\linewidth]{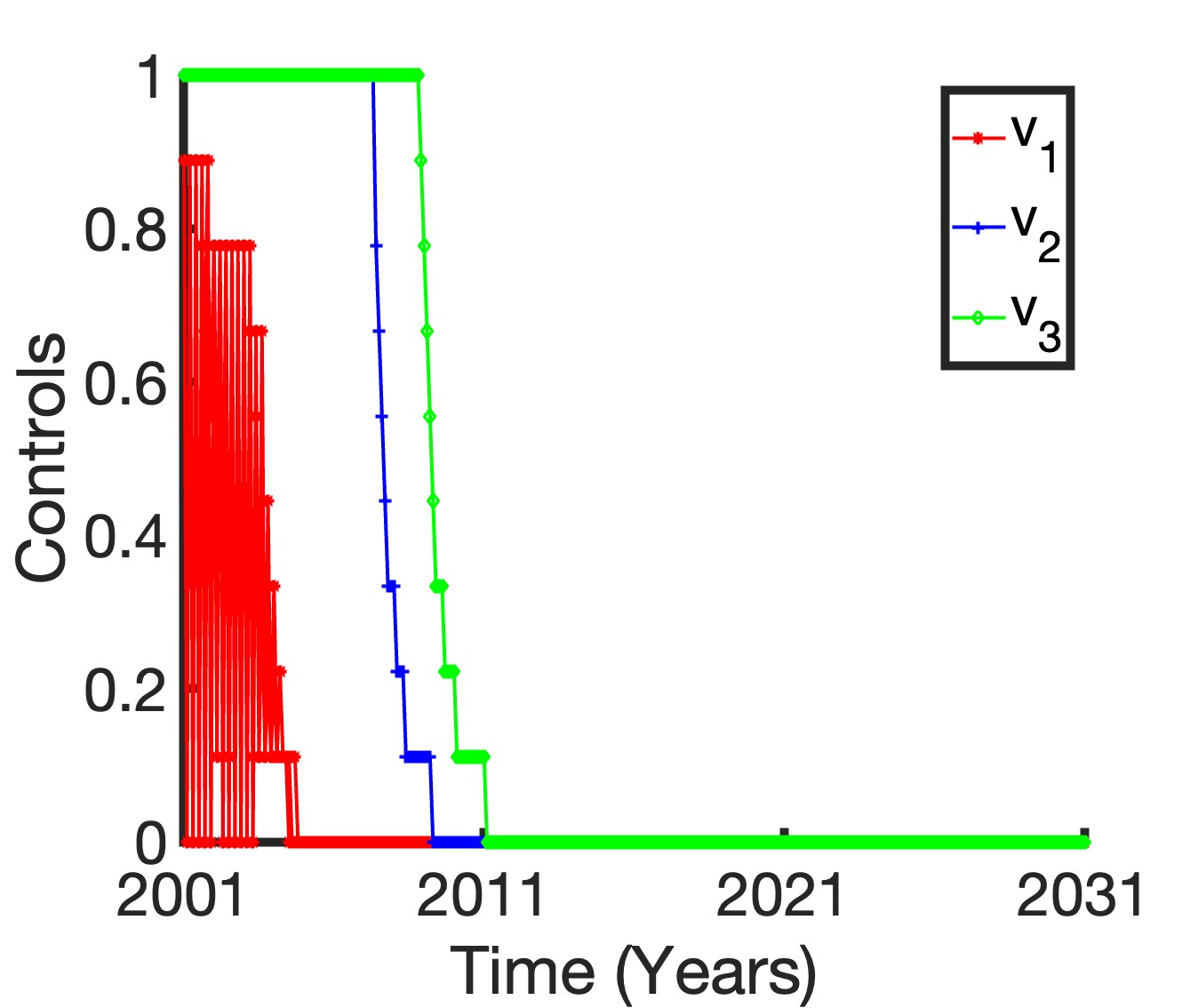}\label{fig:q23}}
\subfigure[The total number of infected people under various $q_2$ cost weights.]{\includegraphics[width=0.45\linewidth]{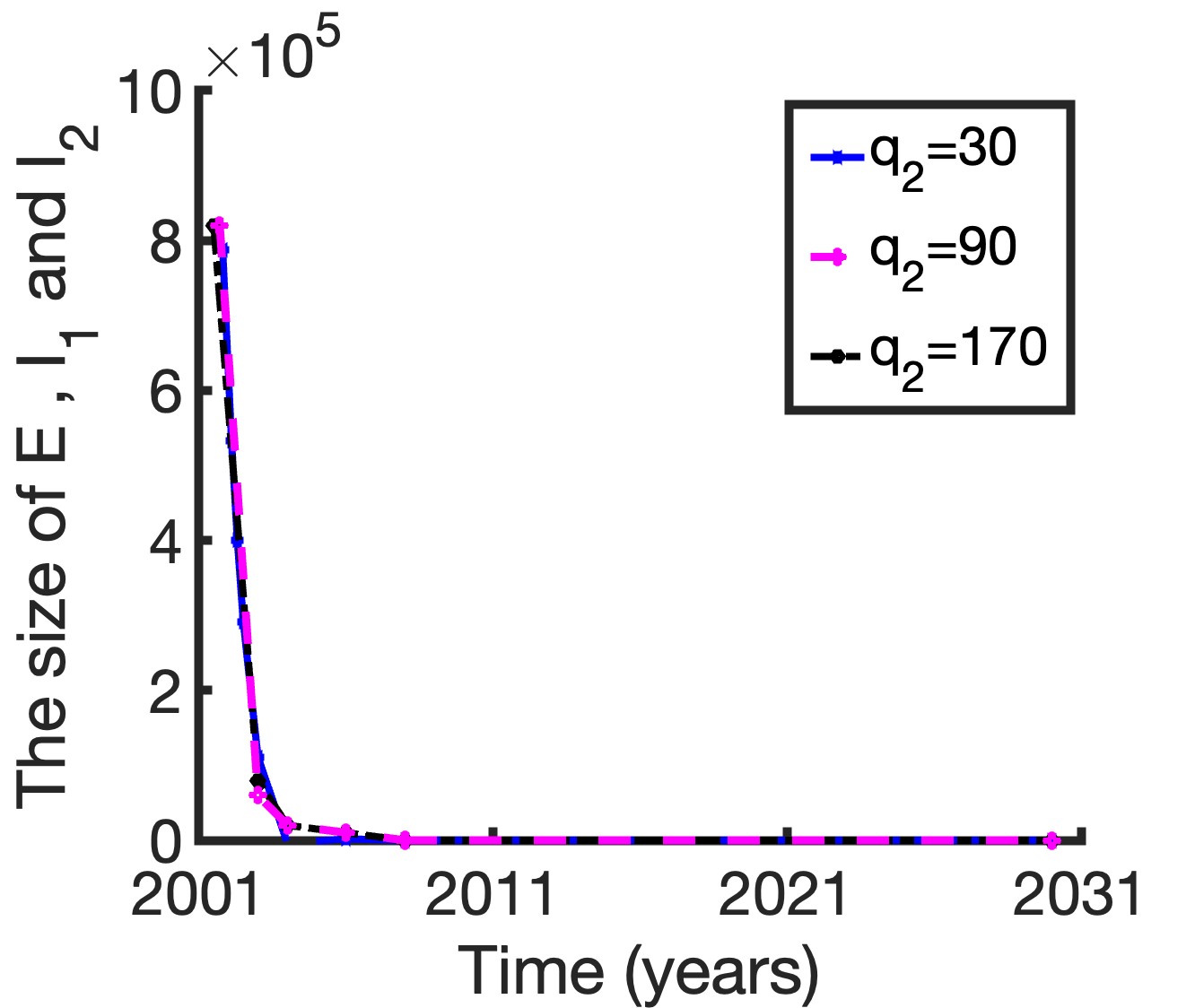}\label{fig:q24}}
\caption{Effects of modifications on the number of infected people and control procedures on the cost weight of control $v_2(t), q_2$.}
\label{fig:str_4_q2}
\end{figure}
\noindent In figure \ref{fig:q22}, with $q_2 = 90$, the controls $v_1, v_2$ and $v_3$ quickly reach their maximum control level. After a certain period of time, the controls $v_1, v_2$ and $v_3$ gradually decrease to zero. In figure \ref{fig:q23}, with $q_2$ set to 170, the controls $v_1, v_2$ and $v_3$ reach their maximum control level and remain active for approximately nine, ten and twelve months, respectively. In addition, we investigate how varying $q_2$ values impact the number of individuals who become infected. According to figure \ref{fig:q24}, a decrease in weight $q_2$ results in a more effective control effect.

\noindent Furthermore, we investigated the impact of varying $q_3$ values on the control strategies and the population of infected individuals. Here, we examine the scenarios where $q_3$ is equal to 5, 10, and 15, respectively. Figure \ref{fig:str_4_q3} demonstrates that as the weight $q_3$ increases, the maximum control level of control strategy $v_3$ gradually decreases. In figure \ref{fig:q31}, with $q_3$ equal to 5, the controls $v_1, v_2$ and $v_3$ quickly reach their maximum control level. After a certain period, the controls $v_1, v_2$ and $v_3$ gradually decrease to zero. This decline occurs at different times for each control. Controls $v_1$ and $v_2$ remain the same throughout this process as we changed the values of $q_3$ in this case. Figure \ref{fig:q32} demonstrates that with $q_3 = 10$, the controls $v_1, v_2$ and $v_3$ quickly reach their maximum control level. Figure \ref{fig:q33} depicts when $q_3=15$. For this the control $v_3$ is extended a little bit more than the other two. Finally, we examine how varying $q_3$ values impact the number of infected individuals, as depicted in figure \ref{fig:q34}. We have observed that reducing the weight $q_3$ leads to a more favorable control effect.

\begin{figure}[H]
\centering
\subfigure[Weight: $q_1=50, q_2=30, q_3=5$.]{\includegraphics[width=0.45\linewidth]{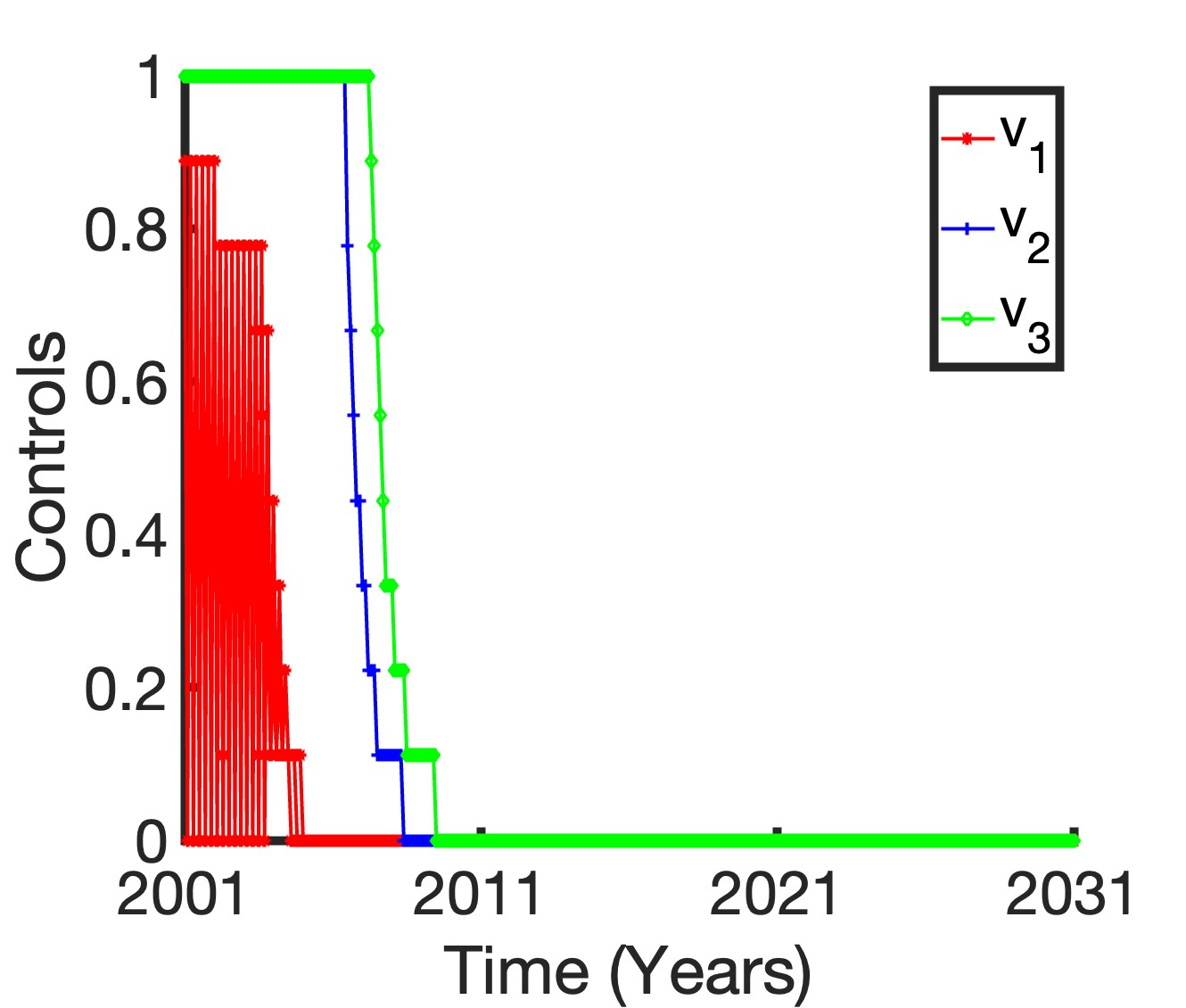}\label{fig:q31}}
\subfigure[Weight: $q_1=50, q_2=30, q_3=10$.]{\includegraphics[width=0.45\linewidth]{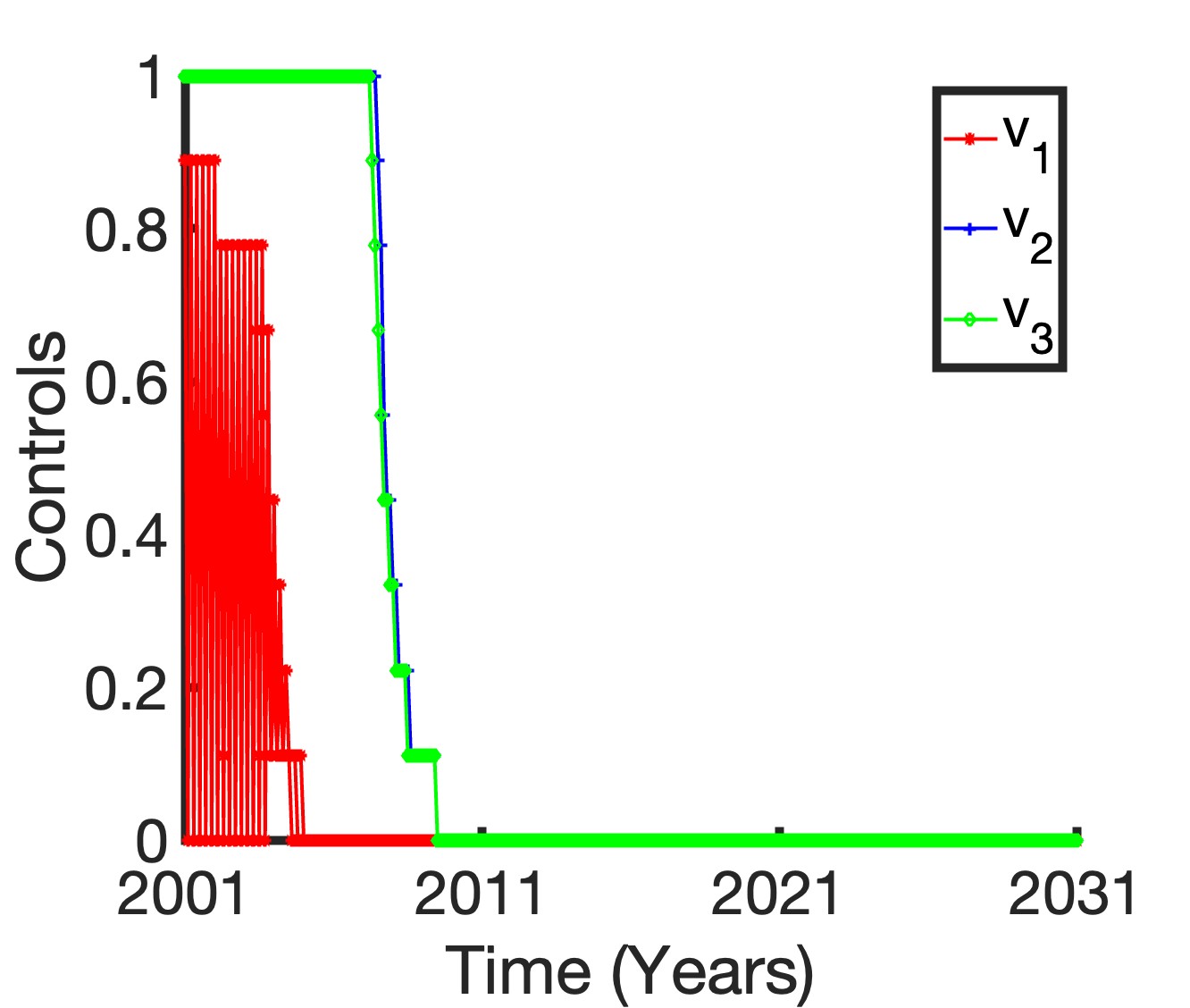}\label{fig:q32}}
\subfigure[Weight: $q_1=50, q_2=30, q_3=15$.]{\includegraphics[width=0.45\linewidth]{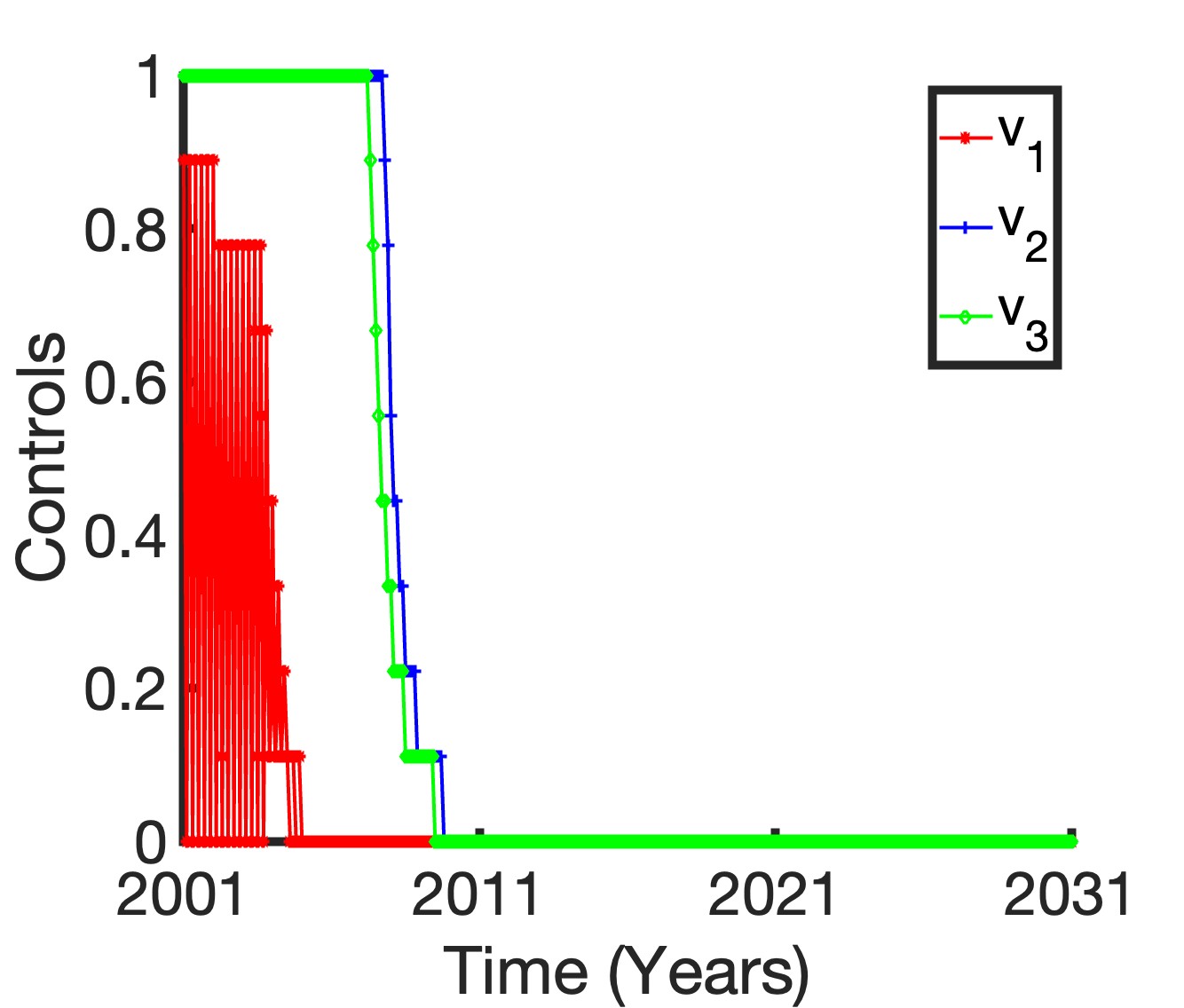}\label{fig:q33}}
\subfigure[The total number of infected people under various $q_3$ cost weights.]{\includegraphics[width=0.45\linewidth]{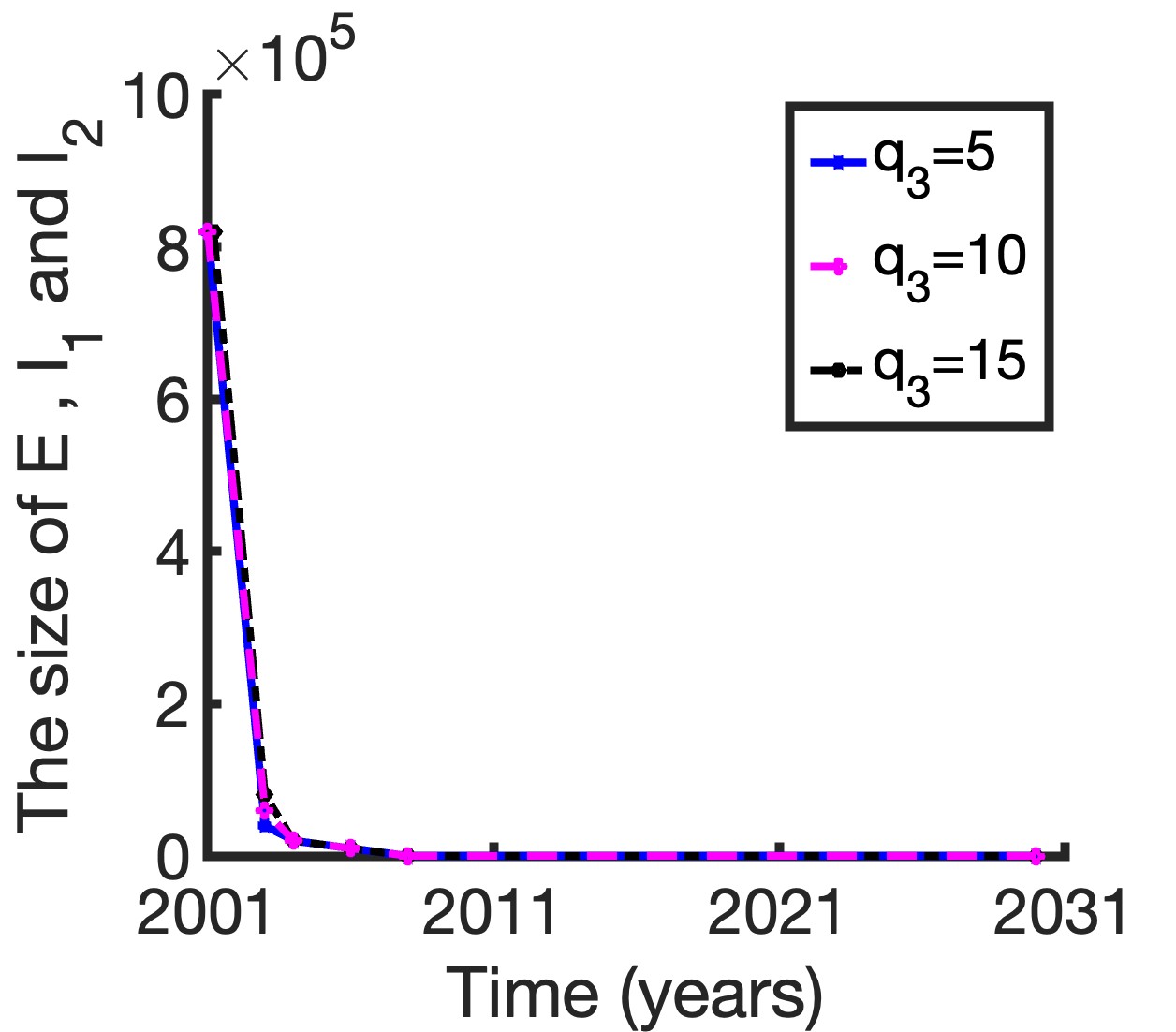}\label{fig:q34}}
\caption{Effects of modifications on the number of infected people and control procedures on the cost weight of control $v_3(t),\; q_3$.}
\label{fig:str_4_q3}
\end{figure}

\noindent In brief, there is a steady incline in the control level of $v_1$ as the weight $q_1$ reduces. Simultaneously, there is a rise in the control levels of $v_2$ and $v_3$. As the weight of $q_2$ grows, there is a progressive drop in the control level of $v_2$, while the control levels of $v_1$ and $v_3$ also decrease. As the weight of $q_3$ grows, there is a steady drop in the control level of $v_3$, while the control levels of $v_1$ and $v_2$ also decrease.

\subsection{Analysis of Efficiency}
Even while each of the four control techniques has the potential to accomplish the control goals that have been established by UNAIDS, we must also take into account the efficiency of each of these strategies. In consideration of this, within the scope of this section, we shall analyze the efficacy of four distinct control methods.\\
The initial step involves doing the efficacy analysis by employing the methodology described in references \cite{ghosh2021mathematical,mondal2021effect,wang2022effect}. Following is a definition of the efficiency index:
\begin{equation*}
    \Sigma = \left( 1 - \frac{\Omega_c}{\Omega_s} \right) \times 100\%,
\end{equation*}
where, the variable $\Omega_c$ represents the total number of persons who have been infected following the implementation of various control strategies and while there are no control methods in place, the total number of infected persons is denoted by the symbol $\Omega_s$. The whole number of people that were infected throughout the period of time spanning from 0 to $t_f$ is represented by the symbol:
\begin{equation*}
    \Omega = \int_{0}^{t_f} [E(t) + I_1(t) + I_2(t)] \, dt,
\end{equation*}
where, the variable $t_f$ represents the final time at which the implementation controls finish. Here, $t_f=28$ years.

\noindent The effective indices of four different control systems are presented in table \ref{tab:control-strategies}. The effectiveness measures demonstrate that Strategies 1, 2, 3, and 4 are successful in reducing the spread of HIV. Nevertheless, strategy 4 is the most effective control approach.

\begin{table}[H]
\centering
\caption{Indicators related to the efficacy of control strategies.}
\label{tab:control-strategies}
\begin{tabular}{lcc}
\toprule
Strategies & $\Omega = \int_{0}^{t_f} E(t) + I_1(t) + I_2(t)\,dt$ & $\Sigma = \left(1 - \frac{\Omega_c}{\Omega_s}\right) \times 100\%$ \\
\midrule
No control & 27,21,770 & 0 \\
Strategy 01 & 2,99,910 & 80.56\% \\
Strategy 02 & 1,99,840 & 90.48\% \\
Strategy 03 & 2,78,528 & 81.67\% \\
Strategy 04 & 40,004 & 97.95\% \\ 
\bottomrule
\end{tabular}
\end{table}

\subsection{Cost-Benefit Analysis}
Here, the focus is on the \textbf{cost-benefit analysis} of different strategies to control the spread of HIV. This analysis is crucial because it quantifies the economic implications of each control strategy, balancing the total cost due to deaths and the total benefits derived from reducing the number of infected individuals. By comparing these factors, the analysis identifies which strategies provide the most significant benefits relative to their costs, thus guiding the allocation of limited resources in public health interventions.

\noindent The cost-benefit analysis in this paper employs the \textbf{Value of Statistical Life (VSL)}, set at 13.2 million dollars \cite{usdot2024vsl}, to monetize the benefits of reducing fatalities. VSL is a standard measure used in economic evaluations to quantify the benefits of risk-reducing measures by assigning a monetary value to each life saved. Here, the total costs associated with deaths are weighed against the benefits obtained from implementing control strategies. These benefits include reduced healthcare costs and the economic value of lives saved. The analysis in this paper is instrumental in determining the most efficient strategy, ensuring that the resources invested in HIV control yield the highest possible return in terms of lives saved and economic gains.

\noindent This approach is particularly important for public health policymakers, as it provides a rigorous economic justification for choosing one strategy over another. The incorporation of VSL into the cost-benefit framework allows for a comprehensive assessment of the trade-offs involved in different intervention strategies, making it a cornerstone of the paper's contribution to the field of HIV management.

\begin{table}[htbp]
	\centering
	\resizebox{\textwidth}{!}{
		\begin{tabular}{|c|c|c|c|c|}
			\hline
			\textbf{Strategies} & \textbf{Total Infected Individuals} & \textbf{Total Cost Due to Death (M)} & \textbf{Total Benefit (M)} & \textbf{Total Cost of Strategies (M)} \\ \hline
			No Strategy & 27,21,770 & 3,59,27,364 & - & - \\ 
			Strategy 1 & 2,99,910 & 39,58,812 & 3,19,68,552 & 34,530.28 \\ 
			Strategy 2 & 1,99,840 & 26,37,888 & 3,32,89,476 & 46,967.56 \\ 
			Strategy 3 & 2,78,528 & 36,76,569 & 3,22,50,795 & 37,311.84 \\ 
			Strategy 4 & 40,004 & 52,852 & 3,94,74,512 & 59,404.84 \\ \hline
	\end{tabular}}
	\caption{Cost-Benefit analysis of the Above-Said strategies.}
\end{table}
\noindent Costs are calculated by using the data from CDC: Centers for Disease Control and Prevention \cite{CDCgov}.
So, from the analysis we can clearly see that we will get the most benefit by applying \textbf{strategy 4}.

\section{Discussion}
\label{sec:discussion}
This study examines the transmission of HIV through a mathematical model that includes various stages of infection and treatment. Utilizing the model, we do both theoretical analysis and numerical simulations. Through a comprehensive theoretical examination, it is clear that when the basic reproduction number ($\mathcal{R}_0< 1$), the disease-free equilibrium is globally asymptotically stable, leading to the progressive eradication of the illness. When the basic reproduction number, $\mathcal{R}_0>1$, the endemic equilibrium becomes globally asymptotically stable, resulting in the spread of the disease \cite{dietz1993estimation}. The parameter values are obtained by the Markov Chain Monte Carlo (MCMC) technique \cite{brooks1998markov}. Through the implementation of numerical simulations, we validate the stability of both the disease-free equilibrium and the endemic equilibrium. We provide a comprehensive examination of the influence of both a single measure and many measures on the spread of diseases. Utilizing numerous measures, as opposed to just one, can greatly boost the efficiency of disease transmission control.

Our strategy incorporates many mitigation techniques, including offering treatment to sick people, conducting screenings for latent cases, and giving education to those who are ignorant. We employ Pontryagin's Maximum Principle to deduce optimum control techniques. Based on the findings, it is evident that the most effective approach is to implement a comprehensive control strategy that incorporates all three measures.To successfully contain the pandemic, it is crucial to ensure that infected persons are informed about their infection and given suitable antiretroviral therapy (ART) treatments. This is because when people are receiving treatment and are aware of their status, the likelihood of transmitting the infection decreases significantly. The elimination of AIDS may be accomplished in the shortest amount of time feasible if all three preventative strategies are put into effect. Our research reveals that the four distinct control strategies examined in this study have the potential to successfully achieve the three 95\% reduction goals outlined by the UNAIDS in 2014, effectively putting an end to the AIDS epidemic by 2030. Additionally, the prompt implementation of control strategies is crucial for the timely eradication of AIDS. Specifically, we observe that the percentage of aware individuals who are infected plays a significant role in the spread of HIV. As more people become aware that they have HIV, it becomes easier to stop the virus from spreading. This is because ART therapy can significantly decrease HIV incidence, while there is a need for more efficient HIV testing \cite{van2013high}. Thus, the spread of AIDS can be reduced through advancements in medical technology, higher rates of testing, and enhanced education for individuals who are unaware of their infection.

Additionally, the model may be applied to explore a broad range of other diseases that share similar transmission patterns, such as COVID-19, Syphilis, and Hepatitis B. This is a significant advantage of the model. During the course of this research, neither the detection rate nor the influence that HIV co-infection with other illnesses has on the spread of HIV were factors that we took into consideration \cite{chang2013hiv}. One of our long-term goals is to investigate the influence that social behavior patterns and economic differences have on the transmission of diseases. Our goal is to investigate how the overall effectiveness of control techniques is affected by the diverse degrees of healthcare access that are available to people from different socioeconomic backgrounds. In addition, we intend to investigate the possibilities of emerging medical technologies, such as CRISPR gene editing \cite{wang2023crispr}, specifically with regard to their function in the management and prevention of diseases. Additionally, educational programs will be evaluated to see how beneficial they are in terms of raising awareness and lowering the stigma that is connected with certain diseases or conditions. In conclusion, it will be essential to collaborate across disciplines with policymakers in order to create and execute health initiatives that are not only scientifically sound but also culturally sensitive and economically feasible. A lot of different methods are needed to not only control current epidemics but also plan ahead for future health problems.

\section*{Acknowledgments}% and Funding Statement}
The author M Kamrujjaman acknowledged the University Grants Commission (UGC), and the  University of Dhaka, Bangladesh for supplementary support of this research.

\section*{Competing Interests}
The authors declare no conflict of interest. 

\section*{Ethical Approval}
N/A %o consent is required to publish this manuscript.

\section*{Data Availability}
The data was collected form the following link:\\
%\url{https://www.unaids.org/sites/default/files/media\_asset/data-book-2023\_en.pdf}\\
\url{https://www.unaids.org/en/resources/documents/2024/global-aids-update-2024}\\
timeline: 2001-2024 or 2006-2024\\
The data will be published as a .csv or Excel file once the manuscript is accepted.

\section*{Author Contributions}
NNKR: Conceptualization; Data curation; Formal analysis; Methodology; Software; Writing – original draft (Leading).\\
MK: Conceptualization; Investigation; Resources; Validation; Software;  Supervision; Writing – original draft, review \& editing. \\
AI: Data curation; Formal analysis; Software;  Validation; Writing – original draft.

%\newpage
\appendix
%\section{Auxiliary Results}\label{appA}
\section{Mathematical Analysis}\label{sec:mathematicalanalysis}
In order to get a better understanding of model (\ref{Model_2}), we first determine the basic reproduction number, the feasible area, and the stability of equilibrium.

\subsection{Basic Reproduction Number}
We will calculate the Basic Reproduction Number for model (\ref{Model_2}) and prove the existense of disease-free equilibrium point.
\noindent By setting the right side of the model (\ref{Model_2}) to zero, the disease-free equilibrium point can be determined as follows:
\[
W_0 = (S_0, E_0, I_{10}, I_{20}, T_0, A_0) = \left( \frac{\kappa}{\mu}, 0, 0, 0, 0,0 \right).
\]
The next generation matrix method is used to find the basic reproduction number of the model (\ref{Model_2}). Define $\mathcal{F}_i$ as the number of newly infected individuals in compartment $i$, and $\mathcal{V}_i$ as the transfer of individuals within each compartment labeled $i$, where $i$ corresponds to the compartments $S$, $E$, $I_1$, $I_2$, $T$, and $A$. We state, $x = (S, E, I_1, I_2, T, A) = (x_1, x_2, x_3, x_4, x_5, x_6)$. Model (\ref{Model_2}) can be rewritten as follows:
\begin{equation*}
	\frac{dx}{dt} = \mathcal{F} - \mathcal{V}
\end{equation*}
where,
\begin{equation*}
	\mathcal{F} = \begin{pmatrix}
		0 \\
		\alpha E S + \alpha I_2 S \\
		0 \\
		0 \\
		0 \\
		0 
	\end{pmatrix}, \quad
	\mathcal{V} = \begin{pmatrix}
		-\kappa + (\alpha \epsilon E + \alpha I_2)S +\mu \\
		k_1 E\\
		-p \beta E + k_1 I_1 - \gamma I_2 \\
		-(1 - p) \beta E + k_3 I_2 \\
		-\psi I_1 + k_4 T \\
		-\delta_1 I_1 - \delta_2 I_2 - \xi T + \mu_0 A + \mu A
	\end{pmatrix}.
\end{equation*}
Using the next generation matrix method to determine the basic reproduction number, we focus on the infected compartments labeled $x_i$, where $i = 2, 3, 4, 5$. At the disease-free equilibrium, denoted as $W_0$, it is observed that:
\begin{equation*}
	F = \left. \frac{\partial \mathcal{F}_i}{\partial x_i} \right|_{W_0} = 
	\begin{pmatrix}
		\alpha E S_0 & 0 & \alpha S_0 & 0 \\
		0 & 0 & 0 & 0 \\
		0 & 0 & 0 & 0 \\
		0 & 0 & 0 & 0 
	\end{pmatrix},
\end{equation*}
and
\begin{equation*}
	V = \left. \frac{\partial \mathcal{V}_i}{\partial x_i} \right|_{W_0} =
	\begin{pmatrix}
		k_1 & 0 & 0 & 0 \\
		-p \beta & k_2 & -\gamma & 0 \\
		-(1 - p) \beta & 0 & k_3 & 0 \\
		0 & -\psi & 0 & k_4 
	\end{pmatrix}.
\end{equation*}
where, $k_1= \beta + \mu > 0, k_2= \mu + \psi +\delta_1 >0, k_3= \gamma + \mu + \delta_2 >0, k_4=\mu + \xi > 0.$ Additionally, the matrices $F$ and $V$ satisfy the conditions (A1)-(A5) referenced in \cite{van2002reproduction}. By determining the spectral radius of the next generation matrix $FV^{-1}$, we can calculate the basic reproduction number, $\mathcal{R}_0$, in the following way:
\begin{equation*}
	\displaystyle \mathcal{R}_0 = \rho(FV^{-1}) = \frac{\alpha \varepsilon S_0}{k_1} + \frac{\beta S_0 \alpha (1 - p)}{k_1 k_3}
\end{equation*}
Here, $\displaystyle \frac{\alpha \varepsilon S_0}{k_1}$ displays the amount of people infected by latent individuals, $\displaystyle \frac{\beta S_0 \alpha (1 - p)}{k_1 k_3}$ depicts the number of persons who have been infected by those who are unaware that they are infected.

\subsection{Stability Analysis of the Disease-Free Equilibrium State}
Firstly, we will examine the stability or steadiness of the disease-free equilibrium within the local environment.
\begin{theorem}
	For model (\ref{Model_2}), the disease free equilibrium, $W_0$, is locally asymptotically stable when $\mathcal{R}_0 < 1$ in the feasible region $\Omega$.
\end{theorem}
\begin{proof}
	We generate the Jacobian matrix in the following manner, using model \eqref{Model_2} as our basis:
	\begin{align*}
		J &= \begin{pmatrix}
			-(\lambda(t) + \mu) & -\alpha \varepsilon S & 0 & -\alpha S & 0 & 0 \\
			\lambda(t) & \alpha \varepsilon S - k_1 & 0 & \alpha S & 0 & 0 \\
			0 & p \beta & -k_2 & \gamma & 0 & 0 \\
			0 & (1 - p)\beta & 0 & -k_3 & 0 & 0 \\
			0 & 0 & \psi & 0 & -k_4 & 0 \\
			0 & 0 & \delta_1 & \delta_2 & \xi & -\mu_0 - \mu
		\end{pmatrix} \\
		&= \begin{pmatrix}
			(J_1)_{4 \times 4} & 0 \\
			(J_3)_{2 \times 4} & (J_4)_{2 \times 2}
		\end{pmatrix}
	\end{align*}
	where, $J_1$ is a $4 \times 4$ matrix, $J_3$ is a $2 \times 4$ matrix, $J_3$ is a $2 \times 2$ matrix and,
	\[
	\begin{aligned}
		k_1 &= \mu + \beta > 0, \\
		k_2 &= \mu + \psi + \delta_1 > 0, \\
		k_3 &= \mu + \gamma + \delta_2 > 0, \\
		k_4 &= \mu + \xi > 0.
	\end{aligned}
	\]
	The eigenvalues, represented by \( r \), are the solutions to the equation \( |rI - J(W_0)| = 0 \). These solutions are also the roots of the equations \( |rI - J_1(W_0)| = 0 \) and \( |rI - J_4(W_0)| = 0 \).\\
	The Jacobian matrix \(J_1(W_0)\) is:
	\[
	J_1(W_0) = \begin{pmatrix}
		-\mu & -\alpha \varepsilon S_0 & 0 & -\alpha S_0 \\
		0 & \alpha \varepsilon S_0 - k_1 & 0 & \alpha S_0 \\
		0 & p \beta & -k_2 & \gamma \\
		0 & (1 - p)\beta & 0 & -k_3
	\end{pmatrix}.
	\]
	\begin{align*}
		\left| rI - J_1(W_0) \right| &= 
		\begin{vmatrix}
			r + \mu & \alpha \varepsilon S_0 & 0 & \alpha S_0\\
			0 & r - (\alpha \varepsilon S_0 - k_1) & 0 & -\alpha S_0\\
			0 & -p \beta & r + k_2 & -\gamma \\
			0 & -(1-p)\beta & 0 & r + k_3 \\
		\end{vmatrix} \\
		&= (r + k_2)(r + \mu) \\
		&\quad \times \left[ r^2 + (k_1 + k_3 - \alpha \varepsilon S_0)r - k_3(\alpha \varepsilon S_0 - k_1) - \alpha S_0\beta(1 - p) \right] \\ &= 0.
	\end{align*}
	Clearly, the first and second roots of $|rI - J_1(W_0)| = 0$ are $r_1= -k_2 < 0$ and $r_2= -\mu < 0.$ 
	
	Next,we will examine the roots of the following equation:
	\begin{equation}
		r^2 + (k_1 + k_3 - \alpha \varepsilon S_0)r - k_3(\alpha \varepsilon S_0 - k_1) - \alpha S_0 \beta(1 - p)
		= r^2 + m_1r + m_0 = 0 %\tag{3}
		\label{eq:3}
	\end{equation}
	Here,
	\begin{align*}
		m_1 &= k_1 + k_3 - \alpha \varepsilon S_0 \\
		m_0 &= k_3(k_1 - \alpha \varepsilon S_0) - \beta S_0\alpha(1 - p) \\
		&= k_3k_1 \left[ 1 - \left( \frac{\alpha \varepsilon S_0}{k_1} + \frac{\beta\alpha S_0(1 - p)}{k_1k_3} \right) \right] \\
		&= k_3k_1(1 - \mathcal{R}_0).
	\end{align*}
	When \( R_0 < 1 \), it is easy to prove that \( m_1 > 0 \) and \( m_0 > 0 \). Here, \( R_0 = \frac{\alpha \varepsilon S_0}{k_1} + \frac{\beta S_0\alpha(1-p)}{k_1k_3} \). Based on the Routh–Hurwitz criterion \cite{chitnis2006bifurcation}, it can be observed that the roots of equation \eqref{eq:3} possess real portions that are negative. 
	
	For \( J_4(W_0) \),
	\[
	\left| rI - J_4(W_0) \right| = (r + k_4)(r + \mu_0 + \mu) = 0.
	\]
	It is obvious that, the eigenvalues of \( J_4(W_0) \) are \( r_5 = -k_4 < 0 \) and \( r_6 = -\mu_0 - \mu < 0 \).
	
	To summarize, the disease-free equilibrium, denoted as \( W_0 \), is locally asymptotically stable inside the region \( \Omega \) when the basic reproduction number \( \mathcal{R}_0 \) is less than 1.
\end{proof}

\noindent We then study the disease-free equilibrium state's global asymptotic stability.
\begin{theorem}
	For model (\ref{Model_2}), when $\mathcal{R}_0 < 1$, the disease-free equilibrium point, $W_0$ is globally asymptotically stable. 
\end{theorem}
\begin{proof}
	According to Theorem 2.1 in \cite{shuai2013global}, we create a Lyapunov function in the following manner:
	\[
	L = (\varepsilon k_3 + \beta(1 - p))E + k_1I_2.
	\]
	It is evident that, $L \geq 0$.\\
	After that, we will do the derivation of $L$.
	\begin{align*}
		\frac{dL}{dt} &= (\varepsilon k_3 + \beta(1 - p))\frac{dE}{dt} + k_1\frac{dI_2}{dt} \\
		&= (\varepsilon k_3 + \beta(1 - p))(\alpha \varepsilon SE + \alpha SI_2 - k_1E) + \beta(1 - p)k_1E - k_1k_3I_2 \\
		&= (\varepsilon k_3 + \beta(1 - p))(\alpha \varepsilon SE + \alpha SI_2) - k_1k_3\varepsilon E - k_1k_3I_2 \\
		&\leq (\varepsilon k_3 + \beta(1 - p))\left(\frac{\kappa\alpha\varepsilon E}{\mu} + \frac{\kappa\alpha I_2}{\mu}\right) - k_1k_3\varepsilon E - k_1k_3I_2 \\
		&= (\varepsilon E + I_2)k_1k_3(\mathcal{R}_0 - 1).
	\end{align*}
	When $\mathcal{R}_0 < 1$, $\frac{dL}{dt} < 0$. Therefore, the largest invariant set contained in
	\[
	\left\{(S,E,I_1,I_2,T,A) \in \Omega : \frac{dL}{dt} = 0\right\}
	\]
	is $\{W_0\}$. Based on the LaSalle's invariance principle \cite{barkana2017can}, the disease-free equilibrium $W_0$ is globally asymptotically stable in $\Omega$ if the basic reproduction number $\mathcal{R}_0$ is less than $1$.
\end{proof}

\subsection{Analysis of the Endemic Equilibrium State}
Here, we will firstly find the existence of an endemic equilibrium point $W^*$ of model \eqref{Model_2}.
Let, $W^*=(S^*,E^*,I_\text{1}^\text{*},I_\text{2}^\text{*},T^*,A^*).$\\
To make the calculation easier, we have considered the rate at which aware individuals move to AIDS class $(\delta_1)=$ the rate at which unaware individuals move to AIDS class $(\delta_2)=\delta.$
From the second to sixth equation in model (\ref{Model_2}), we have,
\begin{align*}
	E^* &= \displaystyle \frac{k_3}{(1 - p)\beta} I_2^*, \\
	I_1^* &= \displaystyle \left[ \frac{pk_3 + (1 - p)\gamma}{(1 - p)k_2} \right] I_2^*, \\
	T^* &= \displaystyle \frac{\psi(pk_3 + (1 - p)\gamma)}{(1 - p)k_2k_4} I_2^*, \\
	A^* &= \displaystyle \frac{\delta I_2^*}{\mu + \mu_0} \left[ \frac{pk_3 + \gamma(1 - p)}{k_2(1 - p)} + 1 \right] + \frac{\xi \psi I_2^*}{(\mu + \mu_0)k_4} \left[ \frac{pk_3 + \gamma(1 - p)}{k_2(1 - p)} \right], \\
	S^* &= \displaystyle \frac{\kappa}{\mu + \frac{k_1k_3I_2^*}{\beta(1-p)}{\frac{\mathcal{R}_0}{S_0}}} , \\
	I_2^* &= \displaystyle \frac{(\mathcal{R}_0 - 1) \mu \beta (1 - p)S_0}{k_1k_3\mathcal{R}_0}.
\end{align*}
Hence, we obtain that $I_2^* > 0$ if and only if $\mathcal{R}_0 > 1$. Depending on the analysis we have done earlier, model (\ref{Model_2}) has an endemic equilibrium, $W^*$ when $\mathcal{R}_0 > 1$. The equilibrium point we found is also unique.

Subsequently, now we will analyse the global stability of the endemic equilibrium, $W^*$.
\begin{theorem}
	For model (\ref{Model_2}), if $\mathcal{R}_0>1,$ the endemic equilibrium $W^*$ is globally asymptotically stable.
\end{theorem}
\begin{proof}
	The Lyapunov function  $\tilde{U}$ is described as follows \cite{li2011global}:   
	Let, \( B, C, D,\) and \(G \) be the constants which are not negative. Consider the Lyapunov function candidate \( \tilde{U} \) defined as follows:
	\begin{align*}
		\widetilde{U} &= \left( S - S^* - S^*\ln \frac{S}{S^*} \right) + B \left( E - E^* - E^*\ln \frac{E}{E^*} \right) \\
		&\quad+ C \left( I_1 - I_1^* - I_1^*\ln \frac{I_1}{I_1^*} \right) + D \left( I_2 - I_2^* - I_2^*\ln \frac{I_2}{I_2^*} \right) \\
		&\quad+ G \left( T - T^* - T^*\ln \frac{T}{T^*} \right).
	\end{align*}
	We postulate that \( \widetilde{U} \) is positive definite. Also, the first derivative of \( \widetilde{U} \) is also positive definite given all parameters of the model are positive and bounded away from zero. Specifically, \( \widetilde{U}(W^*) = 0 \) at the equilibrium point \( W^* \). Moreover, for \( S, E, I_1, I_2, T, A > 0 \) and \( S \neq S^*, E \neq E^*, I_1 \neq I_1^*, I_2 \neq I_2^*, T \neq T^*, A \neq A^* \), we define:
%	\begin{align*}
%		\widetilde{U}(S, E, I_1, I_2, T) &= h(S) + Bh(E) + Ch(I_1) + Dh(I_2) + Gh(T), \\
%		\text{where,} \quad h(S) &= S - S^* - S^*\ln \frac{S}{S^*}, \quad h(E) = E - E^* - E^*\ln \frac{E}{E^*}, \\
%		h(I_1) &= I_1 - I_1^* - I_1^*\ln \frac{I_1}{I_1^*}, \quad h(I_2) = I_2 - I_2^* - I_2^*\ln \frac{I_2}{I_2^*}, \\
%		h(T) &= T - T^* - T^*\ln \frac{T}{T^*}.
%	\end{align*}
	\begin{align*}
			\widetilde{U}(S, E, I_1, I_2, T) = h(S) + Bh(E) + Ch(I_1) + Dh(I_2) + Gh(T), 
	\end{align*}
where, 
$$
h(S) = S - S^* - S^*\ln \frac{S}{S^*}, \;\; h(E) = E - E^* - E^*\ln \frac{E}{E^*}, \;\;
h(I_1) = I_1 - I_1^* - I_1^*\ln \frac{I_1}{I_1^*}, 
$$
$$
h(I_2) = I_2 - I_2^* - I_2^*\ln \frac{I_2}{I_2^*}, \;\;
h(T) = T - T^* - T^*\ln \frac{T}{T^*}.
$$

\noindent It follows that \( h(S)\), \( h(E)\), \( h(I_1)\), \( h(I_2)\), \( h(T) \geq 0 \) due to the properties of the logarithmic function and the positivity of \( S, E, I_1, I_2, T \).
	According to model (\ref{Model_2}),
	\begin{align*}
		\widetilde{U}^{'} &= \left( 1 - \frac{S^*}{S} \right)S' + B \left( 1 - \frac{E^*}{E} \right)E' + C \left( 1 - \frac{I_1^*}{I_1} \right)I_1' + D \left( 1 - \frac{I_2^*}{I_2} \right)I_2' + G \left( 1 - \frac{T^*}{T} \right)T' \\
		&= \left( 1 - \frac{S^*}{S} \right)(\kappa - \alpha (\varepsilon E + I_2)S - \mu S) + B \left( 1 - \frac{E^*}{E} \right)(\alpha (\varepsilon E + I_2)S - (\beta + \mu) E) \\
		&\quad+ C \left( 1 - \frac{I_1^*}{I_1} \right)(p\beta E + \gamma I_2 - (\mu + \psi + \delta) I_1)+ D \left( 1 - \frac{I_2^*}{I_2} \right)((1 - p)\beta E - (\delta + \gamma + \mu) I_2) \\
		&\quad+ G \left( 1 - \frac{T^*}{T} \right)(\psi I_1 - (\mu + \xi) T)\\
		&= \left( 1 - \frac{S^*}{S} \right) \left[ \alpha \varepsilon E^* S^* + \alpha I_2^* S^* + \mu S^* - (\alpha \varepsilon ES + \alpha I_2 S + \mu S) \right]  + B \left( 1 - \frac{E^*}{E} \right)  \\
		&\quad \left[ \alpha \varepsilon ES + \alpha I_2 S - {\frac{\alpha \varepsilon E^* S^* + \alpha I_2^* S^*}{E^*}} E \right] + C \left( 1 - \frac{I_1^*}{I_1} \right) \left[ p \beta E + \gamma I_2 - {\frac{p \beta E^* + \gamma I_2^*}{I_1^*}} I_1 \right]\\
		&\quad + D \left( 1 - \frac{I_2^*}{I_2} \right) \left[ (1 - p) \beta E - {\frac{(1 - p) \alpha E^*}{I_2^*}} I_2 \right] + G \left( 1 - \frac{T^*}{T} \right) \left[ \psi I_1 - {\frac{\psi I_1^*}{T^*}} T \right]
	\end{align*}
	Let,
	\begin{align*}
		\frac{S}{S^*} &= x, & \frac{E}{E^*} &= y,& \frac{I_1}{I_1^*} &= z, & \frac{I_2}{I_2^*} &= m, & \frac{T}{T^*} &= n.
	\end{align*}
	Then,
	\begin{equation*}
		\begin{aligned}
			\widetilde{U}' &= -\mu S^{*}\frac{(1-x)^2}{x} +\alpha \varepsilon E^{*} S^{*}\left( 1 - \frac{1}{x}\right)(1 - xy) + \alpha I_2^{*} S^{*}\left(1 - \frac{1}{x}\right)(1 - xm) \\
			&\quad + \alpha B \varepsilon S^{*} E^{*}\left(1 - \frac{1}{y}\right)(xy - y) + \alpha B I_2^{*} S^{*}\left(1 - \frac{1}{y}\right)(xm - y) \\
			&\quad + Cp \beta E^{*}\left(1 - \frac{1}{z}\right)(y - z) + C \gamma I_2^{*}\left(1 - \frac{1}{z}\right)(m - z) \\
			&\quad + D \beta (1 - p) E^{*}\left(1 - \frac{1}{m}\right)(y - m) + G \psi I_1^{*}\left(1 - \frac{1}{n}\right)(z - n) \\
			&= -\mu S^{*} \frac{(1 - x)^2}{x} + \left[ \alpha \varepsilon S^{*} E^{*} + \alpha I_2^{*} S^{*} + \alpha \varepsilon B E^{*} S^{*} + \alpha B I_2^{*} S^{*} + Cpa E^{*} \right. \\
			&\quad + C \gamma I_2^{*} + D \beta (1 - p) E^{*} + G \psi I_1^{*}] + xy \left( -\alpha \varepsilon S^{*} E^{*} + \alpha \varepsilon BE^{*} S^{*} \right) \\
			&\quad + \frac{1}{x} \left( -\alpha \varepsilon S^{*} E^{*} - \alpha I_2^{*} S^{*} \right) \\
			&\quad + y \left( \alpha \varepsilon S^{*} E^{*} -B \alpha \varepsilon  E^{*} S^{*} - \alpha B I_2^{*} S^{*} + Cp \beta E^{*} + D \beta (1 - p) E^{*} \right) \\
			&\quad + m \left( \alpha I_2^{*} S^{*} + C \gamma I_2^{*} - D \beta (1 - p) E^{*} \right) + xm \left( -\alpha I_2^{*} S^{*} + \alpha B I_2^{*} S^{*} \right) \\
			&\quad + z \left( -Cp \beta E^{*} - C \gamma I_2^{*} + G \psi I_1^{*} \right) - x \alpha \varepsilon B E^{*} S^{*} - \frac{xm}{y} \alpha B I_2^{*} S^{*} - \frac{y}{z} C \beta p E^{*}  \\
			&\quad - \frac{m}{z} C \gamma I_2^{*} - \frac{y} {m} D \beta (1 - p) E^{*} - n G \psi I_1^{*} - \frac{z}{n} G \psi I_1^{*}.
		\end{aligned}
	\end{equation*}
	Based on \cite{li2011global}, $x,\; y,\; m,\;$ and $z$ are all assumed to have coefficients of zero. By doing so, we will get the following equations:
	\begin{equation}
		\begin{cases}
			-\alpha \varepsilon S^* E^* + B \alpha \varepsilon S^* E^* = 0, \\
			\alpha \varepsilon S^* E^* - B \alpha \varepsilon S^* E^* - B \alpha I_2^* S^* + Cp \beta E^* + D \beta (1 - p) E^* = 0, \\
			\alpha I_2^* S^* + C \gamma I_2^* - D \beta (1 - p) E^* = 0, \\
			-Cp \beta E^* - C \gamma I_2^* + G \psi I_1^* = 0. %\tag{4} 
			\label{eq4}
		\end{cases}
	\end{equation}
	Solving equation (\ref{eq4}), we obtain,
	\begin{align*}
		B &= 1, \quad C = 0, \quad D = \frac{\alpha I_2^* S^*}{\beta(1 - p)E^*}, \quad G = 0,
	\end{align*}
	in which, \(p \neq 1\).\\
	Given the values of \( B \), \( C \), \( D \), and \( G \), the Lyapunov function \( \widetilde{U} \) may be expressed as:
	%\begin{equation*}
	%\widetilde{U} = \begin{pmatrix}
		%S - S^* - S^* \big \ln{\frac{S}{S^*}}
		%\end{pmatrix} + \begin{pmatrix}
		%E - E^* - E^* \big \ln{\frac{E}{E^*}}\replace{\replace{text}{replacement}}{replacement}
		%\end{pmatrix} + \frac{\alpha I_2^* S^*}{\beta (1 - p) E^*} 
		%\begin{pmatrix}
		%I_2 - I_2^* - I_2^* \big \ln{\frac{I_2}{I_2^*}}
		%\end{pmatrix}.
		%\end{equation*}
		\begin{equation*}
			\begin{aligned}
				\widetilde{U}' &= -\mu S^{*} \frac{(1 - x)^2}{x} + \left[ \alpha \varepsilon S^{*} E^{*} + \alpha I_2^{*} S^{*} + \alpha \varepsilon B E^{*} S^{*} + \alpha B I_2^{*} S^{*} + Cpa E^{*} \right. \\
				&\quad + C \gamma I_2^{*} + D \beta (1 - p) E^{*} + G \psi I_1^{*}] + \frac{1}{x} \left( -\alpha \varepsilon S^{*} E^{*} - \alpha I_2^{*} S^{*} \right) + Cp \beta E^* \\
				&\quad + D \beta (1-p) E^* -Bx \alpha \varepsilon S^* E^* - \frac{xm}{y} \alpha B I_2^{*} S^{*} - \frac{y}{z} C \beta p E^{*}  \\
				&\quad - \frac{m}{z} C \gamma I_2^{*} - \frac{y} {m} D \beta (1 - p) E^{*} - n G \psi I_1^{*} - \frac{z}{n} G \psi I_1^{*}\\
				&= -\mu S^{*} \frac{(1 - x)^2}{x} + \Bigg[2 \alpha \varepsilon S^* E^* + 2 \alpha I_2^* S^* + \frac{\alpha I_2^* S^*}{\beta (1-p) E^*} \beta (1-p) E^* \bigg]\\
				&\quad +\frac{1}{x} (-\alpha \varepsilon S^{*} E^{*}- \alpha I_2^* S^*)- x \alpha \varepsilon S^* E^* - \frac{xm}{y} \alpha I_2^* S^* \\
				&\quad - \frac{y}{m} \frac{\alpha I_2^* S^*}{\beta (1-p) E^*} \beta (1-p) E^* \\
				&= -\mu S^{*} \frac{(1 - x)^2}{x} + \alpha \varepsilon S^* E^* \bigg(2- \frac{1}{x}-x\bigg)+ \alpha I_2^* S^* \bigg(3- \frac{1}{x}- \frac{xm}{y}- \frac{y}{m} \bigg).
			\end{aligned}
		\end{equation*}
		We know that, the arithmetic mean is always bigger than or similar to the geometric mean, $\displaystyle 2- \frac{1}{x}-x \leq 0$ and $\displaystyle 3-\frac{1}{x}-\frac{xm}{y}-\frac{y}{m} \leq 0$ for $x,y,m>0.$ Here, $\displaystyle 2- \frac{1}{x}-x = 0$ iff $x=1$, and $\displaystyle 3-\frac{1}{x}-\frac{xm}{y}-\frac{y}{m} = 0$ iff $x=1$ and $y=m$. Therefore, $\widetilde{U}' \leq 0$ for $x>0,\; y>0,\; m >0,$ and $\widetilde{U}'=0$ iff $x=1$ and $y=m$. Using LaSalle's Invariance Principle \cite{barkana2017can}, it is reasonably easy to establish that the endemic equilibrium point \( W^* \) is globally asymptotically stable inside \( \Omega \). So, to sum up, when $R_0 >1$, we are able to assert that the endemic equilibrium, denoted by $W^*$, is globally asymptotically stable.
	\end{proof}
	\noindent Given the global asymptotic stability of the endemic equilibrium $W^*$, it follows that $W^*$ is also locally asymptotically stable for $R_0>1.$

\end{document}